\documentclass[1p,number,a4paper]{elsarticle}
\usepackage{anyfontsize}
\usepackage{stix}
\usepackage{lineno}
\usepackage{amsfonts}
\usepackage{amsmath}
\usepackage{amsthm}
\usepackage{mathrsfs}
\usepackage{enumerate}
\usepackage{graphics}
\usepackage{mathtools}
\usepackage{microtype}
\usepackage[all,defaultlines=4]{nowidow}
\usepackage{xcolor}
\usepackage[czech, english]{babel}
\usepackage{esvect}
\usepackage{complexity}
\usepackage{comment}
\usepackage[many]{tcolorbox}
\usepackage{thm-restate}
\usepackage[all]{xy}
\usepackage{hyperref}
\usepackage[nameinlink]{cleveref}
\definecolor{dark-blue}{rgb}{0.05,0.25,0.85}
\hypersetup{colorlinks=true,linkcolor=dark-blue,citecolor=dark-blue}

\theoremstyle{remark}

\newcommand{\cqed}{\ensuremath{\lhd}}
\makeatletter

\tcolorboxenvironment{theorem}{
colframe=cyan,interior hidden, breakable,before skip=10pt,after skip=10pt,	outer arc=0pt,
	arc=0pt,
	colback=white,
	rightrule=0pt,
	toprule=0pt,
	top=0pt,
	right=0pt,
	bottom=0pt,
	bottomrule=0pt,
 }
\tcolorboxenvironment{corollary}{
	colframe=blue,interior hidden, breakable,before skip=10pt,after skip=10pt,	outer arc=0pt,
	arc=0pt,
	colback=white,
	rightrule=0pt,
	toprule=0pt,
	top=0pt,
	right=0pt,
	bottom=0pt,
	bottomrule=0pt,
}

\makeatother

\DeclareTColorBox{answer}{O{orange}O{0cm}}{
	breakable,
	outer arc=0pt,
	arc=0pt,
	colback=white,
	rightrule=0pt,
	toprule=0pt,
	top=0pt,
	right=0pt,
	bottom=0pt,
	bottomrule=0pt,
	colframe=#1,
	enlarge left by=#2,
	width=\linewidth-#2,
}
\DeclareTColorBox{solution}{O{orange}O{0cm}}{
	breakable,
	outer arc=0pt,
	arc=0pt,
	colback=white,
	rightrule=0pt,
	toprule=0pt,
	top=0pt,
	right=0pt,
	bottom=0pt,
	bottomrule=0pt,
	colframe=#1,
	enlarge left by=#2,
	width=\linewidth-#2,
}

\newcommand{\Oof}{\mathcal{O}}
\newcommand{\CCC}{\mathscr{C}}
\newcommand{\Cc}{\mathscr{C}}
\newcommand{\Dd}{\mathscr{D}}

\newcommand{\TTT}{\mathscr{T}}

\newcommand{\Ff}{\mathcal{F}}
\newcommand{\Ll}{\mathcal{L}}

\newcommand{\Pp}{\mathcal{P}}

\newcommand{\Tt}{\mathcal{T}}
\newcommand{\Uu}{\mathcal{U}}

\newcommand{\Xx}{\mathcal{X}}
\newcommand{\minor}{\preccurlyeq}
\newcommand{\N}{\mathbb{N}}

\newcommand{\dist}{\mathrm{dist}}

\newcommand{\strA}{\mathbb{A}}
\newcommand{\strB}{\mathbb{B}}

\renewcommand{\phi}{\varphi}

\renewcommand{\FO}{\mathrm{FO}}

\newcommand{\MSO}{\mathrm{MSO}}

\newcommand{\bag}{\mathrm{Bag}}
\newcommand{\adh}{\mathrm{Adh}}
\newcommand{\Torso}{\mathrm{Torso}}
\newcommand{\mrg}{\mathrm{Margin}}

\renewcommand{\epsilon}{\varepsilon}

\newcommand{\setof}[2]{\left\{#1 \,\mid\, #2 \right\}}
\renewcommand{\dim}{\mathrm{dim}}

\usepackage[all]{xy}

\newcommand{\Ww}{\mathcal{W}}

\DeclareMathOperator{\td}{td}
\DeclareMathOperator{\rtd}{rtd}

\DeclareMathOperator{\col}{col}
\DeclareMathOperator{\wcol}{wcol}
\DeclareMathOperator{\gcol}{gcol}
\DeclareMathOperator{\adm}{adm}
\DeclareMathOperator{\tw}{tw}
\DeclareMathOperator{\fw}{fw}

\DeclareMathOperator{\WReach}{WReach}
\DeclareMathOperator{\SReach}{SReach}
\DeclareMathOperator{\GReach}{GReach}

\newtheorem{theorem}{Theorem}[section]

\newtheorem{corollary}[theorem]{Corollary}
\newtheorem{definition}[theorem]{Definition}
\newtheorem{lemma}[theorem]{Lemma}

\newtheorem{observation}[theorem]{Observation}
\newtheorem{problem}[theorem]{Problem}

\newtheorem{example}[theorem]{Example} 
\crefname{ext_theorem}{Theorem}{Theorems}
\crefname{corollary}{Corollary}{Corollaries}
\crefname{lemma}{Lemma}{Lemmas}
\crefname{problem}{Problem}{Problems}
\crefname{section}{Section}{Sections}

\journal{arXiv}
\begin{document}
\begin{frontmatter}
\title{On the generalized coloring numbers\tnoteref{ERC}}
\author{Sebastian Siebertz}\address{University of Bremen, Bremen, Germany}\ead{siebertz@uni-bremen.de}
\begin{abstract}
  The \emph{coloring number} $\col(G)$ of a graph $G$, which is equal
  to the \emph{degeneracy} of $G$ plus one, provides a very useful
  measure for the uniform sparsity of~$G$.  The coloring number is
  generalized by three series of measures, the
  \emph{generalized coloring numbers}. These are the \emph{$r$-admissi\-bility}~$\adm_r(G)$, the \emph{strong $r$-coloring
    number}~$\col_r(G)$ and the \emph{weak $r$-coloring number}
  $\wcol_r(G)$, where~$r$ is an integer parameter. The generalized coloring numbers measure the edge
  density of bounded-depth minors and thereby provide an even more
  uniform measure of sparsity of graphs. They have found many applications in graph theory and in particular play a key role in the theory of bounded expansion
  and nowhere dense graph classes introduced by Ne\v{s}et\v{r}il and
  Ossona de Mendez. We overview combinatorial and algorithmic
  applications of the generalized coloring numbers, emphasizing new
  developments in this area. We also present a simple proof for the
  existence of uniform orders and improve known bounds, e.g., for the
  weak coloring numbers on graphs with excluded topological minors.
\end{abstract}
\end{frontmatter}

\tableofcontents

\section{Introduction}

The notions of \emph{bounded expansion} and \emph{nowhere dense graph classes} 
were introduced by Ne\v{s}et\v{r}il and Ossona de Mendez as models 
for uniformly sparse graph classes~\cite{Taxi_stoc06,nevsetvril2011nowhere}. 
These are abstract and robust notions of sparsity, 
and many of the commonly studied sparse graph classes have bounded
expansion or at least are nowhere dense. Examples include the class of all planar
graphs, classes excluding a minor, classes with bounded degree, classes excluding a topological minor, classes with locally bounded treewidth, and many more. Bounded expansion and nowhere dense graph classes were originally defined by \emph{restricting the edge densities of 
bounded-depth minors} that may appear in these classes\footnote{A graph $H$ is a \emph{depth-$r$} minor of $G$ if there is 
a map $M$ that assigns to every vertex $v\in V(H)$ a
connected subgraph $M(v) \subseteq G$ of radius at most $r$ and to every edge
$e\in E(H)$ an edge  $M(e)\in E(G)$ such that
$M(u)$ and $M(v)$ are vertex disjoint for distinct vertices $u,v\in V(H)$, and
if $e = uv \in E(H)$, then $M(e) = u'v'\in E(G)$ for 
  vertices $u'\in M(u)$ and $v'\in M(v)$, see Figure~2.
A class of graphs has \emph{bounded expansion} (is \emph{nowhere dense}) if the depth-$r$ minors of its members have edge density (clique number) bounded by $f(r)$ for some function $f$.}. It turns out, however,
that they can equivalently be characterized in many different ways, 
each characterization highlighting specific combinatorial properties and 
leading to specific algorithmic tools. 

One particular useful characterization of bounded expansion and nowhere
dense graph classes is in terms of the so-called \emph{generalized 
coloring numbers}. 
Before introducing these, let us consider the classical 
coloring number. 
The \emph{coloring number} $\col(G)$ of a graph 
$G$, first considered by Erd{\H{o}}s and Hajnal, is the minimum integer~$k$ such that there exists a linear 
order $\pi$ of the vertices of $G$ such that every vertex $v$ has 
back-degree at most~$k-1$, i.e., at most~$k-1$ neighbors $u$ with 
$u<_\pi v$~\cite{erdHos1966chromatic}. The coloring number of $G$ is one more than the 
\emph{degeneracy} of $G$, which is 
the minimum integer~$\ell$ such that every subgraph $H\subseteq G$
has a vertex of degree at most $\ell$, as observed by Lick and White~\cite{Lick_White_1970}. To see this, assume first that $G$ has coloring number~$k$ 
and let $H\subseteq G$ be a subgraph of $G$. Let~$\pi$ 
be an order of $V(G)$ witnessing that $\col(G)\leq k$, and let~$v$ be 
the largest vertex of~$H$ with respect to $\pi$. All neighbors 
of $v$ in $H$ are hence smaller than~$v$,
and, by assumption, there are at most $k-1$ of them. 
Hence, as $H$ was chosen arbitrarily, $G$ is~$k-1$~degenerate. 
Conversely, assume that
the degeneracy of~$G$ is equal to $\ell$. Observe that also every
subgraph $H\subseteq G$ has degeneracy at most $\ell$ and that the
statement is trivial for single vertex graphs. Now, if~$G$ has
more than one vertex, 
we can find an order witnessing that $\col(G)\leq \ell+1$ by taking a 
vertex $v$ of $G$ of degree at most $\ell$, finding inductively an order
with back-degree at most $\ell$ for the graph $H=G-v$, 
and inserting $v$ into the order as the largest vertex, see Figure~1. 

    \begin{figure}[ht]
        \begin{center}
        \includegraphics[width=0.7\textwidth]{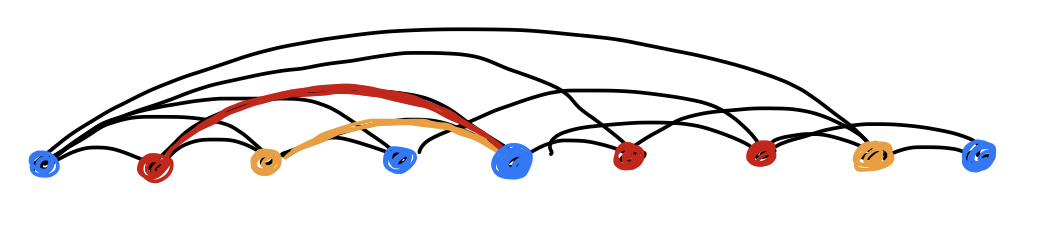}
        \label{fig:coloringnumber}
        \caption{Graphs with coloring number $k$ can have vertices with arbitrarily large degree, however, the back-degree of every vertex is bounded by $k-1$. A simple greedy procedure leads to a proper vertex coloring. }\vspace{-5mm}   
    \end{center}
    \end{figure}

An $\ell$-degenerate graph $G$ is uniformly sparse in the
sense that every $n$-vertex subgraph $H\subseteq G$ 
has at most $\ell\cdot n$
edges.
Hence, despite the name, which suggests a strong connection with
the chromatic number $\chi$, the coloring number
is a parameter that measures uniform sparsity of graphs. 
To its rescue, the coloring number does have a connection to the chromatic number, as a simple greedy coloring algorithm with the degeneracy order provides a bound for the chromatic number: \mbox{$\chi(G)\leq \col(G)$}. To see this, 
let~$\pi$ be an order of $V(G)$ witnessing that $\col(G)\leq k$, 
and let~$v$ be 
the largest vertex of~$G$ with respect to $\pi$. By induction,
we obtain a proper vertex coloring with at most $k$ colors of 
$G-v$. As $v$ has at most $k-1$ neighbors, there is a 
color that is not used by any of its neighbors and we can complete
the coloring of $G-v$ to a coloring of $G$ with at most $k$ colors. 
If $G$ is chordal, this bound is 
tight, while bipartite graphs are $2$-colorable and 
can have arbitrarily large coloring number.

Several generalizations of the coloring number have been studied 
in the literature, including the \emph{arrangeability} by Chen and Schelp~\cite{chen1993graphs}  (connected to Ramsey problems), the \emph{game chromatic number} by Faigle et al.~\cite{faigle1993game}, the \emph{game coloring number} by Zhu~\cite{ZHU1999245}, the \emph{rank} by Kierstead~\cite{kierstead2000simple} (related to the game chromatic number), 
and the \emph{admissibility} of a graph by Kierstead and Trotter~\cite{kierstead1994planar} (used to bound both the arrangeability and the game chromatic number). 
These parameters are yet generalized by the following three 
series of parameters, called the \emph{generalized coloring numbers}.  
For a positive integer $r$, the \emph{strong $r$-coloring number} 
$\col_r(G)$ and \emph{weak $r$-coloring number} $\wcol_r(G)$ were 
introduced by Kierstead and Yang~\cite{kierstead2003orderings},
and the \emph{$r$-admissibility} $\adm_r(G)$ 
was introduced by Dvo\v{r}\'ak~\cite{dvovrak2013constant}. 
All of these generalizations rely on linear orders of the vertices and certain back-connectivity properties (formal definition will be given 
below).

Intuitively, the generalized coloring numbers capture
structural properties not only of subgraphs of a graph, as degeneracy
does, but also 
of its \emph{bounded-depth minors} (see Figure~2). Bounded-depth minors
are the key concept in the theory of bounded  
expansion and nowhere dense graph classes, introduced by
Ne\v{s}et\v{r}il and Ossona de Mendez. The (non trivial) connection between the 
generalized coloring numbers and classes with bounded expansion was made by 
Zhu~\cite{zhu2009colouring}, and this connection naturally extends to nowhere dense classes as observed by Ne\v{s}et\v{r}il and Ossona de Mendez~\cite{nevsetvril2011nowhere}.

\begin{figure}[ht]
    \begin{center}
    \includegraphics[width=0.5\textwidth]{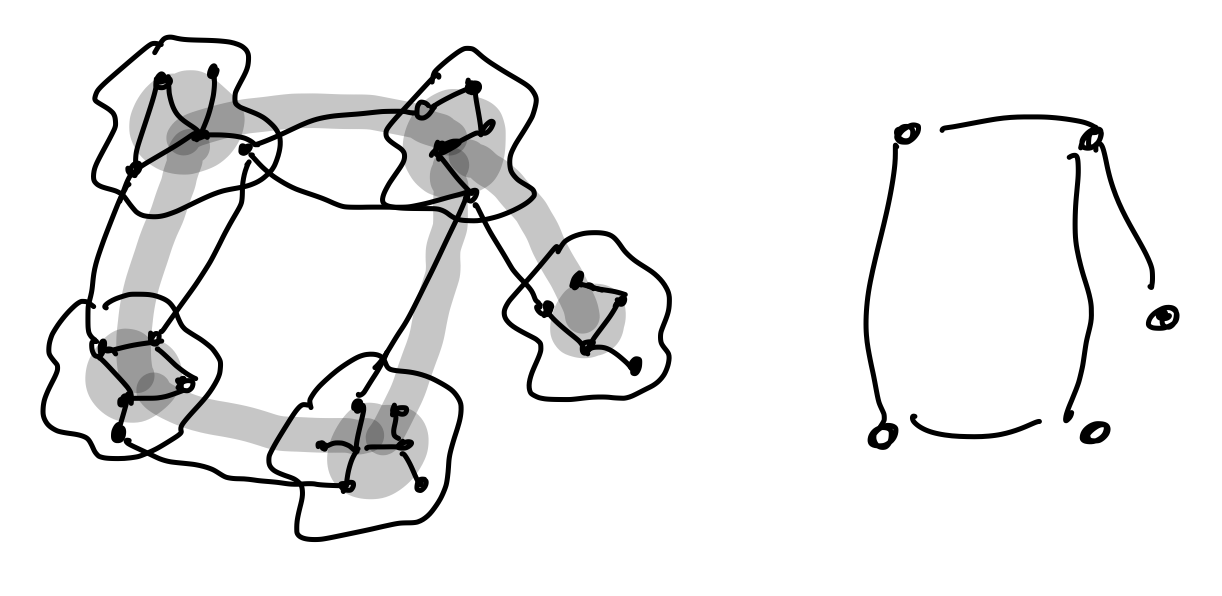}
    \label{fig:depthrminor}
    \caption{Graph $H$ (on the right) is a depth-$2$ minor of $G$ (on the left).}\vspace{-5mm}
\end{center}
\end{figure}

On the other hand, the strong $r$-coloring numbers can 
be seen as gradations between the coloring number $\col(G)$ 
and the \emph{treewidth} $\tw(G)$ of
$G$, and the weak $r$-coloring numbers can be seen
as gradations between $\col(G)$ and the \emph{treedepth} $\td(G)$
of $G$. Just as treewidth captures the global connectivity 
of a graph $G$, the generalized coloring numbers 
capture its local connectivity. 
By now, the generalized coloring numbers have found many
applications in graph theory, and play
a key role in the theory of bounded expansion and nowhere
dense classes of graphs. 

Bounded expansion and nowhere dense classes are intimately linked to first-order logic and many of the developments were driven by the question to solve the first-order model checking problem on these classes. Structural decompositions derived from the generalized coloring numbers lead to quantifier-elimination, and as a consequence to efficient model checking as shown by Dvo\v{r}\'ak, Kr\'al' and Thomas~\cite{dvovrak2013testing}, and enumeration and query counting as shown by Kazana and Segoufin~\cite{kazana2013enumeration}, for first-order logic on classes with bounded expansion. 
Game characterizations and sparse neighborhood covers lead to tractable first-order model checking on nowhere dense classes as shown by Grohe et al.~\cite{grohe2017deciding}. In fact, for monotone (that is, subgraph-closed) classes of graphs nowhere denseness turns out to be the dividing line between tractable and intractable first-order model checking. 
Hence, the classification of tractable first-order model checking on monotone graph classes is complete and leaves the question of a classification of tractability on hereditary graph classes. 
Figure~3 shows an inclusion diagram of relevant graph classes and the current frontiers of tractable model checking. 

The deep connections of nowhere dense classes to classical model theory (stability theory) are currently guiding an exciting development of dense but structurally sparse graph classes. 
Stability is one of the key notions from classical model theory (see the monumental work of Shelah~\cite{shelah1972combinatorial}) and it 
turns out that nowhere dense classes are stable, as observed by Adler and Adler~\cite{adler2014interpreting}, building on work of Podewski and Ziegler~\cite{podewski1978stable}. 
This connection is highly fruitful and many notions from model theory can be appropriately translated to graph theory and turn out to be dense analogs of well known concepts from sparsity theory. 

Another recent exciting development evolves around the graph product structure theory for planar graphs, initiated by Dujmovi{\'c} et al.~\cite{dujmovic2020planar}, which was inspired by the techniques developed for the construction of low treedepth colorings for classes with excluded minors by Pilipczuk and Siebertz~\cite{pilipczuk2019polynomial}. 

\begin{figure}[ht]
    \begin{center}
    \includegraphics[scale=0.42]{figures/universe.pdf}
    \caption{The map of the (structurally) sparse universe and tractable model checking.}\vspace{-3mm}
\end{center}\label{fig:universe}
\end{figure}

In this survey, we overview combinatorial and algorithmic applications
of the generalized coloring numbers, emphasizing new 
developments in this area. 
We provide the proofs of a few key results that connect the relevant concepts and will allow the reader to obtain a deeper understanding. The first part of the survey
deals with upper and lower bounds for the generalized
coloring numbers in general and on restricted graph classes. 
We give background on graph theory and formally define the generalized
coloring numbers in \cref{sec:generalized-col}. 
We also give some background on first-order logic to later give an outlook on the theory of structural sparsity, however, this is not the focus of this survey. 
We present
upper bounds for the generalized coloring numbers and make the connection to the theory of
bounded expansion and nowhere dense graph classes in 
\cref{sec:upper-bounds}. We also present a simple proof for the
existence of uniform orders in \cref{sec:uniform-orders}; 
the existence of such orders 
was first proved by Van den Heuvel and Kierstead~\cite{vdH18}. 
In \cref{sec:tree-decomposable} we turn our attention to tree 
decomposable graphs and make the connection to the recently
emerging graph product structure theory. 
We present lower bounds
in \cref{sec:lower-bounds}.


In \cref{sec:ltc} we study structural decompositions of graphs, 
such as low treedepth decompositions and their generalizations. 
In \cref{sec:Splitter} we consider game characterizations of bounded expansion and nowhere dense graph classes and highlight their connections to the generalized coloring numbers. 
We study further applications, such as neighborhood complexity,  
vc-density, sparse neighborhood covers, 
approximations and parameterized algorithms for the distance-$d$ dominating set problem in \cref{sec:neigh-comp}. Further selected topics are presented in \cref{sec:further-applications}. 
We give more background and references 
to the discussed concepts in the respective sections.

\section{The generalized coloring numbers}\label{sec:generalized-col}

Unless explicitly stated differently, all graphs in this paper
are finite, undirected and simple. We follow the notation of 
Diestel's textbook~\cite{diestel2018graph} and refer to it for all graph theoretic
notations that are left undefined here. 
An order of the vertex set $V(G)$ of an $n$-vertex graph $G$ is a permutation 
$\pi=(v_1,\ldots, v_n)$. We say that $v_i$ is smaller than $v_j$
and write $v_i<_\pi v_j$ if $i<j$. We write $\Pi(G)$ for the 
set of all orders of $V(G)$. 
We write $\mathcal{P}(X)$ for the powerset of a set $X$. 

\subsection{Treewidth and the strong \textit{r}-coloring numbers}

The \emph{generalized coloring numbers} $\col_r(G)$ and $\wcol_r(G)$
were introduced by Kierstead and Yang~\cite{kierstead2003orderings}
as generalizations of the usual coloring number. Here, we take
a different perspective and introduce the strong $r$-coloring number
$\col_r(G)$ as a local version of \emph{treewidth}, while the 
weak $r$-coloring number $\wcol_r(G)$ will be introduced in the 
next subsection as a local version of \emph{treedepth}. This approach 
supports our view of these parameters as \emph{measures for the local 
connectivity properties} of a graph, and sets the stage for the study of tree decomposable graphs. 

Treewidth, as a measure of the global connectivity of a graph, 
plays an important role in graph theory. 
It was originally introduced by Bertelè and Brioschi under the name of dimension~\cite{bertele1973non}.
It was rediscovered by Halin~\cite{halin1976s}, and again rediscovered by Robertson and Seymour~\cite{RobertsonS86} as part of their graph
minors project.
We refer to~\cite{bodlaender1998partial, harvey2017parameters}
for surveys on various characterizations of treewidth and 
equivalent parameters and to~\cite{bodlaender2006treewidth,cygan2015parameterized} 
for surveys on algorithmic applications. 

\begin{definition}\label{def:treedec}
    A {\em{tree decomposition}} of a graph $G$ is a pair
    $\Tt=(T,\bag)$, where $T$ is a tree and
    $\bag\colon V(T)\to \mathcal{P}(V(G))$ is a mapping that assigns each node
    $x$ of $T$ to its {\em{bag}} $\bag(x)\subseteq V(G)$ so that the
    following conditions are satisfied (see Figure 4):
    \begin{enumerate}
    \item\label{p:amoeba} For each $u\in V(G)$, the set
      $\{x \mid u\in \bag(x)\}$ is non-empty and induces a connected
      subtree of $T$.
    \item\label{p:edge} For every edge $uv\in E(G)$, there is a node $x\in V(T)$ such
      that $\{u,v\}\subseteq \bag(x)$.
    \end{enumerate}
  \end{definition}
  
  Let $\Tt=(T,\bag)$ be a tree decomposition of $G$.
  The {\em{width}} of $\Tt$ is the maximum bag size minus~$1$, i.e., $\max_{x\in V(T)} |\bag(x)|-1$. 
  The {\em{treewidth}} of a graph $G$ is the minimum possible width of a tree decomposition of $G$. For $X\subseteq V(T)$ we let $\bag(X)=\bigcup_{x\in X}\bag (x)$. 

  \begin{figure}[ht]
    \begin{center}
    \includegraphics[width=0.6\textwidth]{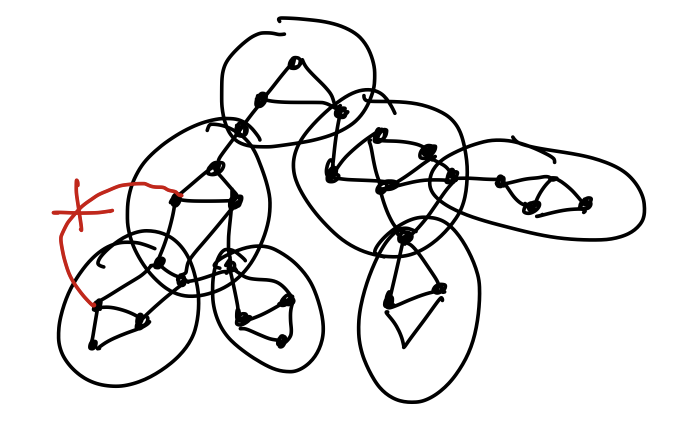}
    \caption{A tree decomposition of a graph. No edge can cross the intersection of two bags, which form separators in $G$.}    
\end{center}
\end{figure}
  
  In the following, it will be convenient to assume that $T$ is a rooted tree. A \emph{rooted tree} is a tree~$T$ with a designated root vertex $r$. Assigning a root imposes the 
  standard \emph{ancestor/descendant
  relation} $\leq_T$ in $T$ where the root is the $\leq_T$ minimal element: 
  a node $v$ is a descendant of all the nodes that appear on the unique 
  path leading from $v$ to the root. Note that we treat a vertex as an 
  ancestor and descendant of itself. The ancestors and descendants of 
  $v$ excluding $v$ are its true ancestors and true descendants, respectively. 
  The \emph{parent} 
  of a non-root node is its closest true ancestor. 
  A \emph{rooted forest} $F$ is a disjoint union of rooted trees. We write
  $\leq_F$ for the partial order that coincides with $\leq_T$ on each rooted
  tree $T$ in $F$.

  \smallskip
  Tree decompositions capture the global connectivity properties of 
  graphs, as made precise~next. 
  
  \begin{definition}
  Let $G$ be a graph and let $U,W\subseteq V(G)$. A set $S\subseteq V(G)$ \emph{separates}
  $U$ and $W$, if every path from a vertex of $U$ to a vertex of $W$
  contains a vertex of $S$. 
  \end{definition}
  
  \begin{lemma}[\cite{RobertsonS86}]\label{lem:sep}
  Let $G$ be a graph and let $\Tt=(T,\bag)$ be a tree decomposition
  of $G$. 
  For an edge $e=xy\in E(T)$ let $T_1,T_2$ be the components of $T-e$. 
  Then $\bag(x)\cap \bag(y)$ separates $\bag(T_1)$ and $\bag(T_2)$.
  \end{lemma}

Let $\Tt=(T,\bag)$ be a tree decomposition of $G$. This decomposition induces a quasi-order on~$V(G)$ as follows. 
For $v\in V(G)$, let $x(u)$ be the unique~$\leq_T$-minimal node of $T$ with $v\in \bag(x)$. 
Define the quasi-order $\leq_{\Tt}$ on $V(G)$ by $u\leq_{\Tt} v$ if and only if $x(u)\leq_T x(v)$. Any linearization of~$\leq_\Tt$ has very nice properties, which we will exploit later. In fact, we can associate a \emph{width} to any linear order of $V(G)$ and thereby provide another characterization of treewidth. 

\begin{definition}\label{def:elimination-order}
Let $G$ be a graph. A vertex $v\in V(G)$ is \emph{simplicial}
if its neighborhood induces a clique in $G$. An order
$\pi=(v_1,\ldots, v_n)$ of $V(G)$ is a \emph{perfect elimination
order} if for every $1\leq i\leq n$, $v_i$ is simplicial in 
$G[V_{i}]$, where $V_{i}=\{v_1,\ldots, v_{i}\}$. 
\end{definition}

\begin{definition}
A graph $G$ is \emph{chordal} or \emph{triangulated} if all 
its cycles with four or more vertices have a \emph{chord}, that is, 
an edge that is not part of the cycle but connects two vertices of the 
cycle. Equivalently, $G$ is chordal if every induced cycle in $G$ has
exactly three vertices. A \emph{triangulation} of $G$ is a 
chordal supergraph $H$ of $G$ on the same vertex set as $G$. 
\end{definition}

\begin{lemma}[\cite{fulkerson1965incidence}]\label{lem:perfect-elimination}
A graph $G$ is chordal if and only if it admits a perfect elimination
order. 
\end{lemma}

Let $G$ be a graph and let $\pi=(v_1,\ldots, v_n)$ be an order
of $V(G)$. We construct the graph~$G^\pi_{\textit{fill}}$,
called the \emph{fill-in graph} of $G$ with respect to the order $\pi$, 
by the following fill-in procedure. We let $G_n\coloneqq G$. 
Now, for $i=n-1,\ldots, 1$, assume that the graph~$G_{i+1}$ has 
been constructed. We define~$G_i$ as the graph obtained
from $G_{i+1}$ by adding all edges between non-adjacent neighbors of~$v_{i+1}$ that are smaller than $v_i$. Finally, we let $G^\pi_{\textit{fill}}\coloneqq G_1$. 
It is immediate from the definition that~$\pi$
is a perfect elimination order of $G^\pi_{\textit{fill}}$ (see Figure 5). 

\begin{figure}[ht]
  \begin{center}
 \includegraphics[width=\textwidth]{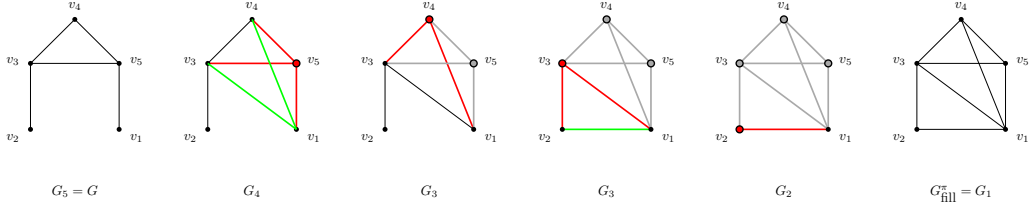}
  \caption{An elimination order and fill-in graph. The neighbors of $v_i$ that are smaller than $v_i$ are marked by red edges in~$G_{i-1}$, the fill-in edges that are added between these neighbors are marked by green edges.}
\end{center}
\end{figure}

Recall that the back-degree of a vertex $v$ with respect to an 
order $\pi$ is the number of neighbors of $v$ that are smaller than
$v$ with respect to $\pi$. 

\begin{definition}\label{def:width-order}
Let $G$ be a graph and let $\pi=(v_1,\ldots, v_n)$ be an order
of $V(G)$. The \emph{width} of $\pi$ on~$G$ is the 
maximum back-degree of any vertex $v$ in 
the fill-in graph $G^\pi_{\textit{fill}}$ (with respect to the order
$\pi$). 
\end{definition}

\begin{lemma}[see e.g.~\cite{bodlaender1998partial}]\label{lem:width-order}
Let $G$ be a graph. The treewidth $\tw(G)$ of $G$ is equal to the minimum
width over all orders $\pi$ of $V(G)$. 
\end{lemma}

An easy induction shows that for $1\leq i<j\leq n$, there is an 
edge between~$v_i$ and~$v_j$ in the fill-in graph $G^\pi_{\textit{fill}}$ 
if and only if there exists a path $P$ between
$v_i$ and $v_j$ such that~$v_i$ is the smallest vertex on $P$
and all internal vertices of $P$ are larger than $v_j$
with respect to $\pi$. We say 
that $v_i$ is \emph{strongly reachable} from $v_j$. For convenience, we also consider each vertex to be strongly reachable from itself and we write
$\SReach[G,\pi,v]$ for the set of vertices that are strongly reachable
from $v$. The above observation is the motivation for the following
definition of \emph{local strong reachability}.

\begin{definition}
Let $G$ be a graph and $r$ a positive integer. Let $\pi$ be a 
linear order of~$V(G)$. We say that a
vertex $u\in V(G)$ is \emph{strongly $r$-reachable} with respect
to $\pi$ from a vertex
$v\in V(G)$ if $u\leq_\pi v$ and there exists a path $P$ (possibly of length $0$)
between $u$ and $v$ of length at most $r$ with $w>_\pi v$ for all internal 
vertices $w\in V(P)$ (see Figure 6). The set of vertices strongly $r$-reachable
by $v$ with respect to the order $\pi$ is denoted $\SReach_r[G,\pi,v]$.
We define
\[\col_r(G,\pi)\coloneqq \max_{v\in V(G)}|\SReach_r[G,\pi,v]|,\]
and the \emph{strong $r$-coloring number} $\col_r(G)$ as
\[\col_r(G)\coloneqq \min_{\pi\in \Pi(G)}\max_{v\in V(G)}|\SReach_r[G,\pi,v]|.\]
\end{definition}

\begin{figure}[ht]
  \begin{center}
  \includegraphics[width=.75\textwidth]{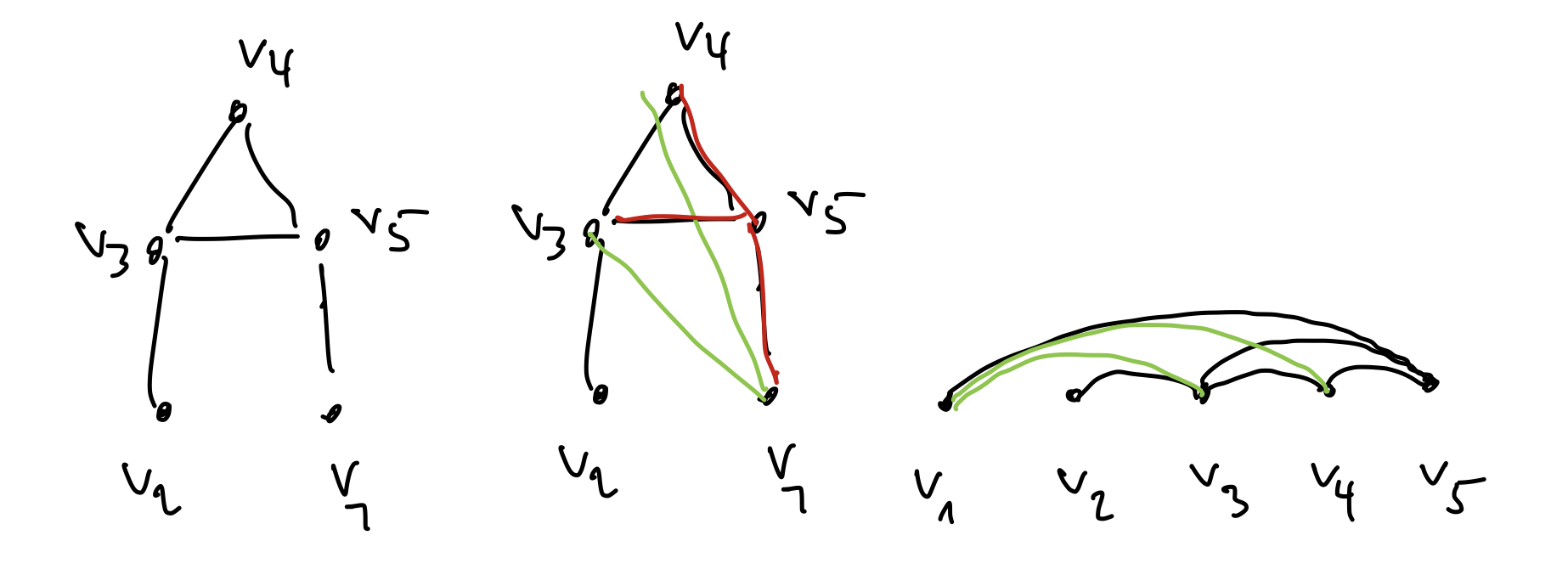}
  \caption{Vertices $v_3$ and $v_4$ strongly $2$-reach vertex $v_1$. This is witnessed by the path along $v_5$. }
\end{center}
\end{figure}

Note that since we allow the path $P$ to be of
length $0$ every vertex strongly $r$-reaches
itself. It is immediate from the definitions that 
\[\col(G)=\col_1(G)\leq \col_2(G)\leq \ldots \leq \col_n(G)=
\tw(G)+1.\]

Hence, the strong $r$-coloring numbers can be seen as gradations 
between the coloring number~$\col(G)$ and the treewidth $\tw(G)$ of
$G$. As such, the strong $r$-coloring numbers 
capture local separation properties of $G$, and many applications are based
on lemmas similar to the following. 

\pagebreak

\begin{lemma}\label{lem:sep-col}
Let $G$ be a graph, let $\pi$ be an order of $V(G)$, and let 
$r$ be a positive integer. Let $u,v\in 
V(G)$ such that $u<_\pi v$. Then every path $P$ of length at most 
$r$ connecting $u$ and $v$ intersects $\SReach_r[G,\pi,v]\setminus\{v\}$. 
\end{lemma}

Observe that if $u\in \SReach_r[G,\pi,v]$, then the statement of the
lemma is trivial. 

 \begin{proof}
 Let $P$ be any path of length
 at most $r$ connecting $u$ and $v$. Traverse $P$ starting 
 at the endpoint $v$. Then the first vertex in this 
 traversal of $P$ that is strictly smaller than $v$ lies in 
 $\SReach_r[G,\pi,v]\setminus\{v\}$. 
 \end{proof}

There is some ambiguity in the literature on whether a vertex
$v$ itself should be included in the set $\SReach_r[G,\pi,v]$. 
We have chosen to include it as some arguments become
a bit nicer. In fact, some people argue that one should not subtract $1$ in the definition of treewidth.

\subsection{Treedepth and the weak \textit{r}-coloring numbers}


The notion of treedepth was introduced by Ne\v{s}et\v{r}il and
Ossona de Mendez 
in~\cite{nevsetvril2006tree}, and
again, equivalent notions were studied before under different 
names. We refer to~\cite{nevsetril2012sparsity} for a discussion 
on various equivalent parameters. 
Where intuitively treewidth measures the similarity
of a graph with a tree, treedepth measures the similarity 
of a graph with a star. 



The \emph{depth} of a vertex $v$ in a rooted forest $F$ 
is the number of vertices on the path from $v$ to the root 
(of the tree to which $v$ belongs).  The \emph{depth} 
of $F$ is the maximum depth of the vertices of~$F$.

\begin{definition}
Let $G$ be a graph. The \emph{treedepth} $\td(G)$ of $G$
is the minimum depth of a rooted forest $F$ on the 
same vertex set as $G$ such that 
whenever $uv\in E(G)$, then $u\leq_F v$ or $v\leq_F u$ (see Figure 7). 
\end{definition}

\begin{figure}[ht]
  \begin{center}
  \includegraphics[width=.75\textwidth]{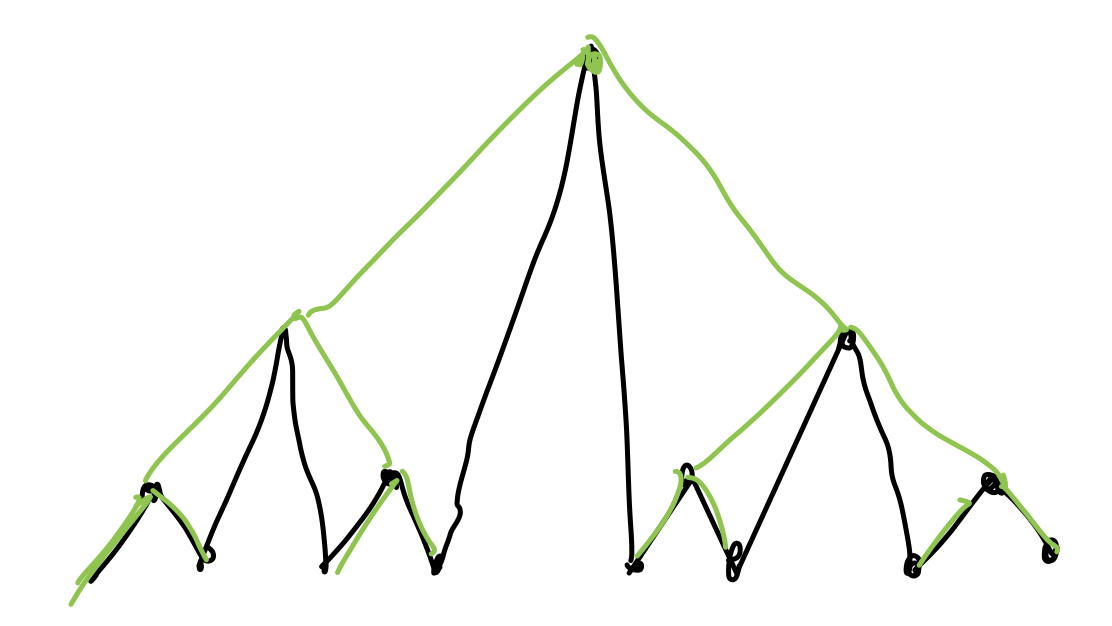}
  \caption{The treedepth of a path $P_n$ with $n$ vertices is $\lceil \log_2(n+1)\rceil$. }
\end{center}
\end{figure}

Similar to the width of an order $\pi$, we can define its depth
via a corresponding notion of reachability as follows. 

\begin{definition}
Let $G$ be a graph and let $\pi$ be 
an order of $V(G)$. We say that a 
vertex $u\in V(G)$ is \emph{weakly reachable} with 
respect to $\pi$ from a vertex $v\in V(G)$ if $u\leq_\pi v$ and 
there exists a path $P$ (possibly of length $0$) between $u$ and $v$ with 
$w>_\pi u$ for all internal vertices $w\in V(P)$. We write $
\WReach[G,\pi,v]$ for the set of vertices that are
weakly reachable from $v$. The \emph{depth} 
of~$\pi$ on $G$ is the maximum over all vertices $v$ of
$G$ of $|\WReach[G,\pi,v]|$. 
\end{definition}

\begin{lemma}[see e.g.~\cite{nevsetril2012sparsity}, Lemma 6.5]
Let $G$ be a graph. The treedepth of $G$ is equal to 
the minimum depth over all orders $\pi$ of $V(G)$. 
\end{lemma}

\pagebreak
And again we can naturally define a local version of 
weak reachability. 

\begin{definition}
Let $G$ be a graph and $r$ a positive integer. Let $\pi$ be a 
linear order of~$V(G)$. We say that a
vertex $u\in V(G)$ is \emph{weakly $r$-reachable} with respect
to $\pi$ from a vertex
$v\in V(G)$ if $u\leq_\pi v$ and there exists a path $P$ (possibly of length $0$)
between $u$ and $v$ of length at most~$r$ with $w>_\pi u$ for all
internal vertices $w\in V(P)$ (see Figure 8). The set of vertices weakly $r$-reachable
by $v$ with respect to the order $\pi$ is denoted $\WReach_r[G,\pi,v]$.
We define 
\[\wcol_r(G,\pi)\coloneqq \max_{v\in V(G)}|\WReach_r[G,\pi,v]|,\]
and the \emph{weak $r$-coloring number} $\wcol_r(G)$ as
\[\wcol_r(G)\coloneqq \min_{\pi\in \Pi(G)}\max_{v\in V(G)}|\WReach_r[G,\pi,v]|.\]
\end{definition}

\begin{figure}[ht]
  \begin{center}
  \includegraphics[width=.5\textwidth]{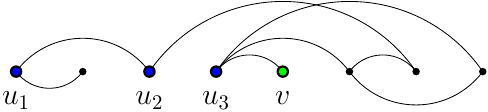}
  \caption{Vertex $v$ weakly $4$-reaches the vertices $u_1,u_2,u_3$ (and itself).}
\end{center}
\end{figure}

Note that again we also consider the path of length $0$ in the definition of weak reachability, hence, each vertex is weakly reachable from itself. And again, it is immediate from the definitions that 
\[\col(G)=\wcol_1(G)\leq \wcol_2(G)\leq \ldots \leq \wcol_n(G)=
\td(G).\] 

Hence, the weak $r$-coloring numbers can be seen as gradations 
between the coloring number $\col(G)$ and the treedepth $\td(G)$ of
$G$. The weak $r$-coloring numbers 
capture local separation properties of $G$ as follows. 

\begin{lemma}\label{lem:wcol-sep}
Let $G$ be a graph, let $\pi$ be an order of $V(G)$ and let $r$ be 
a positive integer. Let $u,v\in 
V(G)$ such that $u<_\pi v$. Then every path $P$ of length at most 
$r$ connecting $u$ and $v$ intersects $\WReach_r[G,\pi,v]\cap
\WReach_r[G,\pi,u]$. 
\end{lemma}


\begin{proof}
Let $P$ be a path of length
at most $r$ connecting $u$ and $v$. Then the minimum 
vertex of $P$ lies both in $\WReach_r[G,\pi,v]$ and in 
$\WReach_r[G,\pi,u]$. 
\end{proof}

%
And again, there is some ambiguity in the literature on whether a vertex
$v$ itself should be included in the set $\WReach_r[G,\pi,v]$ 
and we have chosen to include it.

\medskip
For the strong $r$-reachability, we follow paths that stop at the first point 
smaller than their starting vertex. For the weak $r$-reachability, we consider 
paths that can make up to $r$ hops to smaller and smaller vertices. 
The following gradation between the weak and strong coloring numbers was
introduced by Cort{\'e}s et al.\ in~\cite{cortes2023subchromatic}. 

\begin{definition}
    Let $G$ be a graph and $r$ a positive integer. Let $\pi$ be a 
    linear order of~$V(G)$. We say that a
    vertex $u\in V(G)$ is \emph{$\ell$-hop $r$-reachable} with respect
    to $\pi$ from a vertex
    $v\in V(G)$ if $u\leq_\pi v$ and there exists a path $P=v_0v_1\ldots v_s$ with $v_0=v$, $v_s=u$ for some $s\leq r$ such that $u<_\pi v_{i-1}$ for every $i\in [s]$ and such that $|\{j\leq s \mid v_j<_\pi v_{i-1}$ for every $i\leq j\}|\leq \ell$. 
    The set of vertices that are $\ell$-hop $r$-reachable from $v$ with 
    respect to the order $\pi$ is denoted $\GReach_{\ell,r}[G,\pi,v]$.
    We define the \emph{$\ell$-hop $r$-coloring number of $\pi$} as
    \[\gcol_{\ell,r}(G,\pi)\coloneqq \max_{v\in V(G)}|\GReach_{\ell,r}[G,\pi,v]|,\]
    and the \emph{$\ell$-hop $r$-coloring number} $\gcol_{\ell,r}(G)$ as
    \[\gcol_{\ell,r}(G)\coloneqq \min_{\pi\in \Pi(G)}\max_{v\in V(G)}|\GReach_{\ell,r}[G,\pi,v]|.\]
    \end{definition}

Observe that $\col_r(G)=\gcol_{1,r}(G)$ and $\wcol_r(G)=\gcol_{r,r}(G)$. The paths from 
$\gcol{\ell,r}(G)$ can be split into shorter paths whenever they reach a local minimum with respect to the order. In particular, 
\[\gcol_{\ell+\ell',r+r'}(G) \leq \gcol_{\ell,r}(G)\cdot \gcol_{\ell',r'}(G).\]

\subsection{The r-admissibility}

Finally, we define one last related measure, the 
\emph{$r$-admissibility}
of $G$, which was introduced by 
Dvo\v{r}\'ak~\cite{dvovrak2013constant} as a 
generalization of the admissibility~\cite{kierstead1994planar}.

\begin{definition}
Let $G$ be a graph, $v\in V(G)$ and $A\subseteq V(G)$. 
A $v$-$A$ \emph{fan} is a 
set of paths $P_1,\ldots P_k$ with one endpoint in $v$ and the
other endpoint in $A$ and which are internally 
vertex disjoint from $A$ 
such that $V(P_i)\cap V(P_j)=\{v\}$
for all $i\neq j$.
\end{definition}

Note that if $v\in A$, then we may have the path of length $0$ consisting only of $v$ in a $v$-$A$ fan. We remark that this differs from the standard definition of $v$-$A$ fans, which requires paths with one endpoint in $v$ and the other endpoint in $A\setminus \{v\}$. We have chosen this definition for consistency with the strong and weak coloring numbers, where each vertex is reachable from itself. 

\begin{definition}
Let $G$ be a graph and let $\pi=(v_1,\ldots, 
v_n)$ be a linear order of $V(G)$. For $1\leq i\leq n$, 
we write $V_i$ for the set $\{v_1,\ldots, v_i\}$. 
For $1\leq i\leq n$, let $\adm_\infty[G,\pi,v_i]$ be the maximum 
size of a $v_i$-$V_i$ fan. 
We define
the \emph{$\infty$-admissibility} $\adm_\infty(G)$ of $G$ as 
\[\adm_\infty(G)\coloneqq\min_{\pi\in \Pi(G)}\max_{v\in V(G)}\adm_\infty[G,\pi,v].\] 
\end{definition}

Again, we define a local version of admissibility. 

\begin{definition}\label{def:fan}
Let $G$ be a graph and let $r$ be a positive integer. Let $v\in V(G)$ and $A\subseteq V(G)$. 
A \emph{depth-$r$} $v$-$A$ \emph{fan} is a 
set of paths $P_1,\ldots P_k$ of length at most $r$ 
with one endpoint in $v$ and the
other endpoint in $A$ and which are internally 
vertex disjoint from $A$ 
such that $V(P_i)\cap V(P_j)=\{v\}$
for all $i\neq j$.
\end{definition}

\begin{definition}
Let $G$ be a graph and let $r$ be a positive integer. 
Let $\pi=(v_1,\ldots, 
v_n)$ be a linear order of $V(G)$. 
For $1\leq i\leq n$, let $\adm_r[G,\pi,v_i]$ be the maximum 
size of a depth-$r$ $v_i$-$V_i$ fan. 
We define 
\[\adm_r(G,\pi)\coloneqq \max_{v\in V(G)}\adm_r[G,\pi,v],\]
and the \emph{$r$-admissibility} of $G$, 
$\adm_r(G)$, as 
\[\adm_r(G)\coloneqq \min_{\pi\in \Pi(G)}\max_{v\in V(G)}\adm_r[G,\pi,v].\] 
\end{definition}


One last time, it is immediate from the definitions that 
\[\col(G)=\adm_1(G)\leq \adm_2(G)\leq \ldots \leq \adm_n=\adm_\infty(G).\] 

We will discuss the limit parameter $\adm_\infty(G)$ later, it turns
out that graphs classes with finite~$\adm_\infty(G)$ are 
those whose graphs 
are tree decomposable over torsos that have a bounded number 
of high degree vertices. 

\medskip
The three series $\adm_r(G), \col_r(G)$ and $\wcol_r(G)$ 
are generally referred to as the \emph{generalized coloring
numbers}. We have the following inequalities, which enable
us to switch between them according to 
which is more suitable for a particular need.
The inequalities of the first item are immediate from the definitions. For the
proofs of the other items, we refer to~\cite{siebertz2016nowhere}.

\begin{lemma}\label{lem:adm-col-wcol}
Let $G$ be a graph, let $r$ be a positive integer. 
Then 
\begin{enumerate}
\item $\adm_r(G)\leq \col_r(G)\leq\wcol_r(G)$. 
\item $\col_r(G)\leq (\adm_r(G)-1)\cdot(\adm_r(G)-2)^{r-1}+1 \leq \adm_r(G)^r$.
\item $\wcol_r(G)\leq \col_r(G)^r$.
\item $\wcol_r(G)\leq \adm_r(G)^r$.
\end{enumerate}
\end{lemma}

We remark that the above relations involving the 
$r$-admissibility are stated differently in the work of Dvo\v{r}\'ak~\cite{dvovrak2013constant}, due to 
the fact that the set of vertices weakly or strongly reachable 
from a vertex~$v$ exclude the vertex $v$ in his work. 

\medskip
In our study of specific graph classes below, we will also see that these bounds are almost tight.

\subsection{Logic and transductions}\label{sec:logic}

Let us recall some basics from first-order model theory. For extensive background we refer to Hodges' textbook~\cite{hodges1993model}. 
A \emph{signature} is a collection
of relation and function symbols, each with an associated arity. Let $\sigma$ be a
signature. A {\em $\sigma$-structure} $\strA$ consists of a non-empty
set $A$, the \emph{universe} of~$\strA$, together with an interpretation of
each $k$-ary relation symbol $R\in\sigma$ as a $k$-ary relation
$R^\strA\subseteq A^k$ and an interpretation of each $k$-ary function symbol $f\in \sigma$ as a $k$-ary function $f^\strA:A^k\rightarrow A$. 
In this work we assume that all structures are finite (i.e.\ have a finite universe and a finite signature). 

We now define first-order logic FO and monadic second-order logic MSO (over the signature~$\sigma$). We assume an infinite supply $\textsc{Var}_1$ of first-order variables (which will range over elements) and an infinite supply $\textsc{Var}_2$ of monadic second-order variables (which will range over sets of elements). 
Every variable is a term, and if $t_1,\ldots, t_k$ are terms and $f\in \sigma$ is a $k$-ary function symbol, then also $f(t_1,\ldots, t_k)$ is a term. 
First-order formulas are built from the atomic formulas $t_1=t_2$, where~$t_1$ and $t_2$ are terms, and $R(t_1,\ldots, t_k)$, where \mbox{$R\in \sigma$}
is a $k$-ary relation symbol
and $t_1,\ldots, t_k$ are terms, by closing under the Boolean
connec\-tives~$\neg$,~$\wedge$~and~$\vee$, and by existential and
universal quantification~$\exists x$ and $\forall x$. 
Monadic second-order formulas are defined as first-order formulas, but further
allow the use of monadic quantifiers $\exists X$ and $\forall X$, and of a membership atomic formula $x\in X$, where $x$ is a first-order variable and $X$ a monadic second-order variable. 

A variable~$x$ not in the scope of a quantifier is a {\em free variable} (we do not consider formulas with free second-order variables). A formula without free variables is a {\em sentence}.
The {\em quantifier rank} $\mathrm{qr}(\varphi)$ of a formula $\varphi$ is the
maximum nesting depth of quantifiers in~$\phi$. 
A formula without quantifiers is called {\em quantifier-free}.

If $\strA$ is a $\sigma$-structure
with universe $A$, then an {\em assignment} of the variables in~$\strA$
is a mapping $\bar a:\textsc{Var} \rightarrow A$. We use the standard
notation $(\strA, \bar a)\models \phi(\bar x)$ or $\strA \models \phi(\bar a)$
to indicate that $\phi$ is satisfied in $\strA$ when the free variables $\bar x$
of $\phi$ have been assigned by $\bar a$. 
We write $\models\phi$ to express that $\phi$ is a valid sentence, that is, $\phi$ holds in every structure (of an appropriate signature). 
For a formula $\phi(\bar x)$ we define $\phi(\strA):=\{\bar a\in A^{|\bar x|} \mid \strA\models\phi(\bar a)\}$.

\medskip
In this work we consider the following types of structures:
\begin{itemize}
    \item \emph{Colored graphs} are $\sigma$-structures, where $\sigma$ consists of a single binary relation $E$ and unary relations, with the property that
	$E$ is symmetric and anti-reflexive. 
    \item \emph{Guided pointer structures} are $\sigma$-structures, where $\sigma$ consists of a single binary relation~$E$, unary relations, and unary functions, with the property that $E$ is symmetric  and anti-reflexive, and that every function $f\in\sigma$ is \emph{guided}, meaning that if $f(u)=v$, then $u=v$ or $uv\in E(G)$. 
    \item \emph{Ordered graphs} are $\sigma$-structures, where $\sigma$ consists of a two binary relations $E$ and $<$, with the property that $E$ is symmetric and anti-reflexive, and $<$ is a linear order. 
    \item \emph{Partial orders} are $\sigma$-structures, where $\sigma$ consists of binary relation $\sqsubseteq$ that is interpreted as a partial order, in many cases a linear order or tree order. 
\end{itemize}

Let $\sigma,\tau$ be relational signatures and let $\Ll$ be one of $\FO$ or $\MSO$. 
An \emph{$\Ll$-interpretation} $\mathsf{I}$ of $\sigma$-structures in $\tau$-structures is a tuple $\mathsf I=(\nu(\bar x), (\rho_R(\bar x_1,\ldots, \bar x_{ar(R)}))_{R\in \sigma})$, where $\nu(\bar x)$ and $\rho_R(\bar x_1,\ldots, \bar x_{ar(R)})$ are $\Ll$-formulas, and 
$|\bar x|=|\bar x_i|$ for all $1\leq i\leq ar(R)$ for all $R\in \sigma$, where $ar(R)$ is the arity of $R$.  
For every $\tau$-structure~$\strB$, the $\sigma$-structure $\strA=\mathsf I(\strB)$
has the universe~$\nu(\strB)$ and each relation~$R_\strA$ is interpreted as $\rho_R(\strB)$. 
We say that $\nu$ \emph{defines} the universe and $\rho_R$ defines the relation~$R$ of $\mathsf{I}(\strB)$. 
The number of free variables $|\bar x|$ of $\nu$ is the \emph{dimension} of the interpretation. 
A \mbox{$1$-dimensional} interpretation is called a \emph{simple interpretation}. 
In particular, a simple interpretation~$\mathsf{I}$ of graphs in $\tau$-structures is a pair $\mathsf I=(\nu, \eta)$, where $\nu(x)$ and~$\eta(x,y)$ are formulas (such that $\models\forall x \forall y (\eta(x,y)\rightarrow \eta(y,x)\wedge \neg\eta(x,x))$). 
For a $\tau$-structure $\strB$ with universe $B$, the graph~$\mathsf I(\strB)$
has vertex set~$\nu(\strB)=\{u\in B\mid \strA\models\nu(u)\}$ and edge set $\eta(\strA)=\{uv \mid \strA\models\eta(u,v)\}$. 

A \emph{monadic lift} of a $\tau$-structure $\strB$ is a $\tau^+$-expansion
$\strB^+$ of $\strB$, where $\tau^+$ is the union of $\tau$ and a
set $\Sigma$ of unary relation symbols (also called \emph{colors}). 
We also say that $\strB^+$ is a \emph{$\Sigma$-colored} $\tau$-structure. 

A transduction $\mathsf{T}$ of $\sigma$-structures from $\tau$-structures 
is a pair $\mathsf T=(\Sigma, \mathsf I)$, where $\Sigma$ is a finite set of colors and $\mathsf I$ is a simple interpretation of $\sigma$-structures in $\Sigma$-colored \mbox{$\tau$-structures}.

A $\sigma$-structure $\strA$ can
be $\mathsf T$-transduced from a $\tau$-structure $\strB$ if there exists a $\Sigma$-coloring $\strB^+$ of~$\strB$ such that $\strA=\mathsf I(\strB^+)$. 
A class $\Cc$ of $\sigma$-structures can be $\mathsf T$-transduced from a class $\Dd$ of $\tau$-structures if for every structure $\strA\in \Cc$ there exists a structure $\strB\in \Dd$ such that $\strA$ can be $\mathsf T$-transduced from~$\strB$. 
A class $\Cc$ of $\sigma$-structures can be transduced from a class $\Dd$ of $\tau$-structures if it can be
$\mathsf T$-transduced from $\Dd$ for some transduction $\mathsf T$.
An {\em FO/MSO-transduction} is a transduction whose interpretation uses FO/MSO formulas. 


For a relational signature $\sigma$, the \emph{incidence graph} of a $\sigma$-structure $\strA$ with universe $A$ is the colored graph with vertices $A$ and one vertex for each tuple $\bar v$ appearing in a relation~$R_\strA$ (we take multiple copies if a tuple appears in several relations). 
If $\bar v=(v_1,\ldots, v_k)$ and $\bar v\in R_\strA$, then the vertex $\bar v$ in the incidence graph is marked with a color $R$ and connected with an edge of color $i$ with the vertex~$v_i$ for $1\leq i\leq k$. 
Note that $\strA$ is interpretable by a simple interpretation from its incidence graph, but vice versa, in general we need a higher-dimensional interpretation to interpret the incidence graph of a structure from the structure. 



\medskip
We define the structural complexity of a class of structures as the structural complexity of the class of its incidence graphs.
In particular, we can define the above defined concepts of treedepth and treewidth via transductions as follows. We write $\TTT_d$ for the class of trees of depth at most $d$ and $\TTT$ for the class of all trees. 

\begin{itemize}
    \item A class $\Cc$ of graphs has bounded treedepth if and only if the class of incidence graphs of graphs from $\Cc$ can be FO- or MSO-transduced from $\TTT_d$ for some $d\geq 1$ (follows from the work of Ganian et al.~\cite{ganian2019shrub,ganian2012trees}).
    \item A class $\Cc$ of graphs has bounded treewidth if and only if the class of incidence graphs of graphs from $\Cc$ can be MSO-transduced from $\TTT$ by a classical result of Courcelle~\cite{courcelle1992monadic}. 
    \item A class $\Cc$ of graphs has bounded treewidth if and only if the class of incidence graphs of graphs from $\Cc$ can be FO-transduced from the class of all (finite) tree orders as shown by Colcombet~\cite{colcombet2007combinatorial}. 
\end{itemize}

Transductions allow to define dense but \emph{structurally sparse} classes from sparse classes of graphs (or structures). 
For example, we can define the dense analogs of classes with bounded treedepth and classes with bounded treewidth via transductions. 

\begin{itemize}
    \item A class $\Cc$ of graphs has bounded shrubdepth if and only if $\Cc$ can be FO- or MSO-transduced from $\TTT_d$ for some $d\geq 1$ as shown by Ganian et al.~\cite{ganian2019shrub,ganian2012trees}.
    \item A class $\Cc$ of graphs has bounded cliquewidth if and only if $\Cc$ can be MSO-transduced from~$\TTT$ as proved by Courcelle~\cite{courcelle1992monadic}. 
    \item A class $\Cc$ of graphs has bounded cliquewidth if and only if $\Cc$ can be FO-transduced from the class of all (finite) tree orders as shown by Colcombet~\cite{colcombet2007combinatorial}. 
\end{itemize}

We will revisit the concept of structural sparsity in later sections. 
We conclude this section by stating the following basic quantifier-elimination result, which is a key ingredient for the efficient first-order model checking algorithm on classes with bounded expansion by Dvo\v{r}\'ak, Kr\'al' and Thomas~\cite{dvovrak2013testing}. 

\begin{theorem}[\cite{dvovrak2013testing}]\label{thm:qe-trees}
    For every FO-formula $\varphi(\bar x)$ and every class $\Cc$ of colored graphs with bounded treedepth there exists a quantifier-free formula $\tilde\phi(\bar x)$ and a 
    linear time computable map $Y$ such that, for every $G\in\Cc$, $Y(G)$ is a guided expansion of $G$ such that for all tuples of vertices $\bar v$
	\[
	G\models \varphi(\bar v)\quad\iff\quad Y(G)\models\tilde\varphi(\bar v).
	\]
\end{theorem}

\section{General bounds for the generalized coloring numbers}\label{sec:upper-bounds}

In this section, we present lower and upper bounds for 
the generalized
coloring numbers. It turns out that the numbers are intimately
linked to the density of bounded-depth minors and bounded-depth topological minors. These are the central notions in the theory of bounded expansion
and nowhere dense graph classes, which were introduced by 
Ne\v{s}et\v{r}il and Ossona de Mendez~\cite{nevsetvril2008grad,
nevsetvril2011nowhere}. After introducing these concepts, we 
present simple lower bounds for the generalized coloring numbers in terms of the density of depth-$r$ minors,
and depth-$r$ topological minors. 
Throughout this work when $r$ is not explictly quantified it denotes a non-negative integer. 

We then 
present a characterization of graphs with bounded $r$-admissibility
in terms of forbidden substructures, which is due to Dvo\v{r}\'ak~\cite{dvovrak2013constant}. We derive the fixed-parameter
linear time algorithm of Dvo\v{r}\'ak~\cite{dvovrak2013constant}
to compute $\adm_r(G)$. Based on the forbidden substructure characterization we also derive an upper bound on $\adm_r(G)$
in terms of the density of depth-$r$ topological minors, which
is due to Grohe et al.~\cite{grohe2015colouring} and prove the existence
of \emph{universal orders}, that is, orders $\pi$ that witness 
$\adm_r(G)\leq f(r)\cdot \adm_r(G,\pi)$ for all values of $r$
and some function $f$. The existence of such orders 
was first proved by Van den Heuvel and Kierstead~\cite{vdH18}. 

We then discuss the approach of Grohe et al.~\cite{grohe2017deciding} to approximate the weak $r$-coloring numbers via transitive fraternal augmentations, introduced by Ne\v{s}et\v{r}il and Ossona de Mendez~\cite{nevsetvril2008grad}.

\subsection{Graph expansion and simple lower bounds}

The following notions of bounded-depth minors and bounded-depth
topological minors are the fundamental definitions in the theory of bounded
expansion and nowhere dense graph classes introduced by Ne\v{s}et\v{r}il and Ossona de Mendez~\cite{nevsetvril2008grad,
nevsetvril2011nowhere}. 

\begin{definition}
A graph $H$ is a minor of $G$, written $H\minor G$, if there is 
a map $M$ that assigns to every vertex $v\in V(H)$ a
connected subgraph $M(v) \subseteq G$ of $G$ and to every edge
$e\in E(H)$ an edge  $M(e)\in E(G)$ such that
\begin{enumerate}
  \item if  $u,v\in V(H)$ with $u\not= v$, then $M(v)$ and $M(u)$ are vertex disjoint, and
  \item if $e = uv \in E(H)$, then $M(e) = u'v'\in E(G)$ for 
  vertices $u'$ of $M(u)$ and $v'$ of $M(v)$.
\end{enumerate}

\begin{figure}[ht]
    \begin{center}
    \includegraphics[width=0.5\textwidth]{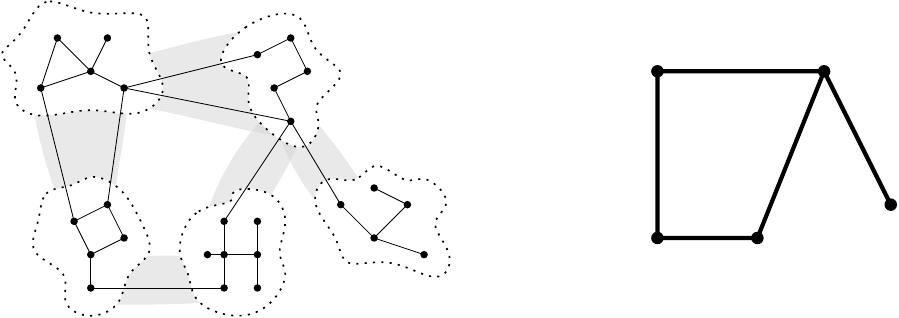}
    \caption{Graph $H$ (on the right) is a depth-$2$ minor of $G$ (on the left).}
\end{center}
\end{figure}

The set $M(v)$ for a vertex $v\in V(H)$ is called the \emph{branch
set} or \emph{model} of $v$ in $G$. 
The map $M$ is called the \emph{model} of $H$ in $G$. 

The graph $H$ is a {\em{depth-$r$ minor}} of $G$, written $H\minor_rG$, if there is a minor model
$M$ of~$H$ in~$G$ such that each branch set $M(v)$ for
$v\in V(H)$ has radius at most $r$. 
\end{definition}

For a class $\Cc$ of graphs we
write $\Cc\mathop{\triangledown}r$ for the class of all
depth-$r$ minors of graphs from $\Cc$.


\begin{definition}
A graph $H$ is a topological minor of $G$, written $H\minor^{top} G$, if there is 
a map $T$ that assigns to every vertex $v\in V(H)$ a
vertex $T(v) \in  V(G)$ of $G$ and to every edge
$e\in E(H)$ a path~$T(e)$ in $G$ such that
\begin{enumerate}
  \item if  $u,v\in V(H)$ with $u\not= v$, then $T(v)\neq T(u)$,
  \item if $e = uv \in E(H)$, then $T(e)$ is a path with endpoints 
  $T(u)$ and $T(v)$, and 
  \item if $e,e'\in E(H)$ with $e\neq e'$, then $T(e)$ and $T(e')$
  are internally vertex disjoint. 
\end{enumerate}

The map $T$ is called the \emph{topological minor model} of $H$ in $G$. 
The graph $H$ is a \emph{topological depth-$r$ minor 
of~$G$}, written $H\minor_r^{top} G$, if there is a topological minor model~$T$ of~$H$ in~$G$ such that
the paths $T(e)$ have length at most $2r+1$ for all $e\in E(H)$ (see Figure 10). 
\end{definition}

For a class $\Cc$ of graphs we
write $\Cc\mathop{\widetilde{\triangledown}}r$ for the class of all
topological depth-$r$ minors of graphs from $\Cc$.

\begin{figure}[ht!]
  \begin{center}
  \includegraphics[width=.5\textwidth]{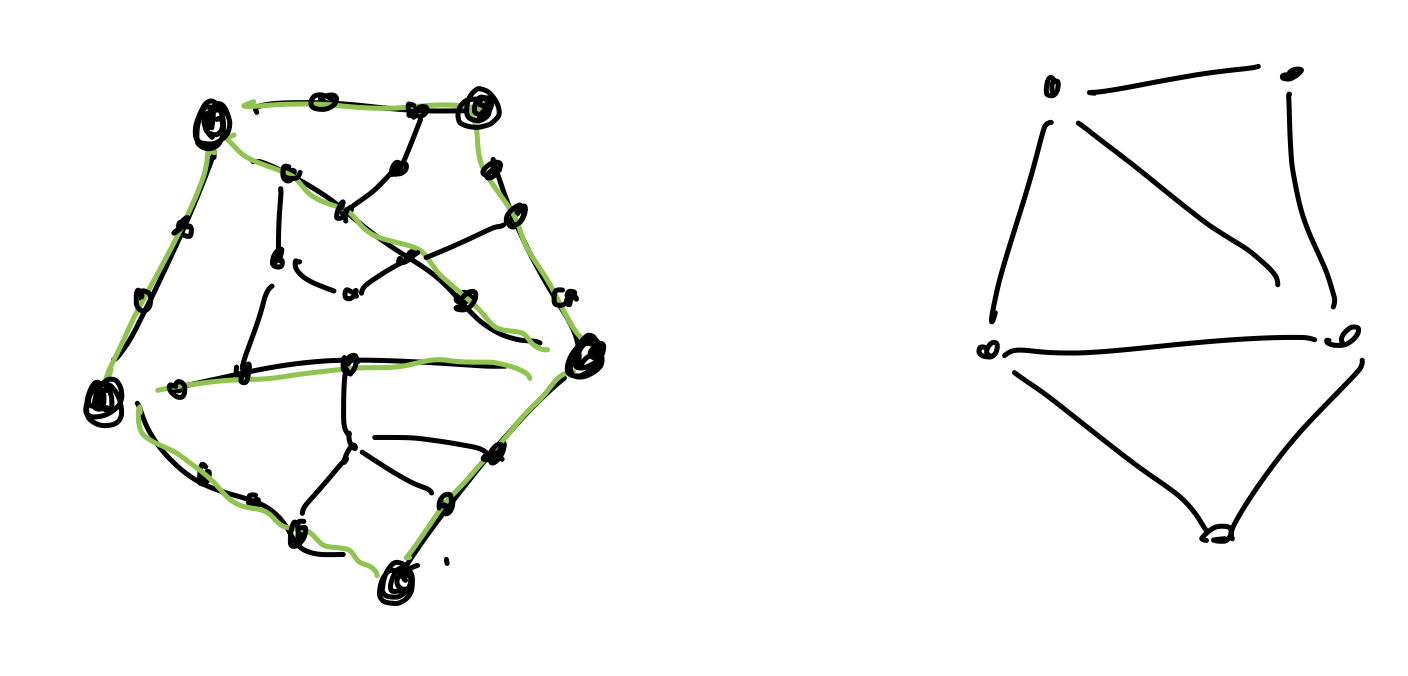}
  \caption{Graph $H$ (on the right) is a depth$-2$ topological minor of $G$ (on the left).}
\end{center}
\end{figure}

\smallskip

Note that if $H\minor^{top} G$, then $H\minor G$ and if $H\minor^{top}_r G$, then $H\minor_r G$. The converse is not true. For example, the maximum degree cannot increase by taking topological minors, that is, if $H\minor^{top} G$, then $\Delta(H)\leq \Delta(G)$, while every graph can be found as a minor of a graph of maximum degree $3$. 
This requires branch sets of large radius though, and as we will see, bounded-depth minors and bounded-depth topological minors behave quite similarly in important aspects. 
Most importantly, the edge densities and clique numbers of graphs that we can find as bounded-depth minors and bounded-depth topological minors are functionally related.



\begin{definition}
Let $G$ be a graph. The \emph{edge density} of $G$ is
\[\nabla(G)=|E(G)|/|V(G)|.\] 
The \emph{greatest reduced average density}~$\nabla_r(G)$  of $G$ with rank~$r$ is 
\[\nabla_r(G)\coloneqq \sup\left\{\frac{|E(H)|}{|V(H)|}\ \colon\ H\minor_r G\right\},\]
and the \emph{greatest topological reduced average density}~$\widetilde{\nabla}_r(G)$  of $G$ with rank~$r$ is 
\[\widetilde{\nabla}_r(G)\coloneqq \sup\left\{\frac{|E(H)|}{|V(H)|}\ \colon\ H\minor_r^{top} G\right\}.\]
Similarly, we define 
\[\omega_r(G)\coloneqq \sup \{t\ \colon\ K_t\minor_r G\}\]
as the largest clique that appears as a depth-$r$ minor, and 
\[\widetilde{\omega}_r(G)\coloneqq \sup \{t\ \colon\ K_t\minor_r^{top} G\}\]
as the largest clique that appears as a topological depth-$r$ minor.
\end{definition}

Note that $\nabla_0(G)$ is equal to the maximum edge density over all subgraphs of $G$. In particular, the degeneracy of $G$ is equal to $2\nabla_0(G)$ and we derive the following lemma for the chromatic number $\chi$ of $G$. 
\begin{lemma}\label{lem:chromatic}
  $\chi(G)\leq 2\nabla_0(G)+1$. If $H\minor_r G$, then $\chi(H)\leq 2\nabla_r(G)+1$ and if $H\minor_r^{top}G$, then $\chi(H)\leq 2\widetilde{\nabla}_r(G)+1$. 
\end{lemma}

We refer to the functions 
$r \mapsto \nabla_r(G)$ and $r \mapsto \widetilde\nabla_r(G)$ as the \emph{expansion} and \emph{topological expansion} of~$G$, respectively.
As proved by Dvo\v{r}\'ak~\cite{dvovrak2007asymptotical}, these measures 
satisfy \[\widetilde\nabla_r(G)\leq \nabla_r(G)\leq 4\big(4\widetilde\nabla_r(G)\big)^{(r+1)^2}.\]

Hence, expansion and topological expansion are functionally equivalent. 
The key observation for the proof is that in the branch set of a high-degree vertex in a bounded-depth minor model, we will always find a vertex of large degree. We can use such a vertex as the branching vertex of a high-degree vertex in a topological minor model.

The functions 
$r \mapsto \omega_r(G)$ and $r \mapsto \widetilde\omega_r(G)$ 
are called the \emph{$\omega$-expansion} and \emph{topological \mbox{$\omega$-expansion}} of~$G$, respectively. 
As shown by Ne\v{s}et\v{r}il and Ossona de Mendez~\cite{nevsetvril2011nowhere} also the $\omega$-expansion and topological \mbox{$\omega$-expansion} are functionally related, see~\cite{notes} for the following bounds. 
\[\widetilde{\omega}_r(G)\leq \omega_r(G)\leq 1+\widetilde{\omega}_{3r+1}(G)^{2r+2}.\]

The above notions  extend to classes $\Cc$ of graphs
by defining $\nabla_r(\Cc)=\sup_{G\in\Cc}\nabla_r(G)$, 
$\widetilde{\nabla}_r(\Cc)=\sup_{G\in\Cc}\widetilde\nabla_r(G)$, 
$\omega(\Cc)=\sup_{G\in\Cc}\omega_r(G)$, and
$\widetilde{\omega}_r(\Cc)=\sup_{G\in\Cc}\widetilde\omega_r(G)$.

\smallskip
We can now define bounded expansion and nowhere dense graph classes.

\begin{definition}
 A class $\Cc$ of graphs has
\emph{bounded expansion} if there exists a function $d$ such 
that for every non-negative integer $r$ we have $\nabla_r(\Cc)\leq d(r)$. 
\end{definition}

The class $\Cc$ has \emph{polynomial expansion} if the function $d$ is polynomially bounded.  

\begin{definition}
A class $\Cc$ of graphs is \emph{nowhere dense} if there exists
a function $t$ such that for every non-negative integer $r$ we have 
$\omega_r(\Cc)\leq t(r)$. 
\end{definition}

Bounded-depth minors in classes of bounded expansion have constant edge density by
definition. It turns out that also bounded-depth minors in nowhere dense classes have small edge density. 

\begin{theorem}[\cite{nevsetvril2011nowhere}]
A class 
$\Cc$ of graphs is nowhere dense if and only 
if for all non-negative integers $r$ and all real
$\epsilon>0$, all
$n$-vertex graphs $G\in\CCC$ satisfy
$\nabla_r(G)\in \Oof_{r,\epsilon}(n^\epsilon)$.  
\end{theorem}

Here, the notation $\Oof_{r,\epsilon}$ hides constant factors depending only on $r$ and $\epsilon$. By rescaling $\epsilon$, we can equivalently say that a class $\Cc$ is nowhere dense if and only if for all real
$\epsilon>0$ and all positive integers $r$ all sufficiently large
$n$-vertex graphs $G\in\CCC$ satisfy
$\nabla_r(G)\leq n^\epsilon$.  

\medskip
Many familiar classes of sparse graphs have bounded expansion
or are nowhere dense. 

\begin{example}
The following classes of graphs have bounded expansion.
\begin{itemize}
\item Every class $\Cc$ that excludes a fixed graph $H$ as
a minor. For such classes, there exists an absolute constant
$c$ such that for all $G\in\Cc$ and all positive integers~$r$ 
we have $\nabla_r(G)\leq c$. In particular, every class that
excludes a fixed minor has polynomial expansion. 
Special cases are the class of
planar graphs, every class of graphs that can be drawn 
with a bounded number of crossings, see~\cite{nevsetvril2012characterisations}, and every class of graphs
that embeds into a fixed surface. 
\item Every class of intersection graphs of low-density objects in low 
dimensional Euclidean space has polynomial expansion~\cite{har2017approximation}. 
\item Every class $\Cc$ that excludes a fixed graph $H$ as
a topological minor. For such classes, there exists an absolute constant
$c$ such that for all $G\in\Cc$ and all positive integers~$r$ 
we have $\widetilde\nabla_r(G)\leq c$. Every class that 
excludes $H$ as a minor also excludes $H$ as a topological 
minor. Further special cases are classes of bounded degree
and classes of graphs that can be drawn 
with a linear number of crossings, see~\cite{nevsetvril2012characterisations}. 
\item Every class of graphs that can be drawn with a bounded 
number of crossings per edge~\cite{nevsetvril2012characterisations}. 
\item Every class of graphs with bounded queue-number, 
bounded stack-number or bounded non-repetitive chromatic number~\cite{nevsetvril2012characterisations}. 
\item The class of Erd\"os-R\'enyi random graphs with 
constant average degree $d/n$, $G(n,d/n)$, has 
asymptotically almost surely bounded expansion~\cite{nevsetvril2012characterisations}. 
\end{itemize}
\end{example}

\begin{example}
The class of graphs $G$ that satisfy $\mathrm{girth}(G)\geq 
\Delta(G)$ is nowhere dense and does not have bounded 
expansion~\cite{nevsetril2012sparsity}.
\end{example}

We will next establish a connection 
between the generalized coloring numbers and the expansion and
topological expansion of a graph, proving that classes of 
bounded expansion are exactly those classes that have 
bounded generalized coloring numbers for every value of $r$. 
We first observe the following simple lower bounds, which 
are variations of Lemma 3.3 of~\cite{zhu2009colouring} proved by Zhu. 

\begin{lemma}[\cite{zhu2009colouring}]\label{lem:wcol-lower-bound}
Let $G$ be a graph and let $r$ be a positive integer. Then 
$\nabla_r(G)\leq \wcol_{2r+1}(G)$. 
\end{lemma}
\begin{proof}
 Let $d\coloneqq \wcol_{2r+1}(G)$. We show that for every 
 depth-$r$ minor~$H$ of~$G$ we can find an orientation $\vec{H}$
such that the out-degree of every vertex $v\in V(\vec{H})$ is at most~$d$. By adding the out-degrees of all vertices in $\vec{H}$ we 
count the edges of $H$, and thereby show that the ratio 
$|E(H)|/|V(H)|$ is bounded by $d$. 

Let $H$ be a depth-$r$ minor of $G$ and let $M$ be a 
depth-$r$ minor model of~$H$ in $G$.
Further, let~$\pi$ be an order of $V(G)$ with
$\wcol_{2r+1}(G,\pi)\leq d$, that is, 
$|\WReach_{2r+1}[G,\pi,v]|\leq d$ 
for each $v\in V(G)$. 
For every vertex $u\in V(H)$, let $c(u)$ be a central vertex of 
$M(u)$, that is, a vertex such that the distance in $M(u)$
between $c(u)$ and every other vertex $v$ of $M(u)$ 
is at most $r$. For $uv\in E(H)$, let~$P(u,v)$ be a path
of length at most $2r+1$ between $c(u)$ and $c(v)$ that
uses only vertices of $M(u)$ and $M(v)$, more precisely, such that
$P(u,v)=(c(u)=v_0,\ldots, v_\ell=c(v))$, $1\leq \ell\leq 2r+1$, and there is an index $i$
such that $v_0,\ldots, v_i$ are vertices of $M(u)$ and $v_{i+1},\ldots, v_\ell$ are vertices of $M(v)$. 

To obtain the orientation $\vec{H}$, we orient an edge $uv\in E(H)$ 
towards $v$ if the smallest vertex~$m(u,v)$ (with respect to $\pi$) of $P(u,v)$ 
lies in $M(v)$, and towards $u$ otherwise. Observe that 
if~$uv$ is oriented towards $v$, then $m(u,v)\in \WReach_{2r+1}[G,\pi, 
c(u)]$. Furthermore, observe that if $uv$ and $uw$ are 
different edges that are oriented away from~$u$, then $m(u,v)$ and
$m(u,w)$ are different vertices (from $M(v)$ and $M(w)$, respectively, which are disjoint by definition). This implies that the out-degree of $u$ in $\vec{H}$ is 
bounded by 
$|\WReach_{2r+1}[G,\pi,c(u)]|\leq d$. Since we considered an arbitrary vertex $u$ of $H$, we have found the desired orientation~$\vec{H}$. 
\end{proof}


\begin{lemma}\label{lem:lowerbound-nabla-col}
 Let $G$ be a graph and let $r$ be a positive integer. Then 
 $\nabla_r(G)\leq \col_{4r+1}(G)$. 
\end{lemma}
\begin{proof}
 Let $d\coloneqq \col_{4r+1}(G)$, let $H$ be a depth-$r$ minor of $G$ and let $M$ be a depth-$r$ minor model of~$H$ in $G$. As in the 
previous proof, we show that we can find an orientation 
$\vec{H}$ of $H$ such that the out-degree of every vertex $v\in V(\vec{H})$ is at most~$d$. Let $\pi$ be an order of $V(G)$ with
$\col_{4r+1}(G,\pi)\leq d$, that is, $|\SReach_{4r+1}[G,\pi,v]|\leq d$ 
for every $v\in V(G)$. 
 
For every vertex $u\in V(H)$, let $m(u)$ be the smallest, 
with respect to $\pi$, vertex of the branch set $M(u)$. To obtain
the orientation $\vec{H}$, we orient an edge $uv\in E(H)$ towards~$v$ if $m(v)<_\pi m(u)$, and towards $u$ otherwise. 

For every edge $uv\in E(H)$ there exists a path $P(u,v)$
of length at most \mbox{$4r+1$} connecting $m(u)$ and $m(v)$ and
which uses only vertices of $M(u)$ and $M(v)$, more precisely, such that
$P(u,v)=(c(u)=v_0,\ldots, v_\ell=c(v))$, $1\leq \ell\leq 4r+1$, and there is an index $i$
such that $v_0,\ldots, v_i$ are vertices of $M(u)$ and $v_{i+1},\ldots, v_\ell$ are vertices of $M(v)$. 
Observe that if~$uv$ is oriented towards $v$, 
then the first vertex~$w(u,v)$ of $P(u,v)$ that lies in $M(v)$ and satisfies $w<_\pi c(u)$ (such a vertex must exist, as $c(v)$ satisfies these properties)
satisfies $w\in \SReach_{4r+1}[G,\pi,c(u)]$. Observe that 
if $uv$ and $uy$ are different edges that are oriented away 
from $u$, then the vertices $w(u,v)$ and $w(u,y)$ are 
different. 
This implies that the out-degree of $u$ in $\vec{H}$ is bounded
by $\SReach_{4r+1}[G,\pi,c(u)]$.  Since $u$ is an arbitrary
vertex of $H$, we have found the desired orientation $\vec{H}$.
\end{proof}

Finally, the following lemma is proved exactly as \Cref{lem:wcol-lower-bound}, by deriving a depth-$r$ minor model $M$ from a topological minor model and using that all branch sets of $M$ are subdivided stars. 

\begin{lemma}
 Let $G$ be a graph and let $r$ be a positive integer. Then 
 $\widetilde\nabla_r(G)\leq \adm_{2r+1}(G)$. 
\end{lemma}
%
%
%
%

\subsection{Bounding \textit{r}-admissibility}

We now turn our attention to upper bounds for the
generalized coloring numbers. Recall \Cref{def:fan} of an $a$-$A$ fan. 

\begin{definition}
Let $G$ be a graph and let $r$ and $k$ be positive integers. 
A \emph{$(k,r)$-fan set} in $G$ is a set $A\subseteq V(G)$ 
such that for every $a\in A$, there exists a depth-$r$ 
$a$-$A$ fan of order $k$ in $G$. 
\end{definition}

\begin{figure}[ht!]
  \begin{center}
  \includegraphics[width=0.75\textwidth]{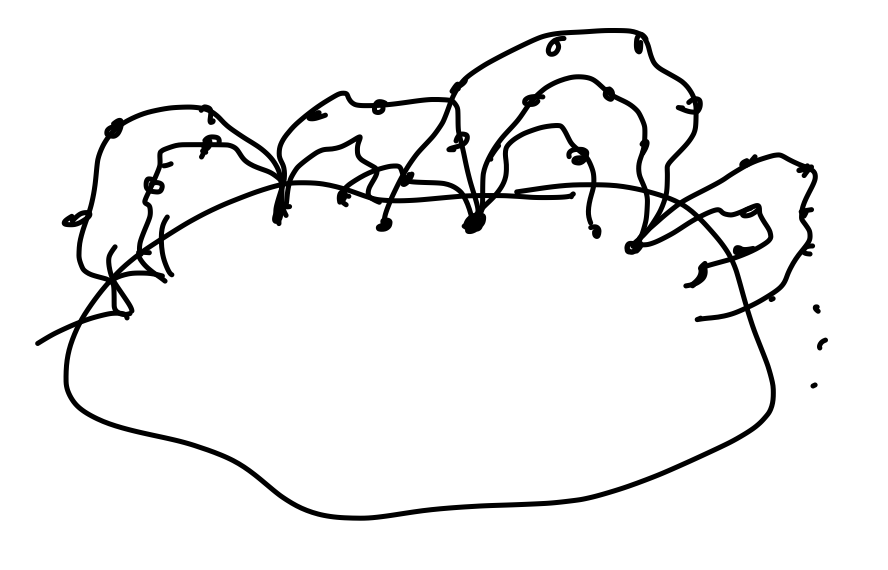}
\end{center}
\caption{A $(3,3)$-fan set.}
\end{figure}

The following theorem is implicit in the work of Dvo\v{r}\'ak~\cite{dvovrak2013constant}. 

\begin{theorem}[\cite{dvovrak2013constant}]\label{thm:adm}
Let $G$ be a graph and let $r$ be a positive integer. Then 
$\adm_r(G)$ is equal to the largest number $k$ such that
there exists a $(k,r)$-fan set in $G$. 
\end{theorem}
\begin{proof}
If $A$ is a $(k,r)$-fan set in $G$, then $\adm_r(G)\geq k$. 
To see this, let $\pi$ be any order of $V(G)$ and let $a$ be the maximal 
element of $A$ with respect to $\pi$. As $A$ is a $(k,r)$-fan set,
there exist paths $P_1,\ldots, P_{k-1}$ of length at most $r$ 
with one endpoint in $a$
and the other endpoint in~$A$ 
that are internally disjoint
from $A$ and such that $V(P_i)\cap V(P_j)=\{a\}$ for all 
$i\neq j$. For every path~$P_i$, there exists an element 
$b_i\in V(P_i)$ that is smaller than or equal to $a$ with respect to $\pi$ by
our choice of the vertex $a$. Hence, the vertex $a$ witnesses
that $\adm_r(G,\pi)\geq k$. 
As $\pi$ was arbitrary, we conclude
that $\adm_r(G)\geq k$.

Vice versa, assume that $k$ is the largest integer such 
that there exists a $(k,r)$-fan set in~$G$. We consider the 
following iterative procedure to construct an order $\pi=(v_1,\ldots,v_n)$ of $V(G)$ with $\adm_r(G,\pi)\leq k$. 
For a set $S\subseteq V(G)$ and a vertex $v\in S$, denote
by $b_r(S,v)$ the maximum size of a depth-$r$ $v-S$ fan. 

We initialize $S_n\coloneqq V(G)$. Now, for $i=n,\ldots, 1$,
if $v_{i+1},\ldots, v_n$ have already been ordered, choose 
$v_{i}$ such   that if $S_i=\{v_{1},\ldots, v_i\}$, then 
$p_i\coloneqq b_r(S_i,v_{i})$ is minimal. 
Clearly, the $r$-admissibility of the resulting order is
the maximum of the values $p_i$ occurring in its 
construction. Furthermore, observe that the set
$S_i$ is a $(p_i,r)$-fan set. Since~$k$ is the largest 
number such that there exists a $(k,r)$-fan set in $G$, 
we conclude that $\adm_r(G,\pi)\leq k$. 
\end{proof}

The iterative procedure described in the proof of \Cref{thm:adm}
can be turned into an algorithm to compute the $r$-admissibility
in a straightforward manner. A problematic step is to compute
the values $b_r(S_i,v_i)$, since in general graphs
it is NP-complete to determine the value $b_r(S,v)$ for a vertex
$v$ and set $S$, for values of $r$ greater than~$4$, as shown by Alon et al.~\cite{itai1982complexity}. However, Dvo\v{r}\'ak presents a simple approach
to approximate the values $b_r(S,v)$, and hence the $r$-admissibility, up to a factor of $r$. This 
approach can be easily implemented and runs in time $\Oof(mn^3)$. 
Using a sophisticated data structure presented by Dvo\v{r}\'ak, Kr\'al' and Thomas~\cite{dvovrak2013testing}, 
Dvo\v{r}\'ak showed how to compute the exact 
values in fixed-parameter linear time on every class of bounded
expansion~\cite{dvovrak2013constant}. 

\begin{theorem}[\cite{dvovrak2013constant}]\label{thm:dvorak-adm}
 Let $\Cc$ be a class of bounded expansion and let $r$ be a positive
 integer. There exists a function $f$ and an algorithm that given 
 an $n$-vertex graph $G\in \Cc$ computes an order $\pi$ with 
 $\adm_r(G,\pi)=\adm_r(G)$ in time $f(r)\cdot n$. 
\end{theorem}

In view of \Cref{lem:adm-col-wcol}, the algorithm of 
\Cref{thm:dvorak-adm} also gives linear time computable 
approximations for
the strong $r$-coloring numbers and weak $r$-coloring numbers. 
It was shown by Grohe et al.~\cite{grohe2015colouring} that for $r\geq 3$, computing exact values for 
$\wcol_r(G)$ is NP-complete. Similarly, it is shown by Dvo\v{r}\'ak~\cite{dvovrak2007asymptotical} that deciding whether a graph contains a depth-$1$ minor of density at least $d$ is NP-complete for $d\geq 2$ and by Muzi et al.~\cite{MuziORS17} that deciding whether a graph contains a $1$-subdivision of density at least $d$ for $d\geq 2$ is also NP-complete.
It was recently proved by Breen{-}McKay et al.~\cite{BreenMcKayLS25} that determining if a
graph has weak or strong $r$-coloring number at most $k$ is para-NP-hard when parameterized by $k$
for all $r\geq 2$. 
This answers an earlier question by the author and under standard complexity assumptions implies that there does not exist an algorithm that, given an $n$-vertex graph $G$, computes $\col_r(G)$ and $\wcol_r(G)$ in time $n^{f(r,\col_r(G))}$, 
or $n^{f(r,\wcol_r(G))}$, for any computable function~$f$. 
Moreover, they showed that there exists a constant $c$ such that it is NP-hard to
approximate the generalized coloring numbers within a factor of $c$.

How can it possibly be hard to compute the values $b_r(S_i,v_i)$? 
According to Menger's theorem the number of vertex disjoint paths between two vertices $s$ and $t$ is equal to the minimum cardinality of an $s$-$t$-cut, that is, of a set of vertices in $V(G)\setminus \{s,t\}$ hitting all paths between $s$ and $t$ (and such cuts can be efficiently computed). 
However, this fact is not true for
paths of bounded length as observed by Ad{\'a}mek and Koubek~\cite{adamek1971remarks}: a hitting set for all $s$-$t$-paths of length at most $r$ can be larger by factor $r$ than the number of disjoint $s$-$t$ paths of length at most $r$. 

\medskip
We now turn our attention to upper bounds for the 
generalized coloring numbers. The connection between 
the generalized coloring numbers and the expansion of 
a graph was first made by Zhu~\cite{zhu2009colouring}.

\begin{theorem}[Lemma 3.4 of~\cite{zhu2009colouring}] 
Let $G$ be a graph and let $r$ be a positive integer.  Then 
\begin{equation*}
  \col_r(G)\leq 1+q_r,
\end{equation*}
where $q_1=2\nabla_{\left\lceil (r-1)/2\right\rceil}(G)$, and for $i\geq 1$, $q_{i+1}$ is
inductively defined as $q_{i+1}=q_1\cdot q_i^{2i^2}$.
\end{theorem}

The currently
best known bounds for $\adm_r(G), \col_r(G)$ and $\wcol_r(G)$ follow from
the following theorem presented by Grohe et al.~\cite{grohe2015colouring}
combined with \cref{lem:adm-col-wcol}.

\begin{theorem}[\cite{grohe2015colouring}]\label{thm:adm-bound}
  Let $G$ be a graph and let $r$ be a positive integer. Then
  $G$ does not contain a $(k,r)$-fan set for $k>6r(\widetilde\nabla_{r-1}(G))^3$. Consequently, 

  \[\adm_r(G)\leq 6r \cdot \widetilde\nabla_{r-1}(G)^3,\]
  and
    \[\col_r(G)\leq \wcol_r(G)\leq (6r)^r\cdot \widetilde\nabla_{r-1}(G)^{3r}.\]
\end{theorem}

\begin{proof}
Let $G$ be a graph with $\widetilde\nabla_{r-1}(G)\le c$, and let
$k\coloneqq6rc^3+1$. Assume towards a contradiction that
$\adm_r(G)>k$. By \Cref{thm:adm} there exists a $(k+1,r)$-fan set $A$. Let $s\coloneqq|A|$.
For $v\in S$ let~$\mathcal{P}_v$ be a set of at least $k$ paths of length at most $r$ between $v$ and $S\setminus \{v\}$ such that 
$V(P)\cap V(Q)=\{v\}$ for $P\neq Q\in \mathcal{P}_v$ (we start with at least $k+1$ paths by definition of a $(k+1,r)$-fan and by dropping the path of length $0$ we are left with at least $k$ paths with one endpoint in $A\setminus \{v\}$).

Choose a maximal set $\mathcal{P}$ of pairwise internally
vertex-disjoint paths of length at most $2r-1$ connecting pairs of
distinct vertices from $A$ whose internal vertices belong to
$V(G)\setminus A$ such that each such pair is connected by at
most one path from $\mathcal{P}$. Let $H$ be the graph with vertex set $A$ and edges between all vertices $v,w\in A$ connected by some path in $\mathcal P$. Then $H\minor_{r-1}^{top}G$ and hence
$|\mathcal{P}|=|E(H)|\leq s\cdot c$.  Let~$M$ be the set of all
internal vertices of the paths in $\mathcal{P}$, and let
$m\coloneqq|M|$. Then $m\leq s\cdot c\cdot (2r-2)$.


By \Cref{lem:chromatic} $H$ is $(2c+1)$-colorable and hence contains an independent set $R$ of size at least $\left\lceil s/(2c+1)\right\rceil$.
For $v\in R$, let $\mathcal Q_v$ be the set of initial
segments of paths in $\mathcal{P}_v$ between $v$ and a vertex in
$(M\cup A)\setminus\{v\}$ with all internal vertices in
$V(G)\setminus(M\cup A)$. Observe that for $u,v\in R$ the paths in
$\mathcal Q_v$ and $\mathcal Q_u$ are internally vertex disjoint: 
if $P\in\mathcal Q_u$ and $Q\in\mathcal Q_v$ had an internal vertex
in common, then $P\cup Q$ would contain a path of length at most
$2r-2$ that is internally disjoint from all paths in $\mathcal P$,
contradicting the maximality of $\mathcal P$.

Let $G'$ be the subgraph of $G$ induced by the vertices of the paths in $\mathcal{P}$ and the paths in $\mathcal Q_v$ for $v\in R$. Let $H'$ be the graph obtained from $G'$ by
contracting for each $v\in R$ all paths in $\mathcal Q_v$ to single
edges. Then $H'\preceq^t_{r-1}G$. 
We have
$|V(H')|\leq s+m\leq s+s\cdot c\cdot(2r-2)\leq s\cdot c\cdot (2r-1)$
and at least
$|E(H')|\geq\left\lceil s/(2c+1)\right\rceil\cdot k$
edges. Thus, the edge density of $H'$ is
\begin{align*}
\frac{|E(H')|}{|V(H')|} \geq \frac{\left\lceil s/(2c+1)\right\rceil\cdot k}{s\cdot c\cdot (2r-1)}\geq \frac{s\cdot k}{s\cdot c\cdot (2c+1)\cdot (2r-1)}\geq \frac{6rc^3+1}{(2c^2+c)\cdot 2r}\geq \frac{6rc^3+1}{3c^2\cdot 2r}>c.
\end{align*}
A contradiction, showing that $k\leq 6rc^3$. 
\end{proof}

\begin{corollary}\label{lem:be-wcol}
 A class $\Cc$ of graphs has bounded expansion if and only
 if there exist functions~$a,c$ and $w$ such that for every positive
 integer $r$ and every $G\in \Cc$ we have
 $\adm_r(G)\leq a(r)$, $\col_r(G)\leq c(r)$ and $\wcol_r(G)\leq w(r)$.  
\end{corollary}

Observe that the function bounding the $r$-admissibility
in \Cref{thm:adm} is polynomial in $r$ and~$\widetilde\nabla_r(G)$. 
One may wonder if a similar relation holds between $\col_r(G)$ and $\nabla_r(G)$. The following problem is attributed to 
Joret and Wood in~\cite{esperet2018polynomial}.

\begin{problem}\label[problem]{prob:pol-exp}
Does there exist a polynomial $p(x,y)$ such that for all
graphs $G$ and all positive integers~$r$ we have 
$\col_r(G)\leq p(r,\nabla_r(G))$?
\end{problem}

The problem appears already in bounded degree graphs. 

\begin{problem}\label[problem]{prob:pol-exp2}
Does there exist a polynomial $p(x,y,z)$ such that for all
graphs $G$ with maximum degree $\Delta$ 
and all positive integers $r$ we have 
$\col_r(G)\leq p(r,\nabla_r(G),\Delta)$?
\end{problem}

\subsection{Universal orders}\label{sec:uniform-orders}

According to \Cref{lem:be-wcol}, if $\Cc$ has bounded
expansion, then there exists a function $f$ such that 
for every value of $r$ and every $G\in \Cc$ there exists 
an order $\pi$ of $V(G)$ such that $\adm_r(G,\pi)\leq f(r)$.
A universal order is an order that works for every value
of $r$ at the same time. 

\begin{definition}
 Let $G$ be a graph and let $f$ be a function. An order 
 $\pi$ of $V(G)$ is a \emph{universal order} bounded by
 $f$ if $\adm_r(G,\pi)\leq f(r)$ for every non-negative integer $r$. 
\end{definition}

By \cref{lem:adm-col-wcol} we could equivalently ask that
$\col_r(G,\pi)\leq f(r)$ or $\wcol_r(G,\pi)\leq f(r)$ in the
above definition. It was shown by Van den Heuvel and 
Kierstead~\cite{vdH18} that all classes of bounded expansion admit
universal orders. In this section, we provide an alternative proof. 

For this, we consider a variant of 
the iterative procedure presented in the proof of \Cref{thm:adm}
to construct the order $\pi=(v_1,\ldots, v_n)$. In the 
procedure of \Cref{thm:adm}, \Cref{thm:adm-bound} guarantees
that in step $i$ we will always find a vertex $v_i$
with back-connectivity $b_r(S_i,v_i)\leq 
6r(\widetilde\nabla_r(G))^3$.  
We are going to prove that we do not only find one vertex with 
small back-connectivity, but in fact, a large fraction of the
vertices have a
back-connectivity that is only slightly larger than the optimal
back-connectivity. For every value of $r$, there will hence be
only a few vertices violating the $r$-admissibility constraint, 
so that we are still left with at least one vertex that can 
be inserted as $v_i$ without violating the constraint for 
any value of $r$. The proof of the following lemma is 
a variation of the proof of \Cref{thm:adm-bound} (only the numbers need to be adapted), which is presented
by Pilipczuk et al.~\cite{pilipczuk2018parameterizedarxiv} (the arXiv version of~\cite{PilipczukST18}). 

\smallskip

\begin{lemma}[\cite{pilipczuk2018parameterizedarxiv}, Lemma 14]\label{lem:adm-half}
Let $G$ be a graph and $S\subseteq V(G)$. Let $\varepsilon>0$
and let $r$ be a positive integer. Then at most $\varepsilon\cdot |S|$ vertices $v\in S$ have $b_r(S,v)>6\cdot\varepsilon^{-1}\cdot r\cdot \bigl(\widetilde\nabla_{r-1}(G)\bigr)^3$.
\end{lemma}

We are ready to prove the existence of universal orders. 

\begin{theorem}
 Let $G$ be a graph and let 
 $f(r)\coloneqq 2^{r+1}\cdot 6r\cdot \bigr(\widetilde\nabla_{r-1}(G)\bigl)^3$.
 Then $G$ admits an order $\pi$ with 
 $\adm_r(G,\pi)\leq f(r)$ for every positive integer $r$. 
\end{theorem}
\begin{proof}
We adapt the iterative procedure presented in the proof of \Cref{thm:adm} to an order $\pi=(v_1,\ldots,v_n)$ of~$V(G)$ 
with $\adm_r(G,\pi)\leq f(r)$ for
all positive integers $r$. 

We initialize $S_n\coloneqq V(G)$. Now, for $i=n,\ldots, 1$,
if $v_{i+1},\ldots, v_n$ have already been ordered, choose 
$v_{i}$ such that if $S_i=\{v_{1},\ldots, v_i\}$, then 
$b_r(S_i,v_i)\leq f(r)$ for all positive integers $r$. Clearly, 
if such a choice for $v_i$ is possible, which remains to be shown, then $\pi$ has the 
desired properties. 

According to \Cref{lem:adm-half}, for any set $S$ and any 
positive integer $r$, 
at most $1/(2^{r+1})\cdot |S|$ vertices have back-connectivity
$b_r(S,v)>f(r)$. Hence, at most $|S|/2$ vertices 
conflict~$f$ for the value $r=0$, at most $|S|/4$ vertices
conflict~$f$ for the value $r=1$, and so on. Hence, at most 
\[\sum_{r=0}^n\frac{1}{2^{r+1}}\cdot |S|<|S|,\]
vertices of $S$ are conflicting $f$ for some value of $r$. We 
conclude that among all vertices of $S$, we find 
at least one vertex $v_i$ satisfying
$b_r(S_i,v_i)\leq f(r)$ for all values
of $r$.
\end{proof}

\begin{corollary}
Let $\Cc$ be a class of bounded expansion. Then there exists
a function~$f$ such that for all $G\in\Cc$ there exists an 
order $\pi$ with $\adm_r(G,\pi)\leq f(r)$ for every positive
integer $r$. 
\end{corollary}

We can use the approach of Dvo\v{r}\'ak to approximate the 
values $b_r(S,v)$ up to a factor of $r$ and thereby, for every fixed class $\Cc$ of bounded
expansion obtain an algorithm that computes a universal order
in polynomial time.

\begin{corollary}
  Let $G$ be a graph and let 
 $f(r)\coloneqq 2^{r+1}\cdot 6r^2\cdot \bigr(\widetilde\nabla_{r-1}(G)\bigl)^3$.
 Then we can compute in polynomial time an order $\pi$ of $G$ with 
 $\adm_r(G,\pi)\leq f(r)$ for every positive integer $r$. 
\end{corollary}



\subsection{Transitive fraternal augmentations}

We now come to a second approach to approximate the generalized
coloring numbers, which is based on Ne{\v s}et{\v r}il and Ossona 
de Mendez's transitive fraternal augmentation technique, see
Chapter 7.4 of~\cite{nevsetril2012sparsity}. In the following, for clarity, we will call directed edges $(u,v)$ \emph{arcs}. 

\begin{definition}
Let $G$ be a graph and let $r$ be a positive integer. An 
\emph{$r$-fraternity function} is a function $w\colon V(G)\times 
V(G)\rightarrow \N\cup \{\infty\}$ such that for all distinct $u,v\in V(G)$
at least one of $w(u,v)$ and $w(v,u)$ is $\infty$, such 
that for all $uv\in E(G)$ either $w(u,v)=1$ or $w(v,u)=1$,
and such that for all distinct $u,v\in V(G)$
\begin{itemize}
\item either $\min \{w(u,v), w(v,u)\}=1$, 
\item or $\min \{w(u,v), w(v,u)\}=
\min_{z\in V\setminus\{u,v\}} (w(z,u)+w(z,v))$, 
\item or $\min \{w(u,v), w(v,u)\}>r$ and 
$\min_{z\in V\setminus\{u,v\}} (w(z,u)+w(z,v))>r$.  
\end{itemize}

If $w$ is an $r$-fraternity function for every value of $r$, 
then $w$ is called a \emph{fraternity function} (see Figure 12). 
\end{definition}

\begin{figure}[ht!]
  \begin{center}
  \includegraphics[width=0.5\textwidth]{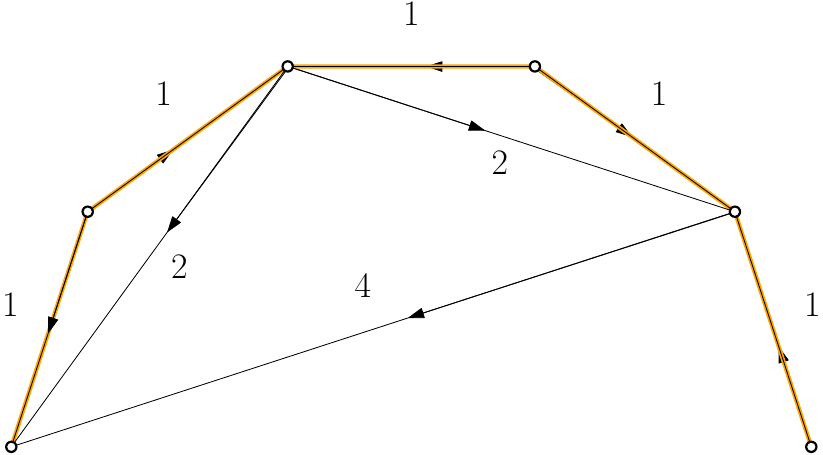}
  \caption{A fraternity function on a path $P_6$. All missing pairs are assigned value $\infty$.}
\end{center}
\end{figure}

\begin{definition}
Let $G$ be a graph, let $r$ be a positive integer and let 
$w$ be an $r$-fraternity function for~$G$. 
For $1\leq i\leq r$, the directed graph $\vec{G}^w_{\leq i}$
on the same vertex set as $G$ containing all arcs 
$(u,v)$ with $w(u,v)\leq i$  is called an 
\emph{$i$-fraternal augmentation} of~$G$. 
\end{definition}

Obviously we have $\vec{G}^w_{\leq 1}\subseteq \vec{G}^w_{\leq 2}\subseteq \ldots\subseteq \vec{G}^w_{\leq r}$. The arcs 
of $\vec{G}^w_{\leq r}$ that are not arcs of~$\vec{G}^w_{\leq 1}$
are called \emph{fraternal arcs}. 

We remark that fraternal augmentations 
in~\cite{nevsetril2012sparsity}  
are defined so that the fraternal arcs in the 
second and third item above are introduced between vertices 
that have a common 
in-neighbor, that is, in that work, it is required that 
$\min \{w(u,v), w(v,u)\}=
\min_{z\in V\setminus\{u,v\}} (w(u,z)+w(v,z))$, 
or if $\min \{w(u,v), w(v,u)\}>r$ then
$\min_{z\in V\setminus\{u,v\}} (w(u,z)+w(v,z))>r$.  
We prefer our definition, as it highlights the connection
to the generalized coloring numbers. Let us make this connection 
explicit. 

\begin{lemma}\label{lem:col-frat}
Let $G$ be a graph, let $r$ be a positive integer and let 
$\pi$ be a linear order of~$V(G)$. Then $w\colon V(G)\times
V(G)\rightarrow \N\cup \{\infty\}$ with 
\[w(u,v)=\begin{cases}
i & \text{if $i\leq r$ and $v\in \SReach_i[G,\pi,u]\setminus 
\SReach_{i-1}[G,\pi,u]$}\\
\infty & \text{if $v\not\in \SReach_r[G,\pi,u]$}. 
\end{cases} 
\]
is an $r$-fraternity function. 
\end{lemma}
\begin{proof}
It is clear that for all distinct $u,v\in V(G)$
at least one of $w(u,v)$ and $w(v,u)$ is $\infty$ and
that for all $uv\in E(G)$ either $w(u,v)=1$ or $w(v,u)=1$.
Now assume $w(u,v)=i$ for some $i\geq 2$, hence, 
$v\in \SReach_i[G,\pi,u]\setminus 
\SReach_{i-1}[G,\pi,u]$. Then there exists a path $P$ of 
length $i$ between~$u$ and~$v$ such that $v$ is 
minimum among $V(P)$ with respect to $\pi$ and such 
that all internal vertices of $P$ are larger than $u$. 
Let $z$ be the smallest vertex of $V(P)\setminus\{u,v\}$ and let 
$\ell_1$ be the distance between~$z$ and~$v$ on $P$ and let $\ell_2$ be the distance between~$z$ and~$u$ on $P$. Then $v\in \SReach_{\ell_1}[G,\pi,z]$ 
and $u\in \SReach_{\ell_2}[G,\pi,z]$. Furthermore, there 
does not exist $z'$ and $d_1,d_2$ with $d_1+d_2<i$
such that $v\in \SReach_{d_1}[G,\pi,z']$ 
and $u\in \SReach_{d_2}[G,\pi,z']$, as otherwise, 
we could concatenate the paths witnessing this to find a 
path of length $i-1$ between $u$ and~$v$ such that $v$ is 
minimum among~$V(P)$ and such 
that all internal vertices of $P$ are larger than $u$. Such a 
path does not exist as $v\not\in\SReach_{i-1}[G,\pi,u]$. 
In particular, $v\not\in \SReach_{\ell_1-1}[G,\pi,z]$ 
and $u\not\in \SReach_{\ell_2-1}[G,\pi,z]$, hence, 
$w(z,u)=\ell_1$ and $w(z,v)=\ell_2$ and there is no
$z'$ with $w(z',u)+w(z',v)<i$. 
\end{proof}

Using \Cref{thm:adm-bound}, we get the following corollary. 

\begin{corollary}
Let $G$ be a graph and let $r$ be a positive integer. Then 
there exists an $r$-fraternity function $w$ such that 
$G^w_{\leq r}$ has maximum out-degree at most 
$(6r)^r\cdot\widetilde\nabla_{r-1}(G)^{3r}$.
\end{corollary}

%

Vice versa, we can turn fraternal augmentations 
into weak coloring orders as follows. 
We need the following property of fraternal augmentations 
that was (in variations) noted several times~\cite{dvovrak2013testing, grohe2017deciding,nevsetril2012sparsity,Reidl16,
siebertz2016nowhere}. 

\begin{lemma}\label{thm:augmentationneighborhood}
  Let~$G$ be a graph, let~$r$ be a positive integer and 
  let $w$ be an $r$-fraternity function of~$G$.  
  Let~$u,v\in V(G)$
  and assume that \mbox{$u=v_0, v_1,\ldots, v_\ell=v$} is a path of
  length $\ell\leq r$ between~$u$ and~$v$ in~$G$. Then 
  one of the following holds. 
  \begin{enumerate}
	\item There are indices $0< i_1<\ldots<i_k< \ell$ 
	such that $u,v_{i_1}, \ldots, v_{i_k},v$ is a directed 
	path from $u$ to $v$ in $G^w_{\leq r}$ and for all 
  $0\leq i_m<i_n\leq \ell$ we have $w(v_{i_m}, v_{i_n})\leq 
  i_n-i_m$.
    \item There are indices $0< i_1<\ldots<i_k< \ell$ 
	such that $v,v_{i_k}, \ldots, v_{i_1},u$ is a directed 
	path from $v$ to $u$ in $G^w_{\leq r}$ and for all 
  $0\leq i_m<i_n\leq \ell$ we have $w(v_{i_n}, v_{i_m})\leq 
  i_n-i_m$.
    \item There are indices $0< i_1<\ldots<i_k< \ell$ and 
    an index $i_j\in \{i_1,\ldots, i_k\}$ such that 
    $u, v_{i_1},\ldots, v_{i_j}$ and $v, v_{i_k},\ldots, v_{i_j}$
    are directed paths in $G^w_{\leq r}$ from $u$ to $v_{i_j}$ and from 
    $v$ to $v_{i_j}$, respectively. Furthermore, for all 
  $0\leq i_m<i_n\leq i_j$ we have $w(v_{i_m}, v_{i_n})\leq 
  i_n-i_m$ and for all 
  $i_j\leq i_m<i_n\leq \ell$ we have $w(v_{i_n}, v_{i_m})\leq 
  i_n-i_m$.
  \end{enumerate}
\end{lemma}
\begin{proof}
We prove the statement by induction on $\ell$. For $\ell=1$, 
we have $(u,v)\in E(\vec{G}^w_{\leq 1})$ if and only if 
$w(u,v)=1$, hence, one of the first two items is true. 
Now assume that 
$\ell>1$ and assume the statements hold for 
$v_0,v_1,\ldots, v_{\ell-1}$. Assume we are in the first case and 
there are indices $0< i_1<\ldots<i_k< \ell-1$ 
such that $u,v_{i_1}, \ldots, v_{i_k},v_{\ell-1}$ is a directed 
path from $u$ to $v_{\ell-1}$ in $G^w_{\leq r}$ and that 
for all $0\leq i_m<i_n\leq \ell-1$ we have 
$w(v_{i_m}, v_{i_n})\leq i_n-i_m$. 
We distinguish
two cases. If $w(v_{\ell-1},v)=1$, $u,v_{i_1}, \ldots, 
v_{i_k},v_{\ell-1},v$ is a directed path in $G^w_r$, hence the
first item holds for $u=v_0,v_1,\ldots, v_\ell=v$. If $w(v,v_{\ell-1})=1$, then we set~$v_j$ in 
the third item to $v_{\ell-1}$ and the third item holds for 
$u=v_0,v_1,\ldots, v_\ell=v$. The statements about the 
weight function follow directly from the hypothesis. 

Now assume we are in the second case and there are indices 
$0< i_1<\ldots<i_k< \ell-1$ 	such that $v_{\ell-1},v_{i_k}, 
\ldots, v_{i_1},u$ is a directed path from $v_{\ell-1}$ to $u$ 
in $G^w_{\leq r}$ and that for all $0\leq i_m<i_n\leq \ell-1$ we have 
$w(v_{i_n},v_{i_m})\leq i_n-
i_m$. Again we distinguish two cases. If $w(v,v_{\ell-1})=1$, then 
the first case holds for \mbox{$u=v_0,v_1,\ldots, v_\ell=v$}. The 
additional statement about the weight function follows directly
from the hypothesis. If $w(v_{\ell-1},v)=1$, then we have 
$w(v,v_{i_k})\leq \ell-i_k$ or $w(v_{i_k},v)\leq \ell-i_k$, 
as~$w$ is a fraternity function and by assumption, we have 
$w(v_{\ell-1},v_{i_k})\leq \ell-i_k-1$. In the former case 
we have a directed path $v,v_{i_k}, 
\ldots, v_{i_1},u$ from $v$ to $u$ in $G^w_{\leq r}$, that is, 
the second case holds for $u=v_0,v_1,\ldots, v_\ell=v$. 
In the latter case, we apply the argument again to the sequence 
$v,v_{i_k}, \ldots, v_{i_1},u$ ($v_{i_k}$ taking the role of 
$v_{\ell-1}$). Again, we either get a directed path $v,v_{i_{k-1}}, 
\ldots, v_{i_1},u$ from $v$ to $u$ in $G^w_{\leq r}$, or
we apply the argument again. In some step, we either get
a directed path $v,v_{i_{m}}, \ldots, v_{i_1},u$ from 
$v$ to $u$ in $G^w_{\leq r}$ for some $1\leq i\leq k$, or 
a directed arc $(u,v)$ in~$G^w_{\leq r}$. The additional property on 
the weight function is easily verified in each step. 

The argument in the third case is analogous to the argument 
in the second case. 
\end{proof}

Compare \cref{thm:augmentationneighborhood} with 
\cref{lem:wcol-sep}, which describes the local separation properties 
of the weak coloring numbers. We now use 
\cref{thm:augmentationneighborhood} to approximate~$\wcol_r(G)$ by~$r$-fraternal augmentations as follows.

\begin{lemma}[adaptation of Lemma 6.7 of \cite{grohe2017deciding}]\label{thm:approximatingcol}
  Let~$G$ be a graph, let~$r$ be a positive integer and let~$w$ be an $r$-fraternity function of $G$. Assume that 
  $G^w_{\leq r}$ has maximum out-degree $\Delta$ and let 
  $d=r\cdot \Delta^{r}$. Then~$\wcol_r(G)\leq 4d^2$.
\end{lemma}
\begin{proof}
Let $\vec{H}$ be the directed graph that we obtain by 
adding all edges $(u,v)$ such that there is a directed path 
of (weighted) length at most $r$ from $u$ to $v$ in 
$G^w_{\leq r}$. As~$G^w_{\leq r}$ has maximum out-degree~$\Delta$, $\vec{H}$ has out-degree at most $\sum_{i=1}^r\Delta^i \leq
r\cdot \Delta^r=d$. 
Furthermore, \cref{thm:augmentationneighborhood} implies that
for all paths $u=v_0, v_1,\ldots, v_\ell=v$ of length $\ell\leq r$
 between~$u$ and~$v$ in~$G$, either
  \begin{enumerate}
	\item $(u,v)\in E(\vec H)$, or 
    \item $(v,u)\in E(\vec H)$, or 
    \item there is an index $0<i<\ell$ such that 
    $(u,v_i), (v,v_i)\in E(\vec{H})$. 
  \end{enumerate}

As $\vec{H}$ has out-degree at most $d$, the underlying undirected
graph $H$ is~$2d$-degenerate and there exists an order $\pi$ of 
$V(H)$ such that each vertex has at most~$2d$ smaller neighbors. 
We claim that $\wcol_r(G,\pi)\leq 4d^2$. 

To see this, we count for each vertex~$u\in V(G)$ the number of
end-vertices of paths of length at most~$r$ from~$u$ such that 
the end-vertex is the smallest vertex of the path. This number
bounds~$|\WReach_r[G,\pi, u)]|$.

By our above observation, for each such path $P$
with end-vertex~$v$, we either have an edge~$(u,v)$ or an
edge~$(v,u)$ or there is~$z$ on the path and we have arcs~$(u,z),
(v,z)$ in~$\vec{H}$.

By construction of the order, there are at most~$2d$ 
edges~$(u,v)$ or~$(v,u)$ such that $v<_\pi u$. Furthermore, 
we have at most~$d$ arcs~$(u,z)$ in $\vec{H}$, as~$u$ has 
out-degree at most~$d$. For each such~$z$ there are at 
most~$2d$ edges~$(v,z)$ such that~$v<_\pi z$ by 
construction of the order. As $v\in \WReach_r[G,\pi, u)]$ 
requires that $v$ is the minimum vertex on $P$, these are the 
only pairs we have to count. Hence, in total we 
have~$|\WReach_r[G,<,u]|\leq 2d+2d^2\leq 4d^2$.
\end{proof}

We can hence use \cref{thm:approximatingcol} to find 
orders with good weak coloring properties, if we are able
to find good $r$-fraternity functions (without applying 
\cref{lem:col-frat} of course). 
As shown by Ne\v{s}et\v{r}il and Ossona de Mendez~\cite{nevsetvril2008grad}, a
simple greedy procedure yields good results. 

Let $G$ be an undirected graph and let $\vec G_1$ be any orientation of~$G$. We construct a sequence $\vec G_1 \subseteq \vec G_2 \subseteq \dots\subseteq G_r$ of
directed graphs together with an $r$-fraternity function 
$w\colon V(G)\times V(G)\rightarrow\N\cup\{\infty\}$. 
We let $w(u,v)=1$ if $(u,v)\in E(\vec G_1)$ and $w(u,v)=\infty$
if $(v,u)\in E(\vec{G}_1)$. Now, assume~$\vec{G}_d$ for 
$d\geq 1$ has been constructed and assume that $w$ is defined
for all pairs $u,v$ with either $(u,v)$ or $(v,u)\in E(\vec{G}_d)$
and undefined for all other pairs.  
We obtain the edge set $E(\vec{G}_{d+1})$ by adding the
following edges to $E(\vec G_d)$ as follows. 

\begin{itemize}
    \item If $u,v,z\in V(G)$ are such that $(z,u) \in E(\vec G_i)$ 
    and $(z,v) \in E(\vec G_j)$ and $i+j=d+1$, and neither 
    $(u,v)$ nor $(v,u)\in E(\vec G_d)$, then we introduce exactly 
    one of $(u,v) \in E(\vec G_{d+1})$ or $(v,u) \in E(\vec G_{d+1})$.
\end{itemize}

We define $w(u,v)=d+1$ if $(u,v) \in E(\vec G_{d+1})$ and 
$w(u,v)=\infty$ if $(v,u) \in E(\vec G_{d+1})$. Finally, 
we define $w(u,v)=\infty$ for all $u,v$ with $(u,v), (v,u)\not\in 
E(\vec{G}_r)$. 

\begin{observation}
The function $w$ as defined above is an $r$-fraternity function. 
\end{observation}
\begin{proof}
It is immediately from the construction that for all $u,v\in V(G)$
at least one of $w(u,v)$ and $w(v,u)$ is $\infty$. It remains
to show that for all distinct $u,v\in V(G)$ 

\begin{itemize}
\item either $\min \{w(u,v), w(v,u)\}=1$, 
\item or $\min \{w(u,v), w(v,u)\}=
\min_{z\in V\setminus\{u,v\}} (w(z,u)+w(z,v))$, 
\item or $\min \{w(u,v), w(v,u)\}>r$ and 
$\min_{z\in V\setminus\{u,v\}} (w(z,u)+w(z,v))>r$.  
\end{itemize}

We prove the statement for all $u,v$ with $\min \{w(u,v), w(v,u)\}\leq d\leq r$, by induction on~$d$. 
It is immediate that for all $uv\in E(G)$ either $w(u,v)=1$ 
or $w(v,u)=1$ and that there are no other pairs for which 
$w$ takes the value $1$. Now assume the statement holds for
$1\leq d<r$ and fix some pair $u,v$ with $\min\{w(u,v), w(v,u)\}
=d+1$, say $w(u,v)=d+1$. Then $(u,v)\in E(\vec G_{d+1})\setminus
E(\vec G_d)$. Hence, there is $z\in V(G)$ and $i,j$ with 
$i+j=d+1$ such that $(z,u)\in E(\vec{G}_i)$ and 
$(z,v)\in E(\vec G_j)$. Furthermore, there is no $z\in V(G)$ 
and $i,j$ with 
$i+j<d+1$ such that $(z,u)\in E(\vec{G}_i)$ and 
$(z,v)\in E(\vec G_j)$, as otherwise by construction, we would
have $(u,v)$ or $(v,u)\in\vec G_d$. Hence, by induction hypothesis,
$\min \{w(u,v), w(v,u)\}=
\min_{z\in V\setminus\{u,v\}} (w(z,u)+w(z,v))$. With the 
same argument, we now get that $\min \{w(u,v), w(v,u)\}>r$ implies 
$\min_{z\in V\setminus\{u,v\}} (w(z,u)+w(z,v))>r$.  
\end{proof}

It remains to specify how we obtain the orientation $\vec G_1$ 
and how we decide whether we introduce the 
edge $(u,v)$ or $(v,u)$ in each augmentation step. 
We can choose the orientation $\vec G_1$ to be the
acyclic ordering derived from the degeneracy ordering of~$G$. 
This is optimal up to factor $2$ if we aim to minimize 
the maximum out-degree $\Delta^+(\vec G_1)$ 
of $\vec G_1$. Second, in any step $d+1$, we
can orient the fraternal edges added in step~$d+1$ by 
first collecting \emph{all} potential
fraternal edges in an auxiliary undirected graph~$H_{d+1}$ and 
then again compute an acyclic orientation $\vec H_{d+1}$ 
that is optimal up to factor $2$. We then insert the arcs 
into~$\vec G_{d+1}$ according to their orientation
in~$\vec H_{d+1}$.

%
%
%
%
%
%
%

\medskip
It remains to prove that we obtain good bounds for the 
maximum out-degree of $\vec G_r$. The simplest way 
to establish such bounds is via \emph{shallow packings},
which we introduce next. 

\begin{definition}
Let $G$ be a graph and let $k,t$ be positive integers. A collection $\mathcal{F}$ of subsets of~$V(G)$ is a 
\emph{$(k,t)$-packing} if $G[F]$ is connected and has radius at most $t$ for 
every $F\in \mathcal{F}$, and every vertex appears in at most $k$ sets of $\mathcal{F}$.
The elements of $\mathcal{F}$ are called {\em{clusters}}.
The \emph{induced packing graph}~$G[\mathcal{F}]$ has $\mathcal{F}$ as the set of
vertices and two clusters $F,F'\in \mathcal{F}$ are connected  by an edge 
if they share a vertex or there are vertices $u\in F$ and $v\in F'$ 
such that $uv\in E(G)$. 
\end{definition}

\begin{definition}
Let $G,H$ be graphs. The \emph{lexicographic product} of 
$G$ with $H$, denoted $G\bullet H$ is defined by 
\begin{align*}
V(G\bullet H) & = V(G)\times V(H)\\
E(G\bullet H) & = \{(u,x)(v,y) :
 uv\in E(G) \text{ or $u=v$ and $xy\in E(H)$}\}.
\end{align*}
\end{definition}

Observe that $G\bullet K_k$ corresponds to the packing graph
$G[\Ff]$, where $\Ff$ is the $(k,0)$-packing that contains 
$k$ copies of each vertex $v\in V(G)$.
Vice versa, if $\mathcal{F}$ is a $(k,t)$-packing, then the induced packing graph $G[\mathcal{F}]$ is a depth-$t$ minor of the lexicographic product $G\bullet K_k$.

\smallskip
Ne\v{s}et\v{r}il and Ossona de 
Mendez~\cite{nevsetril2012sparsity}
and Har-Peled and Quanrud~\cite{har2017approximation} proved that shallow packings of sparse 
graphs remain sparse. The following bounds are from~\cite{notes}.

\newcommand{\Ed}{E_{\mathrm{d}}}
\newcommand{\End}{E_{\mathrm{nd}}}

\begin{lemma}[\cite{notes}]\label{lem:packingdensity}
Let $G$ be a graph and let $\mathcal{F}$ be a $(k, t)$-packing of $G$. 
Then for every positive integer $r$, \[\nabla_r(G[\mathcal{F}])\leq \frac{k-1}{2}+(2(k-1)(2rt+r+t+1)+1)\cdot \nabla_{2rt+t+r}(G).\]
In particular, 
\[\nabla_r(G\bullet K_k)\leq \frac{k-1}{2}+(2(k-1)(r+1)+1)\cdot \nabla_{r}(G)\leq 2(k-1)(r+3)\cdot \nabla_r(G).\]
\end{lemma}

\medskip

%



Let us give some intuition how to now derive bounds for the fraternal augmentations. 
Let $\vec{G}_1\subseteq \vec{G}_2\subseteq \ldots\subseteq 
\vec{G}_d$ be defined as above and assume that $\vec{G}_i$ has
maximum out-degree $d_i$. Recall that the orientation $\vec G_1$ is obtained from the acyclic ordering derived from the degeneracy ordering of~$G$, hence, the out-degree $d_1$ is bounded.  
For $1<i\leq d$ and for each $v\in V(G)$
let $F^i_v\coloneqq \{u~:~(v,u)\in E(\vec G_i)\cup \{v\}$. 
Let $\Ff^i=\{F^i(v) : v\in V(G)\}$. For $k+\ell = d+1$ let 
$G^{k,\ell}$ be the packing graph $G[\Ff^k\cup \Ff^\ell]$. 
Observe that $\Ff^k\cup \Ff^\ell$ is a $(d+1, d_k+d_\ell)$-packing. 
Let $G_{d+1}=\bigcup_{k+\ell=d+1}G^{k,\ell}$. 
Then $G_{d+1}$ has
all fraternal augmentations of length $d+1$ as undirected edges. 
Now use \cref{lem:packingdensity} to derive that $\nabla_0(G_{d+1})$ is bounded and hence we can again find an orientation with small out-degree for the fraternal edges of $G_{d+1}$. 




%


A stronger analysis of the augmentation process gives the
following bounds. We remark that Reidl in~\cite{Reidl16} does not only consider fraternal but also transitive augmentations.  

\begin{theorem}[Theorem 16 of~\cite{Reidl16}, reformulated]
For every graph $G$ the above augmentation process yields fraternal 
augmentation sequence $\vec G_1, \vec G_2,\ldots$
such that 
\[\Delta^+(\vec{G}_{d+1})\leq 2^{8\cdot 5^d}d^{6^d}(d+1)^{2^d}
\cdot \big(\widetilde{\nabla}_{d+1}(G)\cdot \Delta^+(\vec G_1)\big)
^{3\cdot 5^{d-1}}.\] 
\end{theorem}


%

\subsection{Pointer-structures and quantifier elimination}

\Cref{thm:augmentationneighborhood} makes quite clear that transitive fraternal augmentations are just the weak coloring numbers in disguise, and most current research works directly with the generalized coloring numbers instead of with augmentations. 
Nevertheless, the augmentations have been very important, especially in the highly influential quantifier elimination schema for first-order logic on classes of bounded expansion of Dvo\v{r}\'ak, Kr\'al' and Thomas~\cite{dvovrak2013testing}. The recent formulation of results of this type is in terms of \emph{pointer structures} (a structure with a signature that consists of unary relation and unary function symbols, see \cref{sec:logic}) that are \emph{guarded by}~\cite{dvovrak2013testing}, \emph{founded in} the graph~\cite{grobler2021discrepancy} or \emph{guided by} the graph~\cite{nevsetvril2023modulo}, and have found many applications at the intersection of sparse and structurally sparse graphs and logic, see e.g.~\cite{dreier2020two,dreier2021approximate,dvovrak2013testing,grobler2021discrepancy,kazana2013enumeration,nevsetvril2023modulo,PilipczukST18,segoufin2017constant,vigny2018query}. 
Quantifier elimination over pointer structures in nowhere dense classes does not work as for bounded expansion classes as proved by Grobler et al.~\cite{grobler2021discrepancy}. We will revisit the concept of quantifier elimination in \cref{sec:ltc} when discussing low treedepth colorings.

\section{Tree-decomposable graphs}\label{sec:tree-decomposable}

In this section, we turn our attention to upper bounds for the generalized coloring numbers on special graph classes. Special 
focus is put on tree-decomposable graph classes. 
Let us define the used concepts. 

Let $\Tt=(T,\bag)$ be a rooted tree decomposition of $G$.
For a non-root node~$x$ with parent $y$, we define the \emph{adhesion} of $x$ as $\adh(x)=\bag(x)\cap \bag(y)$. If $x$ is the root, then we set $\adh(x)=\emptyset$ by convention.
The {\em{adhesion}} of the tree decomposition $\Tt=(T,\bag)$ is the maximum size of an adhesion set in $\Tt$, i.e., $\max_{x\in V(T)}|\adh(x)|$.
The \emph{torso} of a node $x$ is the graph $\Torso(x)$
on vertex set $\bag(x)$ where two vertices $u,v\in \bag(x)$ are
adjacent if and only if $uv\in E(G)$ or if there exists $y\neq x$
such that $u,v\in \bag(y)$. Equivalently, $\Torso(x)$ is obtained 
from $G[\bag(x)]$ by turning the adhesions of $x$ and of all 
children of $x$ into cliques. The \emph{margin} of a node $x$ is the set $\mrg(x)=\bag(x)\setminus \adh(x)$.   
We say that $\Tt$ is a 
tree decomposition \emph{over} a class
$\Dd$ of graphs if $\Torso(x)\in \Dd$ for every node $x$ of~$T$ (see Figure 13).

\begin{figure}[htt!]
\begin{center}
        \centering
        \includegraphics[width=.75\textwidth]{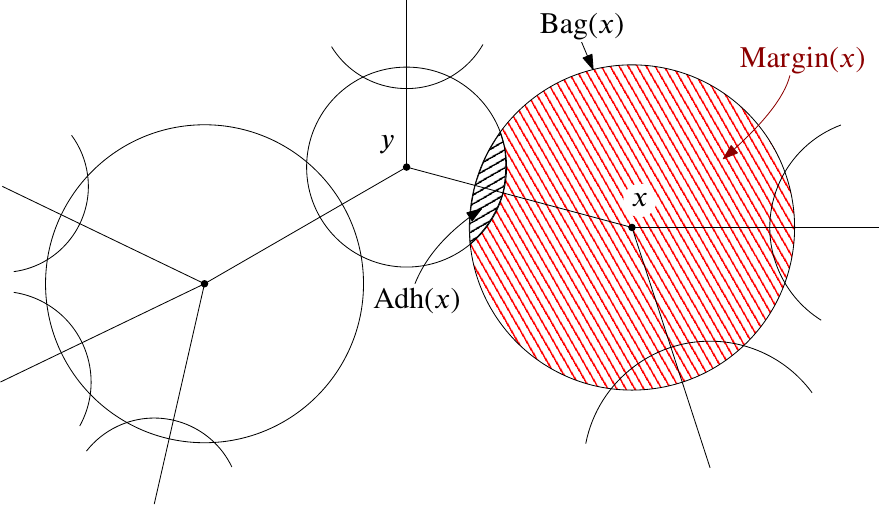}
        \caption{A tree decomposition of a graph $G$. The torso of $x$ is obtained from the graph induced by $\bag(x)$ by turning the adhesions into cliques. }
\end{center}
\end{figure}

We are going to use the separator properties of tree-decomposable
graphs described in \cref{lem:sep} to construct good generalized
coloring orders for classes 
$\Cc$ that admit tree decompositions of small adhesion over
other graph classes $\Dd$. The orders for graphs
from the class $\Dd$ can be naturally combined with the partial order induced by the tree decomposition to obtain good orders
for graphs from $\Cc$. 

Weak reachability for vertices from different bags of a decomposition is guided by the interaction through the adhesion sets. 
The base case capturing this situation, which also guides the general situation of tree decomposable graphs, is the case of bounded treewidth. 
It was proved by Grohe et al.~\cite{grohe2015colouring} that the weak $r$-coloring numbers are bounded polynomially 
in $r$ on every class of bounded treewidth.

\begin{theorem}[\cite{grohe2015colouring}]\label{thm:wcol-tw}
  Let $G$ be a graph of treewidth at most $k$ and 
  let $r$ be a positive integer. Then 
  $\wcol_r(G)\leq \binom{r+k}{k}$. 
\end{theorem}


By combining the partial order induced by the tree decomposition with the orders for graphs from a class $\Cc$, we then obtain the following result. 

\begin{theorem}\label{thm:decomposable-wcol}
  Let $\Cc$ be a class of graphs that is tree decomposable 
  with adhesion at most $k$ over a class $\Dd$. Let $G\in\Cc$
  and let $r$ be a positive integer. Then 
  \begin{itemize}
  \item $\adm_r(G)\leq k+\adm_r(\Dd)$, 
  \item $\col_r(G)\leq k+\col_r(\Dd)$, and
  \item $\wcol_r(G)\leq \binom{k+r}{k}\cdot (\wcol_r(\Dd)+k)$. 
  \end{itemize}
  \end{theorem}

In the rest of this section, we present a proof of \Cref{thm:decomposable-wcol}, which is based on a construction of Pilipczuk and Siebertz~\cite{pilipczuk2019polynomial}. 
In the next section, we present a construction of Van den Heuvel et al.\ for $H$-minor free graphs~\cite{van2017generalised}, showing that the generalized coloring numbers are polynomially bounded on such classes. 
Based on the decomposition theorem for $H$-topological-minor-free 
graphs we give new bounds for the weak $r$-coloring numbers on 
$H$-topological-minor-free graphs. 
Finally, we present Dvo\v{r}\'ak's characterization of graphs $G$ 
with bounded $\adm_\infty(G)$~\cite{dvorak2012stronger}. 

\bigskip
Let $\Tt=(T,\bag)$ be a tree decomposition of $G$. 
For every vertex $u$ of $G$, let $x(u)$ be the unique~$\leq_T$-minimal node of $T$ with $u\in \bag(x)$. 
We define a quasi-order $\leq_{\Tt}$ on the vertex set of~$G$ by $u\leq_{\Tt} v$ if and only if $x(u)\leq_T x(v)$. 
We can use any linearization of $\leq_\Tt$ to obtain an order of~$V(G)$ achieving the bounds of \Cref{thm:wcol-tw} for the weak coloring numbers. 
We now extend this order to tree-decomposable graphs. 

\begin{lemma}\label{lem:margins-prt}
  For every vertex $u$ of $G$, the node $x(u)$ is the unique node of $T$ whose margin contains~$u$.
  \end{lemma}
  
  Note that by the lemma, $\{\mrg(x)\}_{x\in V(T)}$ is a partition of the vertex set of~$G$ and the margins of nodes of $T$ are exactly the classes of equivalence in the quasi-order $\leq_{\Tt}$ on $V(G)$. Now assume that $G$ admits a tree decomposition over a class $\Dd$. 
  Then, for each torso $\Gamma(x)$ there exists a linear order~$\pi_x$ with $\adm_r[\Gamma(x),\pi_x]\leq \adm_r(\Dd)$, $\col_r[\Gamma(x),\pi_x]\leq \col_r(\Dd)$, and $\wcol_r[\Gamma(x),\pi_x]\leq \wcol_r(\Dd)$, respectively (these could be different for each measure, but for simplicity of notation we assume they are the same). 
  We let $\pi'_x$ be the restriction of these orders to $\mrg(x)$. 
  We define $\pi$ such that it linearizes $\leq_\mathcal{T}$ and respects the orders $\pi'_x$. 

We now prove that the constructed order has the claimed properties. The proof follows a proof from ~\cite{pilipczuk2019polynomial}.
We need one more concept. 
 
\begin{definition}
  The \emph{skeleton} of $G$ over $\Tt$ is the directed graph
  with vertex set $V(G)$ and arcs~$(u,v)$ such that there is $x\in V(T)$ with $u\in \mrg(x)$ and $v\in \adh(x)$. 
\end{definition}

Note that if $(u,v)$ is an arc in the skeleton of $G$, then in particular $v<_{\Tt} u$, equivalently $x(v)<_T x(u)$. This implies that the skeleton is always acyclic (that is, it is a DAG).


\medskip

Now we observe that weak reachability across adhesions implies reachability in the skeleton.

\setcounter{claim}{0}

\begin{lemma}[Lemma 23 of \cite{pilipczuk2019polynomial}]\label{lem:path-skeleton}
Let $u,v$ be vertices of $G$ with $v<_{\Tt} u$ and let $P$ be a path in $G$ with endpoints $u$ and $v$ such that every vertex $w$ of $P$ apart from $v$ satisfies $v<_{\Tt} w$.
Then there exists a directed path $Q$ in the skeleton leading from~$u$ to~$v$ and satisfying $V(Q)\subseteq V(P)$.
\end{lemma}
Now the key observation is that if the considered tree decomposition has small adhesion, then every vertex reaches only a small number of vertices via short paths in the skeleton.
The argument is essentially the same as Grohe et al.'s argument that graphs of treewidth $k$ have polynomial weak $r$-coloring number~\cite{grohe2015colouring}.

\setcounter{claim}{0}

\begin{lemma}[Lemma 24 of \cite{pilipczuk2019polynomial}]\label{lem:skeleton-count}
Let $u\in V(G)$ and $r\in \N$. Then there exist at most $\binom{r+k}{k}$
vertices of $G$ that are reachable by a directed path of length at most $r$ from $u$ in the skeleton of $G$ over $\Tt$.
\end{lemma}

We can now use \cref{lem:skeleton-count} to prove \Cref{thm:decomposable-wcol}.

\begin{proof}[Proof of \Cref{thm:decomposable-wcol}]

We first show the claimed bounds for $\adm_r[G,\pi]$ and $\col_r[G,\pi]$, where~$\pi$ is the above constructed order with corresponding vertex enumeration $v_1,\ldots, v_n$.  
Consider any $v_i\in V(G)$ and let $x=x(v_i)$. Then 
a maximal $v_i-V_i$-fan has its ends in $\mrg(x)\cup \adh(x)$. 
Observe that if a path in $G$ leaves $\mrg(x)$ to a larger
bag $y$, i.e., $y$ is a child of $x$, and enters again so that
it ends in a smaller vertex, then this is also a path in the torso
$\Gamma(x)$, not leaving $\mrg(x)$. 
Hence, we have $\adm_r[G,\pi,v]\leq \adm_r(\Dd)+k$. 
The same argument shows that $\col_r[G,\pi,v]\leq \col_r(\Dd)+k$ for all $v\in V(G)$.

Finally, for $\wcol_r(G)$, let $v\in V(G)$ and let $P$ be a path of length at most $r$
between $u$ and $v$ such that $u$ is minimum on~$P$. Let 
$z$ be minimum on $P$ such that every vertex $w$ of $P$ 
apart from $z$ satisfies $z<_\mathcal{T} w$, or $z=u$ if 
no such $z$ exists. Then $u\in \WReach_r[G[\mrg(x(z)),\pi_x, z]\cup
\adh(x(z))$. 
Using
\cref{lem:path-skeleton} and \cref{lem:skeleton-count} we conclude
that $x(z)$ can take at most $\binom{r+k}{k}$ different values, 
and hence, we have $\WReach_r[G,\pi, v]\leq \binom{r+k}{k}\cdot (\wcol_r(\Dd)+k)$. 
\end{proof}

\subsection{\textit{H}-minor-free graphs}

We now turn to $H$-minor-free graphs. 
These graphs are tree decomposable over almost embeddable graphs by the famous structure theorem of Robertson and Seymour~\cite{robertson2003graph}, however, we do not follow this route to obtain the best bounds for the generalized coloring numbers. 
We rather follow the approach of \emph{flat chordal partitions}, which were also later the inspiration for the planar product structure theorem. 

Recall \Cref{def:elimination-order} up to \Cref{def:width-order}, defining 
simplicial vertices, perfect elimination orders, chordal graphs, the back-degree of a vertex and the width of an elimination order. 

\begin{definition}
  Let $G$ be a graph. A partition $\mathcal{P}=(V_1,\ldots, V_\ell)$ of $V(G)$ is a \emph{connected partition} if each $V_i$ induces a connected subgraph of $G$. 
  We denote the minor of $G$ obtained from $G$ by contracting each
  $V_i$ to a single vertex $v_i$ by $G/\mathcal{P}$. 
\end{definition}

Note that we care about the order of the $V_i$, that is, we consider $G/\Pp$ as an ordered graph. 
We inherit the above definitions of simplicial veritces, perfect elimination orders, etc.\ for connected partitions by identifying a part $V_i$ with the vertex $v_i$ in $G/\Pp$. 
For example, we call $V_i$ simplicial if~$v_i$ is simplicial in $G/\Pp$ and the width of $\Pp$ is the width of the order $(v_1,\ldots, v_\ell)$ of $G/\Pp$. 


With this definition, we derive for example the following lemma. 

\begin{lemma}
  Let $G$ be a graph and $\Pp$ be a connected partition of $G$. Then the treewidth of $G/\Pp$ is bounded from above by the width of $\Pp$. 
\end{lemma}

If $K_t\not\minor G$, then any chordal partition $\Pp$ of $G$ can have width at most $t-2$, hence, $G/\Pp$ can have treewidth at most $t-2$. 
Every graph $G$ admits a chordal partition of width $1$, just partition into the connected components of $G$. However, we aim for \emph{flat partitions} that have better local properties. 
The following definition was given by Van den Heuvel et al.~\cite{van2017generalised}. 

\begin{definition}
  Let $G$ be a graph and let
$f:\N\to\N$ be a function. A set $W\subseteq V(G)$ \emph{$f$-spreads in $G$} if,
for every $r\in\N$ and $v\in V(G)$, we have
\[|N^G_r(v)\cap W|\le f(r).\]
A partition $V_1,\ldots V_\ell$ of $G$ is \emph{$f$-flat} if each
$V_i$ $f$-spreads in $G-\bigcup_{1\le j<i}V_j$. 
\end{definition}

Given an $f$-flat connected partition $\Pp=(V_1,\ldots, V_\ell)$ of a graph $G$ of small width we can now define orders with good generalized coloring properties. 
First choose an arbitrary linear
order on the vertices of each set~$V_i$. Now let $\pi$ be the linear
extension of that order where for $v\in V_i$ and $w\in V_j$ with
$i<j$ we define $\pi(v)<\pi(w)$.

\begin{lemma}[Lemma 3.3 of \cite{van2017generalised}]\label{lem:deletevertices}
  Let $\Pp=(V_1,\ldots, V_\ell)$ be a connected partition of a graph $G$,
  and let $\pi$ be an order defined as above. For an integer
  $i$, $1\le i\le\ell$, let $G'= G-\bigcup_{1\le j<i} V_j$. Then
  for every $r\in\N$ and every $v\in V(G)$
  \begin{eqnarray*}
    \SReach_r[G,\pi,v]\cap V_i\:\subseteq\: N_r^{G'}[v]\cap V_i \text{, and}\\
    \WReach_r[G,L,v]\cap V_i\:\subseteq\: N_r^{G'}[v]\cap V_i.
  \end{eqnarray*}
\end{lemma}

Now it is easy to give upper bounds for $\col_r(G)$ and (using \Cref{thm:wcol-tw}) for
$\wcol_r(G)$ in terms of the width of a flat connected partition.

\begin{lemma}[Lemma 3.4 of \cite{van2017generalised}]\label{lem:spd}
  Let $f:\N\to \N$ and let $r,k\in\N$. Let $G$ be a graph that admits an
  $f$-flat connected partition of width~$k$. Then we have
  \begin{equation*}
    \col_r(G)\:\le\: (k+1)\cdot f(r).
  \end{equation*}
\end{lemma}

\begin{lemma}[Lemma 3.5 of \cite{van2017generalised}]\label{lem:spdwcol}
  Let $f:\N\to\N$ and let $r,k\in\N$. Let $G$ be a graph that admits an $f$-flat connected partition of width $k$. Then we have
  \begin{equation*}
    \wcol_r(G)\:\le\: \binom{r+k}{k}\cdot f(r).
  \end{equation*}
\end{lemma}

It remains to prove that $H$-minor-free graphs admit flat connected partitions of small width. The proof goes back to a decomposition of $H$-minor-free graphs due to Abraham~\cite{abraham14}, which is there called a \emph{cop-decomposition}. 
This name
is inspired by an earlier result of Andreae~\cite{andreae86}, which uses a
cops-and-robber game argument to establish the existence of decompositions of $H$-minor-free graphs along isometric paths, which then leads to the desired flat partitions of small width. 

\begin{definition}
  A path $P$ in a graph $G$ is an \emph{isometric path} if $P$ is a shortest
  path between its endpoints. We call a partition $\Pp=(V_1,\ldots, V_\ell)$ an
  \emph{isometric paths partition} if each $V_i$ induces an isometric path in
  $G-\bigcup_{1\le j<i} V_j$. 
\end{definition}

Note that an isometric paths partition is in particular a connected partition. 
Note furthermore, that the vertex sets $V_i$ for $i>0$ not necessarily induce isometric paths in $G$, but only in $G$ after the earlier paths have been deleted. 
It will be the key to the planar product structure theory to partition the graph into globally isometric paths. The next lemma shows that isometric paths partitions are flat partitions.  

\begin{lemma}\label{lem:shortestpath}
  Let $v$ be a vertex of a graph $G$, and let $P$ be an isometric path in
  $G$. Then $P$ contains at most $2r+1$ vertices of the closed
  $r$-neighborhood of $v$:
  $|N_r[v]\cap V(P)|\le \min\{|V(P)|,\,2r+1\}$.
\end{lemma}

Finally, based on the arguments of Abraham and Andreae~\cite{abraham14,andreae86} it was shown by Van den Heuvel et al.~\cite{van2017generalised} that for every graph that excludes~$K_t$ as a minor one can find a flat connected partition where each part consists of at most $t-3$ isometric paths and which is of width at most~$t-2$. 

\begin{lemma}[\cite{van2017generalised}]\label{lem:minordecompnew}
  Let $t\ge4$ and let $f:\N\to \N$ be the function $f(r)=(t-3)(2r+1)$. Let~$G$ be a graph that excludes~$K_t$ as a minor. Then there exists a
  connected $f$-flat partition of $G$ of width at most $t-2$.
\end{lemma}

This immediately leads to the following theorems.

\begin{theorem}[\cite{van2017generalised}]\label{thm:col-Hmin}
  Let $t\ge 4$. For every graph $G$ that excludes $K_t$ as a minor, we have
  \begin{equation*}
    \col_r(G)\:\le\: \binom{t-1}{2}{}\cdot{}(2r+1).
  \end{equation*}
\end{theorem}

\begin{theorem}[\cite{van2017generalised}]\label{thm:wcol-Hmin}
  Let $t\ge 4$. For every graph $G$ that excludes $K_t$ as a minor, we have
  \begin{equation*}
    \wcol_r(G)\:\le\: \binom{r+t-2}{t-2}{}\cdot{}(t-3)(2r+1)\in
    \Oof(r^{\,t-1}).
  \end{equation*}
\end{theorem}

Van den Heuvel and Wood~\cite{van2018improper} started investigating whether the structure of the excluded minor can give even more insights on the generalized coloring numbers and they proved that for every $H$-minor free graph $G$ we have $\wcol_r(G)\in \Oof(r^{|vc(H)+1|})$, where $vc(H)$ is the vertex cover number of $H$. 
This result was recently improved by Dujmovi\'c et al.~\cite{dujmovic2024grid} who proved that there exists an exponential function $g$ such that all $H$-minor free graphs $G$ satisfy $\wcol_r(G)\in \Oof(r^{g(\td(H)})$. 
Hodor et al.~\cite{hodor2024weak} improved this result even more and showed that $H$-minor free graphs $G$ satisfy $\wcol_r(G)\in \Oof(r^{\td(H)})$. 
In fact, they prove that the exponent is bounded by the \emph{rooted $2$-treedepth} of~$H$ (a measure the authors introduced in~\cite{hodor2024weak}), and which is defined as follows. 

\begin{definition}
    The \emph{rooted $2$-treedepth} of a graph $H$, $\rtd_2(H)$, is defined recursively as follows: 
    $\rtd_2(H)$ is 
    \begin{itemize}
        \item $0$ if $H$ is empty,
        \item $1$ if $H$ is a single-vertex graph, and 
        \item the minimum of $\max\{\rtd_2(H[A]), \rtd_2(H[B\setminus A])+|A\cap B|\}$ over all separations $(A,B)$ of $H$ of order at most one with $A\neq \emptyset$ and $B\setminus A\neq \emptyset$. 
    \end{itemize}
    
\end{definition}

\begin{theorem}[\cite{hodor2024weak}]
    For every positive integer $t$ and for every graph $H$ with $\rtd_2(H)\leq t$ there exists an integer $c$ such that for every $H$-minor free graph $G$ and every integer $r\geq 2$ we have \[\wcol_r(G)\leq c\cdot r^{t-1}\cdot \log r.\]
\end{theorem}

\subsection{Planar graph and graphs of bounded genus}
\label{sec:wcolplanar}

In this section, we look at upper bounds for
$\col_r(G)$ and $\wcol_r(G)$ when~$G$ is a graph with bounded genus. 
Since for every genus~$g$, there exists a~$t$ such that every graph with genus at
most~$g$ does not contain $K_t$ as a minor, we could use
\Cref{thm:col-Hmin} to obtain upper bounds for the generalized coloring
numbers of such graphs. However, one can obtain significantly better bounds. 
First, it was shown that planar graphs have isometric paths partitions of width 
at most $2$. 

\begin{lemma}[\cite{van2017generalised}]\label{lem:wcolplanar}
  Every planar graph $G$ has an isometric paths partition of
  width at most 2.
\end{lemma}

For a graph of genus $g$, we first find isometric paths that we can delete to bring down the genus to $0$, that is, to arrive at a planar graph. For every graph of genus $g>0$ there exists a non-separating cycle $C$ that consists of two isometric paths such that $G-C$ has genus $g-1$, see e.g.\ the textbook of Mohar~\cite{mohar2001graphs}. By arranging these vertices first in the linear order, we obtain the following bounds.

\begin{theorem}[\cite{van2017generalised}]\label{thm:wcol-planar}
  For a planar graph $G$, we have
  \begin{equation*}
   \wcol_r(G)\le \binom{r+2}{2}\cdot(2r+1)\in \Oof(r^3),
  \end{equation*}
  and for a graph $G$ with genus $g$, we have
  \begin{equation*}
   \wcol_r(G)\le \biggl(2g+\binom{r+2}{2}\biggr)\cdot(2r+1)\in \Oof(r^3+gr).
  \end{equation*}
\end{theorem}

%


In case the treewidth of $G$ is bounded the following improved bounds hold, as proved by Hodor et al.~\cite{hodor2024weak}. 

\begin{theorem}[\cite{hodor2024weak}]
    For all non-negative integers $g$ and $w$ there exists an integer $c$ such that for every graph $G$ of Euler genus at most $g$ and with treewidth at most $w$ we have \[\wcol_r(G)\leq c\cdot r^2\cdot\log r.\]
\end{theorem}

This bound is tight for planar graphs of bounded treewidth. It is one of the main open problems whether in general the weak coloring numbers of planar graphs is bounded by $\Oof(r^2\cdot \log r)$. 

\begin{problem}
    Is it true that $\wcol_r(G)\in \Oof(r^2\log r)$ for every planar graph $G$?
\end{problem}

From the isometric paths partitions for planar graphs, one would already obtain good bounds also for the strong coloring numbers. A more careful construction (also based on isometric paths partitions) and analysis gives the following bounds. 

\begin{theorem}[\cite{van2017generalised}]
  For a planar graph $G$, we have 
  \begin{equation*}
    \col_r(G)\le 5r+1,
   \end{equation*}
   and for a graph $G$ with genus $g$, we have
   \begin{equation*}
    \col_r(G)\le (2g+3)\cdot(2r+1).
   \end{equation*}
\end{theorem}

For quite a while it was an open question whether the $r$-admissibility of planar graphs could even be sublinear. This was eventually proved by 
Nederlof, Pilipczuk and Wegrzycki~\cite{NederlofPW23}.  

\begin{theorem}[\cite{NederlofPW23}]\label{thm:adm-planar}
  For a planar graph $G$ we have 
  \[\adm_r(G)\in \Oof(r/\log r).\]
\end{theorem}

This bound is tight as stated in the next section. 
The proof of \Cref{thm:adm-planar} is based on a theorem of Koebe, showing that every planar graph has a so-called coin model. A coin model associates to every vertex $v$ of a graph $G$ a disc (circle) $D(v)$ in the plane such that the discs are pairwise internally disjoint, and whenever two vertices $u$ and $v$ are adjacent in $G$, the discs $D(u)$ and $D(v)$ touch. 
Given a coin model, one can naturally derive an order by non-increasing radii of the associated discs, which yields the claimed bounds for the $r$-admissibility. 

\subsection{\textit{H}-topological-minor-free graphs}

To study $H$-topological-minor-free graphs we are coming back to tree decompositions. 
We are going to apply the following structure theorem, 
proved in this quantitative form by Erde and Wei\ss auer~\cite{erde2018short}. A result of this type was first 
proved 
by Grohe and Marx~\cite{grohe2015structure}, and then improved by Dvo\v{r}\'ak~\cite{dvorak2012stronger}. 

\begin{theorem}[\cite{erde2018short}]
Let $G$ be a graph such that $K_t\not\minor^{top} G$. Then 
$G$ has a tree decomposition of adhesion smaller than $t^2$
such that every torso either
\begin{itemize}
\item has fewer than $t^2$ vertices of degree at least
$2t^4$, or
\item excludes $K_{2t^2}$ as a minor. 
\end{itemize}
\end{theorem}

We immediately derive the following corollary. 

\begin{corollary}
Let $G$ be a graph such that $K_t\not\minor^{top} G$ and 
let $r$ be a positive integer. Then 
\begin{itemize}
\item $\adm_r(G)\leq t^2+(t^2+2t^4)+4t^4(2r+1)$
\item $\col_r(G)\leq t^2+(t^2+2t^{4r})+4t^4(2r+1)$
\item $\wcol_r(G)\leq (4t^4(2r+2))^r$. 
\end{itemize}
\end{corollary}

By a classical result of~\cite{bollobas1998proof,komlos1994topological} graphs that exclude $K_t$ as a topological minor have edge density~$\Oof(t^2)$, hence, satisfy $\widetilde{\nabla}_r(G)\in \Oof(t^2)$ for all non-negative integers $r$. 
Hence, these bounds even improve the bounds we would get for the $r$-admissibility from \Cref{thm:adm-bound}. 

\subsection{Graph product structure theory}\label{graph-product}

An isometric paths partition $\Pp=(V_1,\ldots, V_\ell)$ of small width provides already some insight into the treelike structure of a graph. 
The main shortcoming of these partitions is that a path induced by $V_i$ is not globally isometric but only isometric in the subgraph $G-\bigcup_{j<i}V_j$. 
In the study of $p$-centered colorings (which we will consider in \cref{sec:ltc}) it was proved by Pilipczuk and Siebertz that planar graphs actually admit isometric paths partitions of width at most $8$ such that each path of the partition is a globally isometric path~\cite{pilipczuk2019polynomial}. 
This result was further strengthened by Dujmovi\'c et al.~\cite{dujmovic2020planar}, who showed the following result. Consider a tree $T$ rooted at a node $r$. A path $v_1,\ldots,v_t$ in $T$ is called \emph{vertical} if for some $d\geq 0$ and for all $1\leq i\leq t$ we have $\dist_T(v_i,r)=d+i$. 

\begin{theorem}[\cite{dujmovic2020planar}]
  Let $T$ be a rooted spanning tree in a connected planar graph $G$. Then $G$ has a partition $\Pp$ into vertical paths in $T$ such that $G/\Pp$ has treewidth at most $8$.
\end{theorem}

Note that if $T$ is a BFS spanning tree, then all vertical paths in $T$ are isometric paths. Hence, the result yields a strengthening of the result of~\cite{pilipczuk2019polynomial} that has led to the groundbreaking result of Dujmovi\'c et al.~\cite{dujmovic2020planar}, the planar product structure theorem.

\begin{definition}
The strong product $G\boxtimes H$ of two graphs $G$ and $H$ is the graph with vertex set $V(G)\times V(H)$ that includes the edge with endpoints $(v, x)$ and $(w, y)$ if and only if $vw\in E(G)$ and $x = y$; $v = w$ and $xy \in E(H)$; or $vw \in E(G)$ and $xy\in E(H)$ (see Figure 14). 
\end{definition}

\begin{figure}[htt!]
\begin{center}
        \centering
        \includegraphics[width=0.5\textwidth]{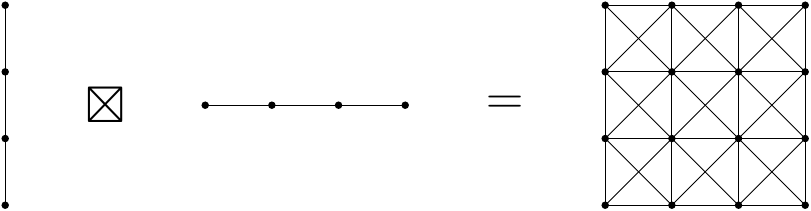}
        \caption{The strong product of two paths.}
\end{center}
\end{figure}

\begin{theorem}[\cite{dujmovic2020planar}]
  Every planar graph is a subgraph of
  \begin{itemize}
    \item $H\boxtimes P$ for some planar graph $H$ with treewidth at most $8$ and some path $P$.
    \item $H\boxtimes P\boxtimes K_3$ for some planar graph $H$ with treewidth at most $3$ and some path $P$.
  \end{itemize}
\end{theorem}

The planar product theorem had an immense impact and was the basis to solve major open problems and to make progress on longstanding open problems. For example, it was used by Dujmovi\'c et al.~\cite{dujmovic2020planar} to show that planar graphs have queue layouts
with a bounded number of queues; this was a longstanding open problem. It was used by Dujmovi\'c et al.~\cite{dujmovic2020minor} that planar
graphs can be nonrepetitively colored with a bounded number of colors; another longstanding open problem. 
%
%
%
We refer to the recent survey~\cite{dvovrak2021notes} for further applications and product structure theorems for various graph classes. 



\subsection{\texorpdfstring{$\infty$}--admissibility}

Let $G$ be a graph, $v\in V(G)$ and $A\subseteq V(G)$. Recall that a 
$v-A$ \emph{fan} is a 
set of paths $P_1,\ldots P_k$ with one endpoint in $v$ and the
other endpoint in $A$ and which are internally 
vertex disjoint from~$A$ 
such that $V(P_i)\cap V(P_j)=\{v\}$
for all $i\neq j$. 
A \emph{$k$-fan set} in $G$ is a set $A\subseteq V(G)$ 
such that for every $a\in A$, there exists a depth-$r$ 
$a-A$ fan of order $k$ in $G$. 
According to \Cref{thm:adm} the $\infty$-admissibility of $G$ is equal to the largest number $k$ such that there exists a $k$-fan set in $G$. 

We define a related concept. Let $G$ be a graph and $k\in \N$. As set $X\subseteq V(G)$ of size at least $k$ is \emph{$k$-inseparable} if no two vertices of $X$ can be separated in $G$ by deleting fewer than $k$ vertices. 
A maximal such set is called a \emph{$k$-block}. The maximum integer $k$
such that $G$ contains a $k$-block is the \emph{block number} of $G$, denoted by $\beta(G)$. 
Weissauer~\cite{weissauer2017block} proved the following.

\begin{theorem}[\cite{weissauer2017block}]\label{thm:blocks-adm}
  For every graph $G$
  \[ \frac{\adm(G)+1}{2}\leq \beta(G)\leq \adm(G).\]
  \end{theorem}

\pagebreak
Weissauer~\cite{weissauer2017block} proved the following decomposition theorems. 

\begin{theorem}
If $G$ has no $(k+1)$-block, then $G$ has a 
tree decomposition of adhesion less than $k$ in which every
torso has at most $k$ vertices of degree at least $2k(k-1)$. 
\end{theorem}

\begin{theorem}
If $G$ has a tree decomposition in which every torso has 
at most $k$ vertices of degree at least $k$, then G has no 
$(k+1)$-block.
\end{theorem}

We conclude the following corollary, which was first proved by Dvo\v{r}\'ak~\cite{dvorak2012stronger}
\begin{corollary}
    A graph $G$ has bounded $\adm_\infty(G)$ if and only if 
it admits a tree decomposition with bounded adhesion over 
a class of bounded degree with a bounded number of apex vertices (vertices of arbitrary degree). 
\end{corollary}

\section{Lower bounds}\label{sec:lower-bounds}

We now turn to more sophisticated lower bounds for the generalized coloring numbers. 
The first result by Grohe et al.~\cite{grohe2015colouring} shows that the upper bounds of \Cref{thm:wcol-tw} for the weak coloring numbers for graphs of bounded treewidth are optimal. 

\begin{theorem}[\cite{grohe2015colouring}]\label{thm:Gkr}
  For every $k\geq 1$, $r\geq 1$, there is a family of graphs $G^k_r$ with $\tw(G^k_r) = k$, such
  that $\wcol_{r}(G^k_r)=\binom{r+k}{k}$. 
\end{theorem}

This result shows also that the upper bounds of \Cref{thm:wcol-Hmin} for $K_t$-minor free graphs are almost optimal. Furthermore, it verifies an exponential gap between 
$\wcol_r$ and $\col_r$.

\begin{corollary}[\cite{grohe2015colouring}]\label{colgap}
  For every $k\geq 1$, $r\geq 1$, there is a graph $G^k_r$ such 
  that $\col_{r}(G^k_r)=k+1$ and $\wcol_{r}(G^k_r) \geq \big(\frac{\col_{r}(G^k_r)}{r}\big)^{r}$.
\end{corollary}

As recently shown by Hodor et al.~\cite{hodor2024weak}, if a minor-closed class $\Cc$ contains graphs with $\wcol_r(G)\in \omega(r^t\log r)$, then it contains $G_r^t$ for every non-negative integer $r$. 
As a consequence they obtain that for minor-closed classes the class of graph $G_r^t$ is, up to an $\Oof(r\log r)$ factor, the unique construction of graphs with large weak coloring numbers. 

\smallskip
By \Cref{lem:lowerbound-nabla-col} for all graphs $G$ and positive integers $r$ we have 
$\col_{4r+1}(G) \geq \nabla_r(G)$. 
It is easy to see that every graph of girth greater than $8r$ and with minimum degree at least $d+1$ satisfies $\nabla_{2r}(G)\geq d^r$. As the $r$-admissibility is trivially bounded by the maximum degree, we conclude that there is also an exponential gap between $\col_r$ and $\adm_r$, as witnessed by $d$-regular high girth graphs. A stronger bound was established in~\cite{grohe2015colouring}. 

\begin{theorem}[\cite{grohe2015colouring}]
  Let $G$ be a $d$-regular graph of girth at least $4g+1$, where $d\geq 7$. Then for every $r\leq g$, \[\col_r(G)\geq \frac{d}{2}\big(\frac{d-2}{4}\big)^{2^{\lfloor \log r\rfloor}-1}\]
\end{theorem}

For the weak coloring numbers we have the following lower bounds on regular high-girth graphs. 

\begin{theorem}[\cite{grohe2015colouring}]\label{wcoldreg}
  Let $G$ be a $d$-regular graph of girth at least $2g+1$, where
  $d\geq 4$. Then for every $r\leq g$, 
  $$\wcol_r(G)\geq \frac{d}{d-3}\left(\Big(\frac{d-1}{2}\Big)^r-1\right)\,.$$
\end{theorem}

We can also give a concrete example of a graph class with polynomial (and in fact even linear) expansion that has super-polynomial weak coloring numbers. For a graph~$G$ denote by $G^{(r)}$ the exact $r$-subdivision
of $G$, that is, the graph obtained from~$G$ by replacing every edge by a path of length
$r+1$ (with $r$ vertices on it). 
Recall the question of Joret and Wood, formulated in \Cref{prob:pol-exp}, whether classes of polynomial expansion have polynomial strong coloring numbers. 

\begin{theorem}[\cite{grohe2015colouring}]\label{PolyExpWeak}
  The class $\CCC=\{G^{(6\tw(G))}~:~G$ graph$\}$ has polynomial (even linear) expansion and super-polynomial weak coloring numbers.  
\end{theorem}

Finally, let us mention that the bounds for the $r$-admissibility on planar graphs presented in \Cref{thm:adm-planar} are tight. 
The construction of Dvo\v{r}\'ak and Siebertz~\cite{zdenek-seb} is presented explicitly in the work of Nederlof et al.~\cite{NederlofPW23}. 

\begin{theorem}[\cite{zdenek-seb}]
  For every planar graph $G$ we have 
  \[\adm_r(G)\in \Omega(r/\log r).\]
\end{theorem}

\section{Structural decompositions}\label{sec:ltc}

Tree decompositions offer a way of globally decomposing a graph 
along small separators into well behaved pieces. 
On the other hand, the generalized coloring numbers offer rather local insights. 
In this section we will study another way of decomposing graphs into well behaved pieces. These turn out to be equivalent to the generalized coloring numbers but offer different insights. 
The structural decompositions we discuss here were first introduced by DeVos et al.~\cite{devos2004excluding} in terms of proper vertex colorings. 
For a positive integer $p$, a \emph{$p$-treewidth coloring} of a graph $G$ is a vertex coloring such that the subgraph induced by any $i\leq p$ color classes has treewidth at most~$i-1$. 
Hence, every single color class induces a graph of treewidth $0$, that is, an edgeless graph. 
In particular, the coloring is a proper vertex coloring. 
A class of graphs admitting a $p$-treewidth coloring for every positive integer $p$ with a number of colors depending only on $p$ is said to admit 
\emph{low treewidth colorings}. De Vos et al.~\cite{devos2004excluding} proved that every class excluding some graph as a minor admits low treewidth colorings. 
Not much later, Ne\v{s}et\v{r}il and Ossona de Mendez proved that these classes even admit \emph{low treedepth colorings},
then introduced classes with bounded expansion and proved \cite{nevsetvril2008grad} that they  are characterized by the existence of low treedepth colorings. 
Finally, Zhu~\cite{zhu2009colouring} provided an alternate characterization of classes with bounded expansion by the boundedness of generalized coloring numbers.
The terminology has changed from structural colorings to structural decompositions and 
after presenting Zhu's result we will switch to this terminology. 

\subsection{p-centered colorings}

Let us begin with the two key definitions. 

\begin{definition}
  For a positive integer $p$, a \emph{$p$-centered coloring} of a graph $G$ is a vertex coloring of~$G$ such that for any connected subgraph $H\subseteq G$, either some color appears exactly once in $H$, or~$H$ gets at least $p$ colors. 
  The least number of colors in a $p$-centered coloring of $G$ is denoted~$\chi_p(G)$. 
\end{definition}

A \emph{centered coloring} of $G$ is a vertex coloring such that, for any
connected subgraph $H$ some color appears exactly once in $H$. It is easy to see that the treedepth of $G$ is equal to the minimum number of colors in a centered
coloring of $G$.

\begin{definition}
  For a positive integer $p$, a \emph{$p$-treedepth coloring} of a graph $G$ is a vertex coloring of~$G$ such that the subgraph induced by any $i\leq p$ color classes has treedepth at most $i$ (see Figure~15). 
\end{definition}

\begin{figure}[h!]
\begin{center}
        \centering
        \includegraphics[width=0.3\textwidth]{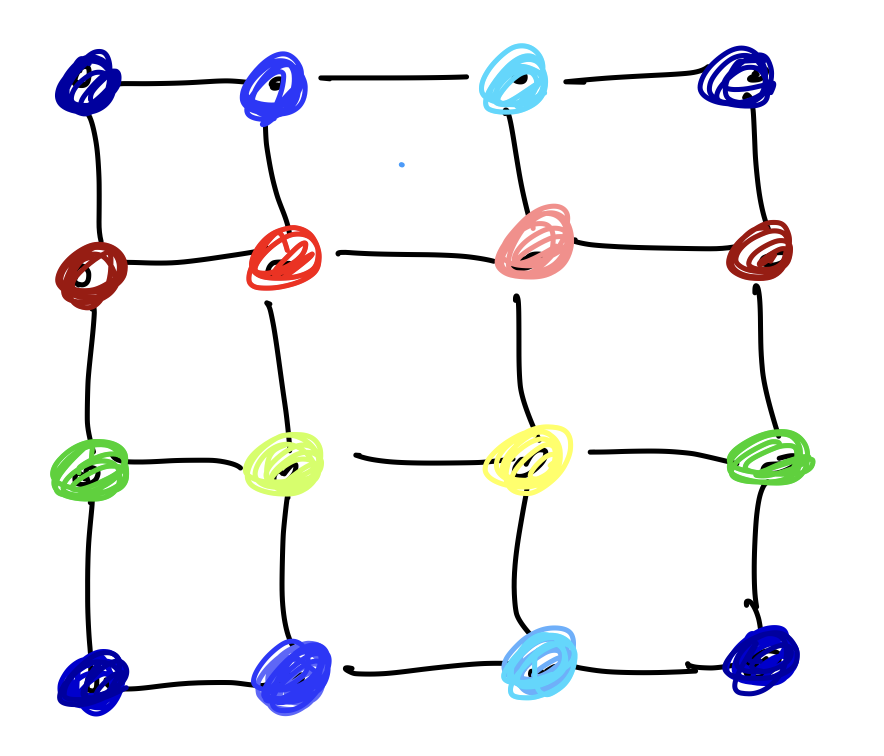}
        \caption{A $3$-centered coloring of a grid with $5$ colors.}
\end{center}
\end{figure}

The following lemma by Ne\v{s}et\v{r}il and Ossona de Mendez~\cite{nevsetvril2006tree} shows that $(p+1)$-centered colorings yield $p$-treedepth colorings. 

\begin{lemma}[\cite{nevsetvril2006tree}]
  A $q$-centered coloring for $q>p$ is a $p$-treedepth coloring. 
\end{lemma}

The converse is not true, however, the existence of a $p$-treedepth coloring with few colors implies the existence of a $p$-centered coloring with few colors. We did not find a reference for this fact, but we believe that it is folklore in the community. 
We include a proof for completeness. 
\begin{lemma}
Let $p\geq 2$.
  If $G$ admits a $p$-treedepth coloring with $k$ colors, then $G$ admits a $(p+1)$-centered coloring with $k\bigl(2p\cdot \binom{k-1}{p-1}+1\bigr)<2pk^p$ colors. 
\end{lemma}
\begin{proof}
Fix a $p$-treedepth coloring $c$ of $G$. Let $\vec H$ be the directed graph with vertex set $V(G)$, where $uv\in E(\vec H)$ if $v$ is an ancestor of $u$ for some choice $I$ of $p$ colors including $c(u)$ in the treedepth decomposition $Y_I$ for the colors $I$. 
There are $\binom{k-1}{p-1}$ such choices for $I$ and in each $Y_I$ vertex $u$ has at most $p$ ancestors. 
Then, every vertex has outdegree at most $p\cdot \binom{k-1}{p-1}$ hence has chromatic number at most $2p\cdot \binom{k-1}{p-1}+1$.
Let $\gamma$ be a proper coloring of $H$.
We consider the product coloring $v\mapsto (c(v),\gamma(v))$.
Assume $H$ is a connected subgraph of $G$ with at most $p$ colors. 
Let $J=\{c(v)\colon v\in V(H)\}$. 
Then $|J|\leq p$, thus $H$ is included in the subgraph $G_I$ of $G$ induced by a set $I\supseteq J$ of $p$ $c$-colors. As $H$ is connected, it contains a vertex $r$ that is an ancestor of all the vertices of $H$. By construction, this vertex is adjacent in $H$ to all the other vertices of $H$, hence the $(c,\gamma)$-color of $r$ is unique in $H$.
\end{proof}

We now show how to obtain $p$-centered colorings from the generalized coloring numbers, proving that nowhere dense and bounded expansion classes admit $p$-centered colorings with a small number of colors.

Let $G$ be a graph, let $r$ be a positive integer and let $\pi$ be an order of $V(G)$. We define $G\left\langle \pi,r\right\rangle$ as the 
graph with vertex set $V(G)$ and for $u\neq v$ we have $uv\in E(G\left\langle \pi,r\right\rangle)\Leftrightarrow u\in \WReach_r[G,\pi,v]$ or $v\in \WReach_r[G,\pi,u]$. 

\begin{definition}
  Let $G$ be a graph and let $r$ be a positive integer. 
  We define $\wcol_r^\star(G)$ as the chromatic number
  of $G\left\langle \pi,r\right\rangle$, where $\pi$ is minimum over
  all linear orders of $V(G)$.
  \end{definition}
  
  Of course, $\wcol_r^\star(G)\leq \wcol_r(G)$ for all 
  graphs $G$ and all integers $r$, since we can just use the order $\pi$ witnessing that $\wcol_r(G)$ is small and the greedy 
  coloring along the order. In some 
  cases $\wcol_r^\star(G)$ may be surprisingly small. For example, $\wcol_2^*(G)$ is bounded by $3\nabla_0(G)^2+1+\min\{\nabla_0(G)\nabla_1(G), \nabla_0(G)^2\nabla_1(G)\}$, as shown Esperet and Wiechert~\cite{EsperetW18}. 
  

  

\medskip
We now turn to Zhu's result~\cite{zhu2009colouring}. 

\begin{theorem}[\cite{zhu2009colouring}]
  Let $G$ be a graph and $p$ a positive integer. Then $\chi_p(G)\leq \wcol_{2^{p-2}}^\star(G)$. 
\end{theorem}
\begin{proof}
  Let $\pi$ be an order of $V(G)$ and $c$ be a coloring witnessing that 
  $G\left\langle \pi,2^{p-2}\right\rangle$ has chromatic number $d$. 
  Let $H$ be a connected subgraph of $G$ and let $v_0$ be the minimum vertex 
  of $H$ with respect to $\pi$. If $c(v_0)$ appears exactly once in $H$ we are done. Otherwise, we prove that $H$ gets at least $p$ colors. 

  As $c(v_0)$ is not unique in $H$, there exists $u\in V(H)$, $u\neq v_0$, with $c(u)=c(v_0)$. As $H$ is connected there exists a path $P_0=v_0v_1\ldots v_\ell=u$ in $H$ connecting $v_0$ and $u$. Note that $P_0$ must have length greater than $2^{p-2}$, as otherwise $v_0$ is weakly $2^{p-2}$-reachable from $u$ ($v_0$ is minimum in~$H$ and the path uses only vertices of $H$), and hence would receive a color different from $u$. 

  Now $P_1=v_1\ldots v_{2^{p-2}}$ does not use color $c(v_0)$ and has $2^{p-2}$ vertices. As $P_1$ has only $2^{p-2}$ vertices, every vertex weakly reaches the minimum vertex $u_1$ of $P_1$. Hence, $u_1$ gets a unique color. We repeat the process with the larger component $P_2$ of $P_1-u_1$, which has at least $2^{p-3}$ vertices. Hence, after $p-2$ steps we have found the vertices $u_1,\ldots, u_{p-2}$ with pairwise different colors and we are left with $P_{p-2}$, which consists of a single vertex $u_{p-1}$ (with color distinct from $v_0,u_1,\dots,u_{p-1}$). Together with $v_0$ we found that $H$ gets at least $p$ colors. 
\end{proof}

\begin{corollary}
  \mbox{}
  \begin{enumerate}
    \item A class $\Cc$ of graphs has bounded expansion if and only if for every $p$ there exists $c(p)$ such that every $G\in \Cc$ satisfies $\chi_p(G)\leq c(p)$. 
    \item A class $\Cc$ of graphs is nowhere dense if and only if for every $p$ and every $\epsilon>0$ there exists $c(p,\epsilon)$ such that every $n$-vertex graph $H\subseteq G\in \Cc$ satisfies $\chi_p(H)\leq c(p,\epsilon)\cdot n^\epsilon$. 
  \end{enumerate}
\end{corollary}

Unfortunately, the general bounds for the $p$-centered colorings depend on $\wcol_{2^{p-2}}^\star(G)$. 
By the large radius $2^{p-2}$, this leads to bounds that are exponential in $p$ even for very restricted graph classes. 
Direct constructions improve these bounds to polynomial bounds on many important graph classes. 
For example, we have polynomial bounds on classes that exclude a minor by the result of Pilipczuk and Siebertz~\cite{pilipczuk2019polynomial}, and, as mentioned before, this result not only improved the bounds for $p$-centered colorings, but it was the inspiration for the planar product structure theory~\cite{dujmovic2020planar}. 
The bounds for planar graphs were improved to $\Oof(p^3\cdot \log p)$ by Dębski et al.~\cite{debski2020improved}. The proof is an easy corollary of a coloring procedure for strong products and the product structure theorem for planar graphs. 
Quite surprisingly, as shown by Dębski et al., also bounded degree graphs have polynomial $p$-centered colorings, and by the tree decomposition method, this result lifts to classes with excluded topological minors~\cite{debski2020improved}. 
The beautiful proof for bounded degree graphs builds on the method of entropy compression. Furthermore, Dębski et al.~\cite{debski2020improved} and Dubois et al.~\cite{dubois2020two} provide several lower bounds. 

\begin{theorem}[\cite{debski2020improved,dubois2020two,pilipczuk2019polynomial}]
  \mbox{}
  \begin{enumerate}
    \item Graphs of treewidth at most $k$ admit $p$-centered colorings with $\binom{p+k}{k}$ colors~\cite{pilipczuk2019polynomial} and this bound is tight~\cite{debski2020improved}. 
    \item If $\Cc$ is a class of graphs that is tree-decomposable with adhesion at most $k$ over a class $\Dd$ that admits $p$-centered colorings with $\Oof(p^d)$ colors, then $\Cc$ admits $p$-centered colorings with $\Oof(p^{d+k})$ colors~\cite{pilipczuk2019polynomial}. 
    \item Planar graphs admit $p$-centered colorings with $\Oof(p^3\log p)$ colors and there is a family of planar graphs of treewidth $3$ that requires $\Omega(p^2\log p)$ colors in any $p$-centered coloring. 
    \item Graphs with Euler genus $g$ admit $p$-centered colorings with $\Oof(gp+p^3\log p)$ colors~\cite{debski2020improved}. 
    \item For every graph $H$ there is a polynomial $f$ such that the graphs excluding $H$ as a minor admit $p$-centered colorings with at most $f(p)$ colors~\cite{pilipczuk2019polynomial}.
    \item Graphs with maximum degree $\Delta$ admit $p$-centered colorings with $\Oof(\Delta^{2-1/p}\cdot p)$ colors~\cite{debski2020improved} and  there are graphs of maximum degree $\Delta$ that require $\Omega(\Delta^{2-1/p}\cdot p\cdot \log^{-1/p}\Delta)$ colors in any $p$-centered coloring~\cite{dubois2020two}. 
    \item For every graph $H$ there is a polynomial $f$ such that the graphs excluding $H$ as a topological minor admit $p$-centered colorings with at most $f(p)$ colors~\cite{debski2020improved}. 
  \end{enumerate}
\end{theorem}

The following question was posed by Hodor et al.~\cite{hodor2024weak}.
\begin{problem}
  Does there exist a function $g$ such that for every fixed graph $H$, for every $H$-minor-free graph~$G$ and every positive integer $p$, $G$ admits a $p$-centered coloring with $\Oof(r^{g(\td(H))})$ colors?
\end{problem}

\subsection{Structural decompositions}\label{sec:structural-decompos}

As discussed above, $p$-treedepth colorings can be seen as structural decompositions. 
The key property they provide is that every subgraph with at most $p$ vertices can be found as a subgraph of bounded treedepth. When $p$-treedepth colorings with few colors can be efficiently computed, this, for example, immediately leads to efficient algorithms for the subgraph isomorphism problem or the model checking problem for existential FO. 
These observations led to the study of colorings such that any $p$ color classes induce graphs with other good structural properties. Let us discuss this more general approach, which in particular opened the road to the study of dense but structurally sparse classes of graphs. 

A \emph{class property} $\Pi$ is a set of graph classes that is closed by subclasses (that is, if $\mathscr C\in\Pi$ and $\mathscr D\subseteq\mathscr C$, then $\mathscr D\in\Pi$). 
For a function $f:\mathbb N\rightarrow\mathbb N$ and a parameter $p$, a class $\mathscr C$ admits an 
\emph{$f$-bounded $\Pi$-decomposition with parameter $p$} if  there exists a graph class $\mathscr D_p\in\Pi$, such that  every $G\in\mathscr C$ has a vertex coloring with at most $f(|G|)$ colors with the property  that every $p$ color classes induce a subgraph in $\mathscr D_p$.  A class $\mathscr C$ has \emph{bounded-size} decompositions (resp.\ \emph{quasibounded-size} decompositions) if, for every positive integer $p$, it has an $f_p$-bounded $\Pi$-decomposition with parameter $p$, where  $f_p(n)=\mathcal O(1)$  (resp. $f_p(n)=\mathcal O(n^\epsilon)$ for every $\epsilon>0$). 
Hence, a class has bounded expansion if and only if it has bounded-size bounded treedepth decompositions and it is nowhere dense if and only if its hereditary closure has quasibounded-size bounded treedepth decompositions. 

The study of dense but structurally sparse classes by structural decompositions was initiated by Kwon et al.~\cite{kwon2020low}, who proved that fixed powers of bounded expansion classes have bounded-size bounded cliquewidth decompositions. 
This result was strengthened and extended by Gajarsk{\'y} et al.~\cite{gajarsky2020first} showing that even all classes that are transductions of bounded expansion classes have bounded-size shrubdepth decompositions. Such classes are called classes with \emph{structurally bounded expansion}. 
We also define dependent, monadically dependent, stable, and monadically stable classes, which turn out to be important generalizations of nowhere dense classes. 

\begin{itemize}
    \item A class $\Dd$ has \emph{structurally bounded expansion} if there exists a class $\Cc$ of bounded expansion such that $\Dd$ can be $\FO$-transduced from $\Cc$~\cite{gajarsky2020first}. 
    \item A class $\Dd$ is \emph{structurally nowhere dense} if there exists a nowhere dense class $\Cc$ such that $\Dd$ can be FO-transduced from $\Cc$~\cite{gajarsky2020first}. 
    \item A class $\Cc$ is \emph{stable} if from $\Cc$ one cannot interpret (with an FO-interpretation in powers) the class of all linear orders. 
    \item A class $\Cc$ is \emph{monadically stable} if every monadic lift of $\Cc$ is stable. 
    This is equivalent to saying that from $\Cc$ one cannot 
    FO-transduce the class of all linear orders (follows from work of Baldwin and Shelah~\cite{baldwin1985second}). 
    \item A class $\Cc$ is \emph{dependent} if from $\Cc$ one cannot interpret (with an FO-interpretation in powers) the class of all graphs. 
    \item A class $\Cc$ is \emph{monadically dependent} if every monadic lift of $\Cc$ is dependent. 
    This is equivalent to saying that from
    $\Cc$ one cannot FO-transduce the class of all graphs (follows from work of Baldwin and Shelah~\cite{baldwin1985second}).  
\end{itemize}

Also the recently introduced and highly influential notion of twinwidth, introduced by Bonnet et al.~\cite{bonnet2021twin}, can be defined via transductions, as shown by Bonnet et al.~in~\cite{bonnet2024twin}  
\begin{itemize}
    \item A class $\Cc$ has bounded \emph{twinwidth} if and only if every $G\in \Cc$ can be linearly ordered such that the resulting class $\Dd$ of ordered graphs is monadically dependent. 
\end{itemize}

Every class with structurally bounded expansion is structurally nowhere dense, which in turn is monadically stable, and which in turn is monadically dependent. 
Classes with bounded shrubdepth have bounded cliquewidth, which in turn are monadically dependent. 
Recall that a class of graphs is \emph{monotone} if it is closed under taking subgraphs and \emph{hereditary} if it is closed under taking induced subgraphs. 
On hereditary classes the notions of stability and monadic stability, as well as the notions of dependence and monadic dependence coincide~\cite{braunfeld2021characterizations}. 
On monotone classes the notions of stability, monadic, stability, dependence, monadic dependence and nowhere denseness coincide~\cite{adler2014interpreting}. 
\pagebreak
From this we easily get the following nice observation. 

\begin{observation}
    \mbox{}
    \begin{enumerate}
        \item A hereditary class $\Cc$ is monadically stable (dependent) if and only if the class of its incidence graphs is stable (dependent). 
        \item A monotone class $\Cc$ is nowhere dense if and only if the class of its incidence graphs is monadically stable. 
    \end{enumerate}
\end{observation}

\begin{proof}
    Let $\Cc$ be hereditary. Then $\Cc$ is monadically stable (dependent) if and only if it is stable (dependent) by the result of~\cite{braunfeld2021characterizations}. 
    As every interpretation of a stable (dependent) class is stable (dependent), the statement follows from the fact that the class of incidence graphs of $\Cc$ is bi-interpretable with $\Cc$. 

    Let $\Cc$ be monotone. 
    If $\Cc$ is nowhere dense, then it is stable by the result of~\cite{adler2014interpreting}. 
    Since the class of incidence graphs of $\Cc$ is an interpretation of $\Cc$ and every interpretation of a stable class is stable, it follows that the class of incidence graphs $\Dd$ of $\Cc$ is stable. 
    As $\Cc$ is monotone it follows that $\Dd$ is hereditary\footnote{Formally we also have to deal with incidence vertices that are connected to only one graph vertex, but it is easy to see that monadic stability is not affected by these vertices.}. It follows that $\Dd$ is monadically stable by the result of~\cite{braunfeld2021characterizations}. 
    Vice versa, if a class $\Dd$ of incidence graphs is monadically stable it follows that its hereditary closure $\Dd'$ is hereditary and also monadically stable. 
    Then the class $\Cc$ of graphs encoded in $\Dd'$ is monotone and monadically stable, which implies that it is nowhere dense by the result of~\cite{adler2014interpreting}. 
\end{proof}
 
Structurally bounded expansion classes admit bounded-size bounded shrubdepth decompositions as shown by Gajarsk{\'y} et al.~\cite{gajarsky2020first}, structurally nowhere dense classes admit quasibounded-size bounded shrubdepth decompositions as shown by Dreier et al.~\cite{dreier2022treelike} and in fact, a hereditary class admits quasibounded-size bounded shrubdepth decompositions if and only if it is stable as shown by Braunfeld et al.~\cite{braunfeld2024decomposition}. 
One interesting open problems in this direction is whether hereditary dependent classes admit quasibounded-size bounded twinwidth decompositions. 
We refer to \cite{bonnet2022model,braunfeld2022decomposition,dreier2023first,dreier2022treelike,DreierMST23,dreier2024flip,gajarsky2020new,gajarsky2020first,gajarsky2023flipper,gajarsky2024classes,nevsetvril2020linear,nevsetvril2021rankwidth,ohlmann2023canonical,torunczyk2023flip} for more background on the recent exciting development of a theory of structural sparsity. 

\begin{problem}
  Does every monadically dependent class admit quasibounded-size bounded twinwidth decompositions?
\end{problem}

\subsection{Quantifier elimination revisited}

According to \Cref{thm:qe-trees} graphs with bounded treedepth admit quantifier elimination in the following sense. 
    For every FO-formula $\varphi(\bar x)$ and every class $\Cc$ of colored graphs with bounded treedepth there exists a quantifier-free formula $\tilde\phi(\bar x)$ and a 
    linear time computable map $Y$ such that, for every $G\in\Cc$, $Y(G)$ is a guided expansion of $G$ such that for all tuples of vertices $\bar v$
	\[
	G\models \varphi(\bar v)\quad\iff\quad Y(G)\models\tilde\varphi(\bar v).
	\]

This quantifier elimination result easily lifts through bounded-size bounded treedepth decompositions. A similar statement was first proved by Dvo\v{r}\'ak et al.~\cite{dvovrak2013testing} and Grohe and Kreutzer~\cite{GroheK09}. 
Our proof sketch follows Gajarsk\'y et al.~\cite{gajarsky2020first}. 

\begin{theorem}[\cite{dvovrak2013testing,gajarsky2020first,GroheK09}]\label{thm:qe-BE}
    For every FO-formula $\varphi(\bar x)$ and every class $\Cc$ of guided pointer structures with bounded expansion there exists a quantifier-free formula $\tilde\phi(\bar x)$ and a 
    linear time computable map $Y$ such that, for every $G\in\Cc$, $Y(G)$ is a guided expansion of $G$ such that for all tuples of vertices $\bar v$
	\[
	G\models \varphi(\bar v)\quad\iff\quad Y(G)\models\tilde\varphi(\bar v).
	\]
\end{theorem}

\begin{proof}[Proof sketch.]
    We show the statement by induction on the structure of $\varphi$. 
    In the atomic case there is nothing to prove, and in the inductive step, the interesting case is the case of existential quantification. 
    In this case we show how an existential quantifier can be eliminated at the cost of introducing guided unary functions and colors. 
    Consider a formula \[\phi(\bar x)=\exists y(\psi(\bar x,y)),\] where
    $\psi(\bar x,y)$ is a quantifier-free formula
    using unary functions (obtained by induction). 
  
    As $\Cc$ has bounded expansion, it admits a bounded-size bounded treedepth decomposition with parameter $p$, where $p$ is the number of different terms occurring in $\psi(\bar x,y)$.
    For any graph $G\in \Cc$
    we consider the family $\Uu_G=\{U_I \mid I$ set of $p$ colors$\}$, where $U_I$ is the set of vertices induced by the colors from $I$. 
    By assumption, each $G[U_I]$ has bounded treedepth and there is a constant $N$ such that $|\Uu_G|\leq N$ for all $G\in \Cc$.
    
    For each $G\in\Cc$, the quantified variable $y$ must be in one of the sets $U_I\in \cal U_G$. 
    Therefore, we can rewrite 
    $\phi(\bar x)$ as 

    \[\bigvee_{I}\exists y (y\in U_I\wedge \psi(\bar x,y)).\]

    Here, we make the colors available in the structure by adding them as unary predicates $P_1,\ldots, P_t$ and write $y\in U_I$ as an abbreviation for $\bigvee_{i\in I}P_i(y)$. 
    Write $\psi_I$ for $\exists y (x\in U_I\wedge \psi(\bar x,y)$. 
    By \Cref{thm:qe-trees} there exists a 
    there exists a quantifier-free formula $\tilde\psi_I(\bar x)$ and a 
    linear time computable map $Y_I$ such that, for every $G\in\Cc$, $Y_I(G[U_I])$ is a guided expansion of $G[U_I]$ such that $\psi_I$ is equivalent to $\tilde\psi_I$ on $G[U_I]$. 
    By expanding $G$ with the colors and all the $Y_I$ and restricting quantification appropriately, 
    we obtain the guided expansion $Y$ of $G$ such that 

    \[\tilde\phi(\bar x) = \bigvee_{I} \tilde\psi_I(\bar x)\]

    is equivalent to $\phi$ on $Y(G)$. 
\end{proof}

\Cref{thm:qe-BE} extends to classes with structurally bounded expansion~\cite{gajarsky2020first}, as well as to first-order logic with modulo counting on classes with bounded expansion~\cite{nevsetvril2023modulo}.


\section{Game characterizations and wideness}\label{sec:Splitter}

In this section we are going to study game characterizations of the generalized coloring numbers as well as connections to a related notion called wideness. 
Just as treewidth can be characterized by cops-and-robber games in a very intuitive way, it turns out that the generalized coloring numbers can be characterized by games. Cops-and-robber games are played on a graph by a team of cops and a robber. The robber resides on the vertices of the graph. The cops aim to catch the robber who is trying to avoid capture. 
Depending on the concrete setting, e.g.\ visibility, speed of cops and robber, and local restrictions in the game rounds, we obtain different models that correspond to different width measures. 

\subsection{Cops-and-robber games}

Let us first recall the classical cops-and-robber game introduced by Seymour and Thomas~\cite{seymour1993graph} characterizing treewidth. 
The game is played by two players on a graph $G$. The first player controls a robber that stands on a vertex of $G$ and can run at infinite speed along paths in the graph. The second player controls a team of $k$ cops that also occupy vertices of $G$. 
The goal of the cops is to catch the robber. The game proceeds in rounds. In each round, some of the cops may step into helicopters (these cops are temporarily removed from the game) and announce the next position where they want to be placed. Now the robber can run to its new position, however, it is not allowed to run through a path that contains a vertex occupied by a cop. The cops win if they can land on the position of the robber, that is, if it has no moves left when the cops announce to land on its position. 
Otherwise, the game continues forever and the robber wins. 
Seymour and Thomas proved that the cops-and-robber game can be used to characterize treewidth. 

\begin{theorem}[Seymour and Thomas \cite{seymour1993graph}]
The least number of cops needed to catch the robber on $G$ is equal to the treewidth of~$G$ plus one.    
\end{theorem}

In the above game, the robber is visible, agile, and can run at infinite speed. 
In a different variant introduced by Dendris et al.~\cite{dendris1997fugitive} the robber is invisible and lazy, in the sense that it may move only when one of the cops plans to occupy the vertex it currently resides on. 
This variant turns out to be equivalent to the 
visible-and-agile variant of the game~\cite{dendris1997fugitive}. Dendris et al.\ also considered a bounded speed variant of the invisible-and-inert game. 
A robber with speed $r$ in its move is allowed to traverse a cop-free path of length at most $r$. 
They observed that for speed~$1$ the game characterizes the degeneracy of the graph and 
they called the number of cops needed to capture a robber with speed $r$ the \emph{$r$-dimension} of the graph, which corresponds exactly to the strong $r$-coloring number $\col_r(G)$. 
To the best of our knowledge (and to our greatest regret) they did not follow up on studying the $r$-dimension. 

\begin{theorem}[\cite{dendris1997fugitive}]
The least number of cops needed to catch the robber 
in the invisible-and-inert game with speed $1$ on $G$ is equal to the degeneracy of $G$ plus $1$.     

\smallskip
The least number of cops needed to catch the robber 
in the invisible-and-inert game with speed $r$ on $G$ is equal to the $r$-dimension of $G$, which is equal to $\col_r(G)$.     
\end{theorem}

Surprisingly, to the best of our knowledge, the bounded speed variant of the visible-and-agile game was only very recently studied by Toru\'nczyk~\cite{torunczyk2023flip}. 
It does not seem to correspond one-to-one to an established width measure but Toru\'nczyk proved that the number of cops required to catch the robber in the speed-$r$ variant of the visible-and-agile game is sandwiched between the $r$-admissibility and the weak $2r$-coloring number. 




\begin{theorem}[\cite{torunczyk2023flip}]
    The least number of cops needed to catch the robber 
in the visible-and-agile game with speed $1$ on $G$ is equal to the degeneracy of $G$ plus $1$.    

\smallskip
The least number of cops needed to catch the robber 
in the visible-and-agile game with speed~$r$ on $G$ is bounded from below by $\adm_r(G)+1$ and bounded from above by $\wcol_{2r}(G)$. 
\end{theorem}

We sketch the proof of the second statement. 

\begin{proof}
  The lower bound follows easily from the characterization of $r$-admissibility by $(k,r)$-fan sets, which provide a shelter for the robber (\Cref{thm:adm}). 

  Let us prove the upper bound. Let $\pi$ be a linear order witnessing that $\wcol_{2r}(G) =: c$. The following is a winning strategy for the cops:
  in each round, if the robber is placed at a vertex~$v$, then the cops fly to the vertices $w\leq v$ that are weakly-$2r$ reachable from $v$. 
  
  We show that this is a winning strategy for the cops. 
  Let $P_1,P_2,\ldots$ be the paths in $G$ such that~$P_i$ is the path of length at most $r$ along which the robber moved from its $i$-th position $v_i$ to its $(i+1)$-st position $v_{i+1}$. 
  Denote by $m_i$ the smallest vertex of the path $P_i$ with respect to $\pi$. We claim that $m_{i+1}>_\pi m_i$ for all $i$. 
  Otherwise, $m_{i+1}$ is weakly-$2r$ reachable from $v_i$ (as witnessed by the concatenated path $P_i P_{i+1}$ cut at $m_{i+1}$). 
  
  Hence, when the robber traversed the path $P_{i+1}$ the vertex $m_{i+1}$ was occupied, which is not possible. 
  Therefore, we have $m_1<_\pi m_2<_\pi \ldots$, so the cops win after at most $n$ rounds.
\end{proof}

\begin{corollary}\label{cor:copw}
A graph class $\Cc$ has bounded expansion if and only if for every $r\in\N$ there exists a constant $c(r)$ such that for every $G\in \Cc$ the least number of cops needed to catch the robber in the invisible-and-inert game with speed $r$ or in the visible-and-agile game with speed $r$ is bounded by $c(r)$. 

\medskip
A class $\Cc$ is nowhere dense if and only if 
for every $r\in\N$ and $\epsilon>0$ there exists a constant~$c(r,\epsilon)$ such that 
for every $H\subseteq G\in \Cc$ the least number of cops needed to catch the robber in the invisible-and-inert game with speed $r$ or in the visible-and-agile game with speed $r$ on $H$ is bounded by $c(r,\epsilon)\cdot |H|^\epsilon$. 
\end{corollary}

\medskip
The \emph{lift-free game}, introduced by Ganian et al.~\cite{ganian2009digraph}, is played like the visible-and-agile game but the cop player is not allowed to move a cop from a vertex once it has landed. 
That is, in the lift-free game the cops are placed on the vertices of the graph one after another and we count the number of cops needed to catch the robber. 

\begin{theorem}[\cite{ganian2009digraph}]
    The least number of cops needed to catch the robber 
in the lift-free game on $G$ is equal to the treedepth of $G$. 
\end{theorem}

In order to appropriately localize this game we take a different view on counting the number of rounds in the lift-free game. 
For $r\in \N\cup\{\infty\}$, in the \emph{radius-$r$ counter game} every vertex is equipped with a counter that is initially set to $0$. 
In the radius-$r$-counter game the cops are placed on the vertices of the graph one after the other until all vertices are occupied by a cop. 
Denote the set of placed cops in round $i$ by $X_i$ (and $X_0=\emptyset$. 
In round $i$ when a cop is placed on vertex~$v_i$, we increase the counters of all vertices in the $r$-neighborhood of $v_i$ in $G-X_{i-1}$ by $1$ (also the counter of the vertex itself). 
The \emph{radius-$r$ counter-depth} of a play is the largest number that a counter takes in the play and the \emph{radius-$r$ counter-depth} of $G$ is the smallest counter-depth that the cops can ensure in the radius-$r$-counter game with optimal play. 

\begin{theorem}
    The radius-$r$ counter-depth of $G$ is equal to $\wcol_r(G)$. 
\end{theorem}

In particular, the radius-$1$ counter-depth of $G$ is equal to the degeneracy of $G$ plus $1$ and the radius-$\infty$ counter-depth of $G$ is equal to the treedepth of $G$. 

\begin{proof}
    Let $G$ be a graph and let $\pi=v_1,\ldots, v_n$ be an order of $V(G)$. 
    Then $v_i\in \WReach_r[G,\pi,v_j]$ if and only if there is a path $P$ of length at most $r$ between $v_j$ and~$v_i$ such that $v_j$ is the minimum vertex on $P$, which is exactly the set of vertices that increase the counter of $v_j$ when the cops occupy the vertices in order $v_1,\ldots, v_n$. 
%
%

\end{proof}

\begin{corollary}
    A graph class $\Cc$ has bounded expansion if and only if for every positive integer~$r$ there exists a constant $c(r)$ such that for every $G\in \Cc$ the radius-$r$ counter depth is bounded by~$c(r)$. 

\medskip
A class $\Cc$ is nowhere dense if and only if 
for every positive integer $r$ and every $\epsilon>0$ there exists a constant~$c(r,\epsilon)$ such that 
for every $H\subseteq G\in \Cc$ the radius-$r$ counter depth of $H$ is bounded by $c(r,\epsilon)\cdot |H|^\epsilon$. 
\end{corollary}

\subsection{Splitter Game and Quasi-Wideness}

Another way of localizing lift-free games is to restrict the game graph in each round to a local neighborhood of the last cop move. 
This leads to the Splitter game, which characterizes nowhere dense classes, as shown by Grohe et al.~\cite{grohe2017deciding}. 
We present an equivalent characterization in terms of uniform quasi-wideness and derive bounds for the length of the Splitter game as well as for uniform quasi-wideness in terms of the generalized coloring numbers.

\smallskip
Let $G$ be a graph. The \emph{radius-$r$ Splitter game} on $G$ is played by two players: Connector and Splitter, as follows.
The game starts on $G_0 := G$. 
In round $i+1$ of the game, Connector chooses a vertex $c_{i+1} \in V(G_i)$.
Then Splitter picks a vertex $s_{i+1} \in N_r^{G_i} (c_{i+1})$ and deletes it, that is, the game continues on $G_{i+1} := G_i[N_r^{G_i} (c_{i+1}) \setminus
\{s_{i+1}\}]$.
Splitter wins if $G_{i+1} = \emptyset$, otherwise, the game continues.

\begin{observation}
  Let $G$ be a graph. Then $\td(G)$ is equal to the least number such that Splitter wins the radius-$\infty$ Splitter game. 
\end{observation}

As shown in~\cite{grohe2017deciding} the Splitter game characterizes nowhere dense classes. 

\begin{theorem}
  [\cite{grohe2017deciding}]
  A class $\Cc$ of graphs is nowhere dense if and only if for every positive interger $r$ there exists a constant~$\ell(r)$ such that Splitter wins the radius-$r$ Splitter game on every $G\in \Cc$ in at most $\ell(r)$ rounds. 
\end{theorem}

The proof is based on the fact that nowhere dense classes are uniformly quasi-wide (defined below). 
It was observed by Kreutzer et al.~\cite{KreutzerPRS16} that the weak coloring number bounds the number of rounds that Splitter needs to win.

\begin{lemma}[\cite{KreutzerPRS16}]
  Let $G$ be a graph. Then Splitter has a strategy to win the radius-$r$ Splitter game on~$G$ in $\ell=\wcol_{2r}(G)+1$ rounds.
\end{lemma}
\begin{proof}
  Given a graph $G$ and an integer $r$. Let $\ell := \wcol_{2r}(G)$ and let $\pi$ be a witnessing order. 
  Splitter plays the following strategy: in any round $i+1$, Splitter picks $s_{i+1}=\min_{\pi}(N_r^{G_i}(c_{i+1}))$. 
  Then for every remaining vertex $v\in V(G_{i+1})=N_r^{G_i}(c_{i+1})\setminus \{s_{i+1}\}$ we have $s_{i+1}\in \WReach_{2r}[G,\pi,v])$, as there is a path of length at most $2r$ connecting $v$ and $s_{i+1}$ (via the central vertex $c_{i+1}$) and $s_{i+1}$ is minimum on this path (compare this to the counter game). 
  We conclude that the game must end after at most~$\ell$~rounds. 
\end{proof}

The fact that nowhere dense classes are characterized by the Splitter game is derived from a characterization of these classes by the notion of uniform quasi-wideness, which was introduced by Dawar~\cite{dawar2010homomorphism}. 
A set $B\subseteq V(G)$ is called \emph{distance-$r$ independent} if the vertices of $B$ pairwise have distance greater than $r$. 

\begin{definition}
  A class $\Cc$ of graphs is \emph{uniformly quasi-wide} if for every
$m\in \N$ and every $r\in \N$ there exist numbers $N(m,r)$ and $s(r)$ such 
that the following holds (see Figure 16). 
\medskip
\begin{quotation}
\noindent \textit{Let $G\in\Cc$ and let $A\subseteq V(G)$ with $|A|\geq N(m,r)$. Then
there exists a set $S\subseteq V(G)$ with $|S|\leq s(r)$ and a set $B\subseteq 
A\setminus S$ of size at least $m$ that is distance-$r$ independent in $G-S$. }
\end{quotation}
\end{definition}

\begin{figure}[ht!]
  \begin{center}
  \includegraphics[width=.75\textwidth]{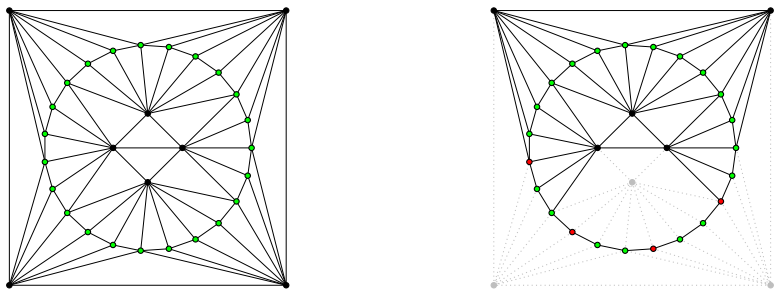}
\end{center}
\caption{Uniform quasi-wideness (case $r=2$). By deleting a subset $S$ of $s=3$ vertices (in gray on the right), the set $A$ of vertices on the circle will include a subset $B$ of $m=4$ vertices (in red on the right) that is $2$-independent.}
\end{figure}

Ne\v{s}et\v{r}il and Ossona de Mendez~\cite{nevsetvril2010first} proved that the notions of nowhere denseness and uniform quasi-wideness coincide. 

\begin{theorem}[\cite{nevsetvril2010first}]\label{thm:uqw}
  A class $\Cc$ of graphs is nowhere dense if and only if it is uniformly quasi-wide. 
\end{theorem}

\Cref{thm:uqw} was proved in model theory by Podewski and Ziegler~\cite{podewski1978stable} already in 1978. 
Podewski and Ziegler introduced the notion of \emph{superflatness}, which is equivalent to nowhere denseness on classes of finite graphs. 
This result was brought to the attention of the graph theory community by Adler and Adler~\cite{adler2014interpreting} and inspired the recent investigation of connections between model-theoretic notions and structural sparsity. 

It was shown by Kreutzer et al.~\cite{kreutzer2018polynomial} (improved and constructive bounds were provided by Pilipczuk et al.~\cite{pilipczuk2018number}) that the bounds for the function $N(m,r)$ can be made polynomial, which is important for many algorithmic applications. 
While the proofs for nowhere dense classes is quite technical, a quite simple proof that bounded expansion classes are uniformly quasi-wide is given by Nadara et al.~\cite{nadara2019empirical}. 

\begin{theorem}[\cite{nadara2019empirical}]
    Assume we are given a graph $G$, a set $A \subseteq V(G)$,
    integers $r \geq 1$ and $m \geq 2$, and an order $\pi$ of $V(G)$ with
    $c = \wcol_r(G,\pi)$. 
    Furthermore, assume that
    $|A| \geq 4\cdot (2cm)^{c+1}$. 
    Then in polynomial time, one can compute sets $S \subseteq V(G)$
    and $B \subseteq A \setminus S$ such that $|S| \leq c$, $|B| \geq m$, and $B$ 
    is distance-$r$ independent in $G-S$.
\end{theorem}

Let us mention the following application of uniform quasi-wideness due to Ne\v{s}et\v{r}il and Ossona de Mendez~\cite{nevsetvril2016structural}, showing that nowhere dense classes are locally well separable. 

\begin{definition}
A class $\Cc$ of graphs admits \emph{balanced neighborhood separators} if for every non-negative integer $r$ and every real $\epsilon>0$ there exists a number $s=s(r,\epsilon)$ such that the following holds. For every
graph $G\in \Cc$ and every subset $A\subseteq V(G)$ there exists a set $S\subseteq V(G)$ of size at most $s$ such that $|N_r^{G-S}(v)\cap A|\leq \epsilon |A|$ for all $v\in V(G)\setminus S$. 
\end{definition}

\begin{theorem}[\cite{nevsetvril2016structural}]
A class $\Cc$ of graphs is nowhere dense if and only if $\Cc$ admits balanced
neighborhood separators.
\end{theorem}

\subsection{Flipwidth, flipper-game, and flip-flatness}

We conclude this section with a short overview of how the presented games and the notion of uniform quasi-wideness have recently been generalized to dense graphs with the hinted at beautiful connections to model theory. 
We begin with the \emph{flipwidth game}, which was recently introduced by Toru\'nczyk~\cite{torunczyk2023flip}.
The idea of game is to (temporarily) flip adjacencies between arbitrarily large sets of vertices (instead of occupying them by cops) -- an operation that is well suited to simplify dense graphs. 
Note also the close connection between flips and quantifier-free transductions. 

\begin{definition}
Let $G$ be a graph and $X,Y\subseteq V(G)$. 
A \emph{flip} between $X$ and $Y$ results in the graph~$G'$ obtained from $G$ by inverting the adjacency of every pair $(x,y)$ of vertices with $x\in X$ and $y\in Y$:
every pair $(x,y)$ that is adjacent in $G$ becomes non-adjacent in $G'$, and vice-versa.
A \emph{$k$-flip} of a graph $G$ is obtained by partitioning $V(G)$ 
into $k$ parts, and then performing flips between some pairs $X,Y$ of parts of the partition (possibly with $X=Y$).
\end{definition}

The flipwidth game is defined as follows. 

\begin{definition}
The \emph{flipwidth game} with radius $r$ and width $k$ on a graph $G$ is a game played by two players, \emph{Flipper} and \emph{Runner}, as follows.
We let $G_0:=G$ and let Runner pick any vertex $v_0$ of~$G$.
In the $i$th round, Flipper announces a $k$-flip $G_i$ \emph{of the original graph} $G$.
Runner can move from their previous position $v_{i-1}$ to a new position $v_i$, by traversing a path of length at most~$r$ in~$G_{i-1}$. 
The game is won by Flipper if the runner's new position $v_i$ is isolated in the announced $k$-flip $G_i$.

The \emph{flipwidth} of radius $r$ of a graph $G$, denoted $\fw_r(G)$, is the smallest number $k$ for which Flipper has a winning strategy in the flipwidth game with radius $r$ and width $k$. 

A class $\Cc$ of graphs has \emph{bounded flipwidth} if $\fw_r(\Cc)<\infty$ for all $r\in \N$. 
It has \emph{almost bounded flipwidth}
if for every $\epsilon>0$ and $r\in\N$,
we have $\fw_r(H)\in \Oof_{r,\epsilon}(|H|^{\epsilon})$ for every
$n$-vertex graph $H$ which is an induced subgraph of a graph $G\in \Cc$.
\end{definition}

Toru\'nczyk proved the following connections with graph theory and model theory. 
First, he proved that the flipwidth game is a robust generalization of the bounded-speed variants of the cops-and-robber game by considering weakly sparse graph classes. 
A graph class is \emph{weakly sparse} if it
excludes some fixed biclique $K_{t,t}$ (balanced complete bipartite graph with parts of size $t$) as a subgraph.

\begin{theorem}[\cite{torunczyk2023flip}]
  Let $\Cc$ be a class of graphs.
  Then: 
  \begin{enumerate}
    \item $\Cc$ has bounded degeneracy if and only if $\Cc$ is weakly sparse and $\fw_1(\Cc)<\infty$,
     \item $\Cc$ has bounded treewidth if and only if $\Cc$ is weakly sparse and $\fw_\infty(\Cc)<\infty$,
    \item $\Cc$ has bounded expansion if and only if $\Cc$ is weakly sparse and $\Cc$ has bounded flipwidth.
    \item $\Cc$ is nowhere dense if and only if $\Cc$ is weakly sparse and $\Cc$ has almost bounded flipwidth. 
  \end{enumerate}  
\end{theorem}

On dense graphs he obtained the following wonderful results. 

\begin{theorem}
    Let $\Cc$ be a class of graphs. Then 
    \begin{enumerate}
        \item $\Cc$ has bounded cliquewidth if and only if $\fw_\infty(\Cc)<\infty$.
        \item if $\Cc$ has bounded twinwidth, then $\Cc$ has bounded flipwidth. Furthermore, if $\Cc$ is a class of ordered graphs then $\Cc$ has bounded twinwidth if and only if $\Cc$ has bounded flipwidth. 
        \item if $\Cc$ is structurally nowhere dense, then $\Cc$ has almost bounded flipwidth. 
        \item if $\Cc$ has bounded flipwidth and $\mathsf{T}$ is a first-order transduction, then $\mathsf{T}(\Cc)$ has bounded flipwidth. 
        Hence, every class of bounded flipwidth is monadically dependent. 
    \end{enumerate}
\end{theorem}

Toru\'nczyk conjectures that a class has almost bounded flipwidth if and only if it is monadically dependent. 

\bigskip
We now come to a second game, the \emph{flipper game}, which was introduced by Gajarsk\'y et al.\ and which can be seen as a generalization of the Splitter game, and which characterizes monadically stable classes~\cite{gajarsky2023flipper}. 

\begin{definition}
    The \emph{radius-$r$ flipper game} is played by two players, \emph{Flipper} and \emph{Connector}, on a graph $G$ as follows.
The game starts at $G_0:=G$. In round $i+1$ of the game, Connector chooses a vertex $v_{i+1}\in V(G_i)$. 
Flipper chooses a flip and applies it to $G_i[N_r^{G_i}(v_{i+1})]$ and the game continues on the resulting graph $G_{i+1}$. Flipper wins if $G_{i+1}$ consists of a single vertex. 
\end{definition}

\begin{theorem}[\cite{gajarsky2023flipper}]
    A class $\Cc$ of graphs is monadically stable if and only if for every positive integer $r$ there exists a constant $\ell(r)$ such that for every graph $G\in \Cc$, Flipper can win the radius-$r$ Flipper game on $G$ within $\ell(r)$~rounds. 
\end{theorem}

The theorem is proved by another characterization of monadically stable classes via \emph{uniform flip-flatness}, which was introduced by Dreier et al.~\cite{DreierMST23}. 
We remark that flip-flatness is not called flip-wideness (as it should be) in order to avoid confusion with flipwidth. 

\begin{definition}
A class $\Cc$ of graphs is \emph{uniformly flip-flat}
if for every $r \in \N$ there exists a function $N_r \colon \N \rightarrow \N$ and a constant $k_r\in \N$ such that for all $m\in\N$, $G\in \Cc$, and $A\subseteq V(G)$ with $|A|\geq N_r(m)$, there exists a $k_r$-flip $H$ of $G$ and $B\subseteq A$ with $|B|\geq m$ such that $B$ is distance-$r$ independent in~$H$. 
\end{definition}

\begin{theorem}[\cite{DreierMST23}]
    A class of graphs is uniformly flip-flat if and only if it is monadically stable. 
\end{theorem}

As observed by Dreier et al.~\cite{dreier2024flip}, $\infty$-flip-flatness corresponds to bounded shrubdepth (here and below distance $\infty$ means that vertices lie in different connected components).

\begin{theorem}[\cite{dreier2024flip}]
    A class of graphs is uniformaly $\infty$-flip-flat if and only if it has bounded shrubdepth.
\end{theorem}

Dreier et al.\ also introduced the new notion of \emph{flip-breakability} and used it to characterize monadically dependent classes. 

\begin{definition}
A class of graphs $\Cc$ is \emph{uniformly flip-breakable}
if for every $r \in \N$ there exists a function $N_r \colon \N \rightarrow \N$ and a constant $k_r\in \N$ such that for all $m\in\N$, $G\in \Cc$, and $W\subseteq V(G)$ with $|W|\geq N_r(m)$, there exists a $k_r$-flip $H$ of $G$ and $A,B\subseteq W$ with $|A|,|B|\geq m$ such that in~$H$ all vertices of $A$ have distance greater than $r$ from all vertices in $B$. 
\end{definition}

They also considered the notion of \emph{deletion-breakability}, which relates to flip-breakability as uniform quasi-wideness relates to uniform flip-flatness. 

\begin{definition}
A class of graphs $\Cc$ is \emph{uniformly deletion-breakable}
if for every $r \in \N$ there exists a function $N_r \colon \N \rightarrow \N$ and a constant $k_r\in \N$ such that for all $m\in\N$, $G\in \Cc$, and $W\subseteq V(G)$ with $|W|\geq N_r(m)$, there exists a set $S$ with $|S|\leq k_r$ and $A,B\subseteq W\setminus S$ with $|A|,|B|\geq m$ such that all vertices of $A$ have distance greater than $r$ from all vertices in $B$ in $G-S$. 
\end{definition}

Dreier et al.~\cite{dreier2024flip} proved the following beautiful characterization theorems. 

\begin{theorem}[Dreier et al.~\cite{dreier2024flip}]
A class of graph is 
\begin{enumerate}
    \item uniformly flip-breakable if and only if it is monadically dependent. 
    \item uniformly $\infty$-flip-breakable if and only if it has bounded cliquewidth.
    \item uniformly deletion-breakable if and only if it is nowhere dense. 
    \item uniformly $\infty$-deletion-breakable if and only if it has bounded treewidth.
\end{enumerate}
    
\end{theorem}

\section{Vapnik-Chervonenkis dimension}\label{sec:neigh-comp}

In this section, we study the classical notions of \emph{VC-dimension} and \emph{order-dimension} and their connection to the generalized coloring numbers. 
VC-dimension was originally introduced by Vapnik and Chervonenkis~\cite{chervonenkis1971theory} as a measure of complexity of set systems. 
VC-dimension plays a key role in learning theory as it characterizes PAC-learnability by the result of Blumer et al.~\cite{blumer1989learnability}.
Independently it was introduced by Shelah in model theory where it defines dependent theories~\cite{shelah1971stability}. 
The connection between PAC-learning and dependence was observed by Laskowski~\cite{laskowski1992vapnik}.

A second related measure is \emph{Littlestone-dimension} (we are going to work with \emph{order-dimension}, which is functionally equivalent), which characterizes online-learning by the result of Littlestone~\cite{littlestone1988learning}. 
In model theory it defines stability, as observed by Chase and Freitag~\cite{chase2019model}. 

As mentioned before, nowhere dense graph classes are both stable and dependent, hence, set systems definable by first-order formulas in nowhere dense classes have bounded order-dimension and bounded VC-dimension. 
Besides the direct applications of VC-dimension and order-dimension for definable set systems in nowhere dense classes, these connections have led to further interesting combinatorial and algorithmic applications that we will present in the following. 
The generalized coloring numbers as well as uniform quasi-wideness have played a big role in the development of the theory in this direction. 

\subsection{VC-dimension and number of types}

The VC-dimension of a set system is defined as follows. 

\begin{definition}
    Let $X$ be a set and let  $\Ff\subseteq \mathcal{P}(X)$ 
    be a family of subsets of $X$.
    A subset $A\subseteq X$ is \emph{shattered by $\Ff$} if
    $\{A\cap F\colon F\in \Ff\}=\mathcal{P}(A)$; that is, every subset of $A$ can be obtained as the intersection of some set from $\Ff$ with $A$. 
    The \emph{VC-dimension},
    of $\Ff$ is the maximum size of a subset $A\subseteq X$ that is shattered by
    $\Ff$.
\end{definition}

\begin{definition}
    Let $X$ be a set and let  $\Ff\subseteq \mathcal{P}(X)$ 
    be a family of subsets of $X$.
    The \emph{order-dimension} of $\Ff$ is the largest integer $\ell$ such that
    there exist elements $x_1,\ldots, x_\ell\in X$ and $F_1,\ldots, F_\ell\in \Ff$ with $x_i\in F_j\Leftrightarrow i\leq j$. 
\end{definition}

The key property of VC-dimension is that it implies polynomial upper bounds on the number of different \emph{traces} that a set system can have on a given parameter set. This is the statement of the famous Sauer-Shelah Lemma.

\begin{theorem}[Sauer-Shelah Lemma, \cite{sauer1972density, shelah1972combinatorial,chervonenkis1971theory}]\label{thm:sauer-shelah}
  For any family $\Ff$ of subsets of a set $X$, if the VC-dimension of $\Ff$ is $d$,
  then for every finite $A\subset X$,
$$|\{A\cap F \mid F\in {\cal F}\}|\le c\cdot |A|^d,  
\qquad\textit{where $c$ is a universal constant.}$$
\end{theorem}

\begin{definition}
    The \emph{vc-density} (also called the 
    \emph{vc-exponent})
    of a set system $\cal F$
    on an infinite set $X$ is the infimum of all reals $\alpha>0$ such that 
    $|\setof{A\cap F}{F\in \cal F}|\in \Oof(|A|^\alpha)$, for all finite $A\subset X$. 
    \end{definition}

The vc-density of neighborhood set systems in bounded expansion classes was first studied by Reidl et al.~\cite{Reidl16,reidl2019characterising} under the name \emph{neighborhood complexity}. 

\begin{theorem}[\cite{reidl2019characterising}]\label{thm:reidl}
    Let $G$ be a graph and $A\subseteq V(G)$. Then \[|\{N_r(v)\cap A \mid v\in V(G)\}|\leq \frac{1}{2}(2r+2)^{\wcol_{2r}(G)}\cdot \wcol_{2r}(G)\cdot |A|+1.\] 
\end{theorem}

The problem with extending this result to nowhere dense classes in a useful way was that the weak coloring number, which can only be bounded by $\Oof(n^\epsilon)$ in nowhere dense classes, appears in the exponent in the above inequality. 

As shown by Pilipczuk and Siebertz~\cite{pilipczuk2021kernelization}, the VC-dimension of the $r$-neighborhoods set system in nowhere dense classes is very small. 

\begin{lemma}[\cite{pilipczuk2021kernelization}]
    If $K_t\not\minor_r G$, then the VC-dimension of the set system \[\{N_r(v)\mid v\in V(G)\}\]
    of $r$-neighborhoods is at most $t-1$. 
\end{lemma}

This yields already polynomial neighborhood complexity for nowhere dense graph classes. 
On the way to proving that neighborhood set systems in nowhere dense classes have in fact vc-density~$1$, Eickmeyer et al.~\cite{EickmeyerGKKPRS17} first proved that also the set system of weak coloring numbers in nowhere dense classes has bounded VC-dimension. 

\begin{lemma}[\cite{EickmeyerGKKPRS17}]\label{lem:eickmeiyer}
    Let $\Cc$ be a nowhere dense class of graphs and let $G\in \Cc$. Let $\pi$ be a linear order of $V(G)$ and let 
    \[\mathcal{W}_{r,\pi}=\{\WReach_r[G,\pi,v]\mid v\in V(G)\}\]
    be the set system of weak $r$-reachability sets. 
    Then there exists a constant $vc(r)$, depending only on $\Cc$ and $r$ (and not on $G$ and $\pi$) such that $\mathcal{W}_{r,\pi}$ has VC-dimension at most $vc(r)$. 
\end{lemma}

Based on this result, the proof of \Cref{thm:reidl} could be lifted to nowhere dense classes. 

\begin{theorem}[\cite{EickmeyerGKKPRS17}]
    Let $\Cc$ be a nowhere dense class of graphs. Then there exists a function~$f(r,\epsilon)$ such that for every $G\in \Cc$, $A\subseteq V(G)$, $r\geq 1$ and $\epsilon>0$ 
    \[|\{N_r(v)\cap A \mid v\in V(G)\}|\leq f(r,\epsilon)\cdot |A|^{1+\epsilon}.\]
\end{theorem}

The setting studied in model theory is to consider definable set systems. 
One of the key results of Shelah is that stable and dependent theories can be characterized by upper bounds on the number of types (as defined next). 

For a first order formula $\phi(\bar x,\bar y)$ 
with free variables partitioned into \emph{object variables} $\bar x$ and \emph{parameter variables} $\bar y$, a structure $\strA$,
and a subset of its domain $B\subseteq A$, define
the set system of \emph{$\phi$-types} with parameters from~$B$ that are realized in $\strA$ as 
\begin{align*}
S^\phi(\strA/B)=\left\{\big\{\bar v\ \in B^{|\bar y|}\, \colon\, \strA\models\phi(\bar u,\bar v)\big\} \colon\, \bar u\in A)^{|\bar x|}\right\}\ \subseteq\  \Pp(B^{|\bar y|}).
\end{align*}

Every tuple $\bar u\in A^{|\bar x|}$ defines the set of those tuples $\bar v\in B^{|\bar y|}$ for which $\phi(\bar u,\bar v)$ is true.
The set $S^\phi(\strA/B)$ consists of all subsets of $B^{|\bar y|}$ that can be defined in this way.
We can also work with the following definitions of stability and dependence. 

\begin{definition}
    A class $\Cc$ of structures is \emph{dependent} if for every $\phi(\bar x,\bar y)$ there exists a constant $d$ such that for all $\strA\in \Cc$ and all $B\subseteq A$ the VC-dimension of $S^\phi(\strA/B)$ is bounded by $d$. 
\end{definition}

\begin{definition}
    A class $\Cc$ of structures is \emph{stable} if for every $\phi(\bar x,\bar y)$ there exists a constant $d$ such that for all $\strA\in \Cc$ and all $B\subseteq A$ the order-dimension of $S^\phi(\strA/B)$ is bounded by $d$. 
\end{definition}

A class $\Cc$ is called \emph{monadically dependent} or \emph{monadically stable}, respectively, if the class of all colorings of structures from $\Cc$ is dependent or stable, respectively. 
As proved by Baldwin and Shelah~\cite{baldwin1985second} in this case it suffices to consider formulas $\phi(x,y)$ with only two free variables, which allows the treatment of monadically dependent and monadically stable classes via transductions, as presented in \Cref{sec:structural-decompos}. 

We can now define the VC-dimension, order-dimension and vc-density of formulas over structures in the natural way. 
For example, the vc-density of a formula is defined as follows. 

\begin{definition}
The vc-density of a formula $\phi(\bar x,\bar y)$ over a class of structures~$\Cc$
is the infimum of all reals $\alpha>0$
such that $|S^\phi(\strA/B)|\in \Oof(|B|^\alpha)$,
for all $\strA\in \Cc$ and all finite $B\subset A$.
\end{definition}

Using the locality of first-order logic, uniform quasi-wideness and neighborhood complexity, it was proved by Pilipczuk et al.~\cite{pilipczuk2018number} that the vc-density of first-order formulas in nowhere dense classes is very small. 

\begin{theorem}[\cite{pilipczuk2018number}]\label{thm:vc-density-nd}
Let $\Cc$ be a class of graphs and let $\phi(\bar x,\bar y)$ be a first-order formula. 
\begin{enumerate}[(1)]
\item If $\Cc$ is nowhere dense, then for every $\epsilon>0$ 
there exists a constant~$c$ such that for every $G\in \Cc$ and every nonempty
$A\subseteq V(G)$, we have $|S^\phi(G/A)|\leq c\cdot |A|^{|\bar x|+\epsilon}.$
\item If $\Cc$ has bounded expansion, then there exists a constant~$c$ such that for every $G\in \Cc$ and every nonempty $A\subseteq V(G)$, we have $|S^\phi(G/A)|\leq c\cdot |A|^{|\bar x|}$.
\end{enumerate}
 \end{theorem}

As shown by Dreier et al.~\cite{dreier2023first} as a key ingredient for their first-order model checking algorithm on monadically stable classes, also monadically stable classes behave very well.

\begin{theorem}[\cite{dreier2023first}]\label{thm:vc-density-ms}
    Let $\Cc$ be a monadically stable class of graphs and let $\phi(x,y)$ be a first order formula with two free variables. Then for every $\epsilon>0$ 
there exists a constant~$c$ such that for every $G\in \Cc$ and every nonempty
$A\subseteq V(G)$, we have $|S^\phi(G/A)|\leq c\cdot |A|^{1+\epsilon}.$
\end{theorem}

Dreier et al.~\cite{dreier2023first} conjecture that the theorem generalizes to formulas $\phi(\bar x,\bar y)$, just as \Cref{thm:vc-density-nd} for nowhere dense classes. This conjecture would be implied by the conjecture that the notions of monadic stability and structural nowhere denseness coincide~\cite{ossona2021first,de2021sparsity}. 
They also conjecture that \Cref{thm:vc-density-ms} generalizes to monadicallly dependent classes. 
This is known for some special cases, e.g.\ for classes with bounded twinwidth (even in the general setting of \Cref{thm:vc-density-nd} and with the strong bounds as for classes with bounded expansion) as proved by Gajarsk\'y et al.~\cite{GajarskyPPT22}. 

\subsection{Distance-r domination}

A \emph{distance-$d$ dominating set} in a graph $G$ is a set $D\subseteq V(G)$ such that every vertex of $G$ is at distance at most $d$ to a vertex of $D$. 
The algorithmic problem of deciding whether a graph admits a distance-$d$ dominating set of size at most $k$ is one of the benchmark problems on nowhere dense classes and many of the combinatorics we have seen so far have been developed to approach this algorithmic problem. 
Recall that a problem is fixed-parameter tractable with respect to a parameter~$k$ if it can be solved in time $f(k)\cdot n^{\Oof(1)}$ for some computable function $f$. It admits a kernel of size $g(k)$ if one can compute in polynomial time an equivalent instance of size $g(k)$.

After Ne\v{s}et\v{r}il and Ossona de Mendez~\cite{nevsetvril2008structural} demonstrated first that the dominating set problem is fixed-parameter tractable based on low treedepth decompositions,  
Dawar and Kreutzer~\cite{DawarK09} extended these results to nowhere dense classes and the distance-$d$ dominating set problem. 
They introduced the notion of \emph{domination cores} and developed an irrelevant vertex technique that was highly relevant for the further algorithmic development of nowhere dense classes based on uniform quasi-wideness. 
A $k$-domination core is a set $C\subseteq V(G)$ such that every set of size at most $k$ that dominates $C$ also dominates the whole graph $G$. 
Obviously, the existence of small domination cores is highly useful for the kernelization of the dominating set problem. 
The next step by Drange et al.\ was to establish linear kernels for distance-$d$ domination for bounded expansion classes based on neighborhood complexity combinatorics~\cite{DrangeDFKLPPRVS16}. 
With the development of polynomial bounds for uniform quasi-wideness functions by Kreutzer et al.~\cite{kreutzer2018polynomial} first a polynomial kernel~\cite{kreutzer2018polynomial} and then an almost linear kernel for distance-$d$ domination on nowhere dense classes was developed by Eickmeyer et al.~\cite{EickmeyerGKKPRS17}. 
Telle and Villanger proved that small domination cores exist even for biclique-free graphs~\cite{telle2012fpt} and most generally for semi-ladder-free graphs as proved by Fabianski et al.~\cite{FabianskiPST19} (however, these can only be computed in fpt time). 
Polynomial kernels for the independent set and dominating set problem
in powers of nowhere dense classes can be efficiently computed as shown by Dreier et al.~\cite{dreier2022combinatorial}. 
A unified view on kernelization for packing and covering problems in nowhere dense classes is provided by Ahn et al.~\cite{AhnKK23} and by Einarson and Reidl~\cite{EinarsonR20}. 
The weak coloring numbers also play a key role in the constant factor approximation algorithms for the distance-$d$ dominating set problem~\cite{akhoondian2018distributed,dvovrak2013constant,dvovrak2019distance,kreutzer2017structural,pilipczuk2021kernelization}

\subsection{Neighborhood covers}

The \emph{dual} of a set system $\Ff$ on a universe $X$ is the set system $\Ff^*$ that has one element for each set of $\Ff$ and contains the sets $S_x=\{A \in \Ff \mid x\in A\}$ for all $x\in X$. Intuitively, if we represent the set system $\Ff$ by its incidence matrix, then we obtain the dual set system $\Ff^*$ by the transposed incidence matrix. 
The \emph{dual VC-dimension} of $\Ff$ is the VC-dimension of $\Ff^*$. 
Observe that the set system $\{N_r(v) \mid v\in V(G)\}$ is equal to its dual, hence, its VC-dimension is equal to its dual VC-dimension. In general, if the VC-dimension of $\Ff$ is $d$, then its dual VC-dimension is less than $2^{d+1}$~\cite{assouad1983densite}. The dual set system of $\Ww_{r,\pi}$ considered in \Cref{lem:eickmeiyer} is the system $\Ww^*_{r,\pi}=\{X_v\mid v\in V(G)\}$, where $X_v=\{w\mid v\in \WReach_r[G,\pi,w]\}$. This system hence also has bounded VC-dimension in nowhere dense classes. 
Its most interesting property however is that it forms an \emph{$r$-neighborhood cover}. 

\begin{definition}
    An \emph{$r$-neighborhood cover} of $G$ is a set $\Xx$ of vertex subsets of $G$, called \emph{clusters} such that the $r$-neighborhood of every vertex is contained in some cluster $X\in \Xx$. The \emph{radius} of a cluster $X$ is the radius of $G[X]$. The cover $\Xx$ is said to have \emph{radius at most $d$} if all its clusters have radius at most $d$. The \emph{degree} of the cover is the maximum number of clusters that intersect at some vertex. 
\end{definition}

The existence of sparse neighborhood covers for nowhere dense classes was one of the key ingredients for the first-order model checking algorithm on these classes by Grohe et al.~\cite{grohe2017deciding}. 

\begin{theorem}[\cite{grohe2017deciding}]
    Let $\pi$ be an order of $V(G)$. Then $\Ww_{2r,\pi}^*$ is an $r$-neighborhood cover of radius at most $2r$ and degree at most $\wcol_{2r}(G,\pi)$. 
\end{theorem}

This result was vastly generalized by Dreier et al.~\cite{dreier2023first}, who proved that sparse 
neighborhood covers can be constructed from linear orders with small \emph{crossing number}. 

\begin{definition}
    Let $(U,\Ff)$ be a set system and $\pi$ be a linear order of $U$. 
    The \emph{crossing number} of $X\in \Ff$ is the  number of pairs $(u,v)$ of elements of $U$ such that $v$ is an immediate successor of $u$ in $\pi$ and exactly one of $u$ and $v$ belongs to $X$. 
\end{definition}


\begin{theorem}[\cite{dreier2023first}]
     Let $G$ be a graph and assume the set system of closed neighborhoods has crossing number at most $k$. Then $G$ admits a distance-$1$ neighborhood cover of radius at most $4$ and degree at most $k+1$. 
\end{theorem}

Welzl proved that small vc-density leads to linear orders with small crossing number~\cite{welzl1988partition}. 

\begin{theorem}[\cite{welzl1988partition}]
     Let $(U,\Ff)$ be a set system with vc-density at most $d$, where \mbox{$d>1$} is a real. 
     Then there exists a linear order of $U$ with crossing number bounded by $\Oof(|U|^{1-1/d}\log |U|)$. 
\end{theorem}

As the $r$th-power of every monadically stable class is again monadically stable one obtains the following corollary. 

\begin{theorem}[\cite{dreier2023first}]
    Let $\Cc$ be a monadically stable class of graphs, let $r$ be a positive integer and let $\epsilon>0$. Then $\Cc$ admits \mbox{$r$-neighborhood} covers with radius at most $4r$ and degree~$\Oof(n^\epsilon)$. 
\end{theorem}

Furthermore, the neighborhood covers of the theorem can be efficiently computed. 

Together with the flipper-game characterization of monadically stable classes this leads to the efficient model checking algorithm of Dreier et al.~\cite{dreier2023first,dreier2023firstb} on monadically stable classes of graphs. 

We conclude this section with one final result that Braunfeld et al.\ used to show that monadically stable classes have quasi-bounded bounded shrubdepth decompositions~\cite{braunfeld2024decomposition}. 

\begin{definition}
    An \emph{induced $r$-neighborhood cover} of a graph $G$ is a set $\Xx$ of 
    induced subgraphs of $G$, called \emph{clusters}, such that the $r$-neighborhood of every vertex is contained in some $X\in \Xx$. The radius of $\Xx$ is the 
    maximum radius of the connected components of the clusters of $\Xx$. The \emph{size} of $\Xx$ is the cardinality of $\Xx$.
\end{definition}

Note that the connected components of any induced $r$-neighborhood cover of radius $d$ and size $k$ yield an $r$-neighborhood cover with radius $d$ and degree $k$.

\begin{theorem}[\cite{braunfeld2024decomposition}]
    Let $\Cc$ be a monadically stable class of graphs, let $r$ be a positive integer and let $\epsilon>0$. Then $\Cc$ admits induced $r$-neighborhood covers with radius at most $4r+2$ and size $\Oof(n^\epsilon)$. 
\end{theorem}

\section{Further applications}\label{sec:further-applications}

We conclude our survey with a few selected applications of the generalized coloring numbers. 

\subsection{Exact distance colourings}\label{sec:exact-distance}

We begin with an elegant result of Van den Heuvel at al.~\cite{van2019chromatic} bounding the chromatic number of exact distance graphs by the weak coloring numbers.
Exact distance graphs were introduced by Ne\v{s}et\v{r}il and Ossona de Mendez in~\cite{nevsetril2012sparsity}. We refer to \cite{nevsetril2012sparsity,van2019chromatic} for more background on colorings of power graphs and exact distance graphs. 

\begin{definition}
For a graph $G$ and positive integer $p$, the \emph{exact
distance-$p$ graph} $G^{[\sharp p]}$ is the graph with vertex set $V(G)$ that has an edge between vertices $u$ and $v$ if and only if $u$ and $v$ have distance exactly $p$ in $G$. 
\end{definition}

\newcommand{\dcol}{\mathrm{dcol}}

For odd $p$, Ne\v{s}et\v{r}il and Ossona de Mendez proved that for a class $\Cc$ with bounded
expansion, the chromatic number of $G^{[\sharp p]}$ is bounded by an absolute constant for all $G\in \Cc$~\cite{nevsetril2012sparsity}. 
We provide the elegant proof of the following result (in fact, \cite{van2019chromatic} proves a stronger result in terms of a measure called $\dcol$ that we refrain from defining here). 

\begin{theorem}[\cite{van2019chromatic}]
For every odd positive integer $p$ and graph $G$ we have
$\chi(G^{[\sharp p]})\leq \wcol_{2p-1}^*(G)$. 
\end{theorem}
\begin{proof}
    Let $\pi$ be a linear ordering minimizing the chromatic number of $G\langle \pi, 2p-1\rangle$ and let $\rho$ be a corresponding coloring witnessing that $\wcol_{2p-1}^*(G)=:c$, hence any two vertices $u,v\in V(G)$ have different colors under $\rho$ whenever $u\in \WReach_{2p-1}[G,\pi,v]$ or $v\in \WReach_{2p-1}[G,\pi,u]$.
    For a vertex $v$ write $m(v)$ for the minimum vertex in the $\lfloor p/2\rfloor$-neighborhood of $v$ with respect to~$\pi$. 
    Define the coloring $\chi$ such that every vertex $v$ receives the color $\rho(m(v))$. 
    We claim that $\chi$ is a proper coloring of $G^{[\sharp p]}$. 

    Consider an edge $uv\in E(G^{[\sharp p]})$. 
    By definition, $u$ and $v$ have distance exactly $p$ in $G$ and the $\lfloor p/2\rfloor$-neighborhoods of $u$ and $v$ are disjoint. 
    We have $m(u)\neq m(v)$ and there exists a path of length at most $2p-1$ between $m(u)$ and $m(v)$. 
    Assume without loss of generality that $m(u)<_\pi m(v)$. Then $m(u)\in \WReach_{2p-1}[G,\pi, m(v)]$, hence, $u$ and $v$ receive different colors under $\chi$.  
\end{proof}



\subsection{Dimension of posets}

A partial order is a reflexive, antisymmetric, and transitive binary relation. 
A partially ordered set (poset) is a set equipped with a partial order $\leq$. 
The \emph{dimension} $\dim(P)$ of a poset $P$ is the least integer $d$ such that the elements of $P$ can be embedded into $\R^d$ such that $x<y$ in $P$ if and only if the point of $x$ is below the point of $y$ with respect to the product order of $\R^d$. 
The dimension is a key measure of a poset’s complexity.

A convenient way of representing a poset is to draw its Hasse diagram: 
Each element is drawn as a point in the plane such that if $x<y$, then $x$ is drawn below $y$. Then, for each relation $x<y$ that is not implied by transitivity (these are called the cover relations), we draw a $y$-monotone curve going from $x$ to $y$. 
The diagram defines a corresponding undirected graph, where edges correspond to pairs
of elements in a cover relation. 
This graph is called the cover graph of the poset. 
The height of a poset is the maximum size of a chain in the poset (a set of pairwise
comparable elements).

There is a large body of literature connecting poset dimension with graph structure. Streib and Trotter proved that for every fixed heigh $h>1$, posets of height $h$ with a planar cover graph have bounded dimension~\cite{streib2014dimension}. 
In general, Kelly proved that there is no function bounding the dimension posets by their height and a planar cover graph also does not imply a bound on the dimension of the poset~\cite{Kelly81}. 
In the sequel, Joret et al.\ showed that posets have dimension upper bounded by a function of their height if their cover graphs have bounded treewidth, bounded genus, or exclude an apex-graph 
as a minor~\cite{JoretMMTWW16}, exclude a minor~\cite{walczak2017minors} (Walczak) or topological minor~\cite{micek2017topological} (Micek and Wiechert), or come from a class with bounded expansion~\cite{joret2018sparsity} (Joret et~al.). 
The following nice result was proved by Joret et al.~\cite{joret2019nowhere}. 

\begin{theorem}[\cite{joret2019nowhere}]
    Let $P$ be a poset of height at most $h$ and let $G$ be its cover graph. Then $\dim(P)\leq 4^{\wcol_{3h-3}(G)}$. 
\end{theorem}

This line of research culminated in the following characterization of monotone graph classes. 

\begin{theorem}[\cite{joret2019nowhere}]
    Let $\Cc$ be a monotone class of graphs. Then $\Cc$ is nowhere dense if and only if for every integer $h>1$ and real $\epsilon>0$, $n$-element posets of height at most $h$ whose cover graphs belong to $\Cc$ have dimension $\Oof(n^\epsilon)$. 
\end{theorem}

\subsection{Directed graphs}

With the great success of treewidth and related measures for undirected graphs, researchers have attempted to define similar concepts for directed graphs, see e.g.~\cite{barat2006directed,berwanger2006dag,ganian2009digraph,
hunter2008digraph,obdrvzalek2006dag,safari2005d}. 
However, the algorithmic impact of directed width measures has by far not been as successful as for undirected graphs. 
For example, Ganian et al.~\cite{ganian2010there} call a width measure for digraphs \emph{algorithmically useful} if all MSO definable problems admit polynomial-time algorithms on digraphs of bounded width. 
However, Courcelle's 
Theorem~\cite{Courcelle90}, stating that every MSO
property can be decided in linear time on every class of
bounded treewidth, under standard complexity theoretic
assumptions cannot be extended to graph classes
of unbounded treewidth that additionally satisfy certain 
mild closure conditions~\cite{ganian2010there, kreutzer2010lower}. 
Nevertheless, these lower bounds still allow the search for 
classes of digraphs that are substantially different from 
classes whose underlying undirected graphs have bounded
treewidth or are nowhere dense, and that still admit 
efficient algorithms for individual problems. 

Kreutzer and Tazari~\cite{kreutzer2012directed} introduced generalizations of bounded expansion and nowhere dense classes of graphs to the directed setting. 
Excluding some tournament (a directed graph with exactly one edge between each two vertices) as a bounded depth minor does not turn out to be the right approach. 
\emph{Nowhere crownful} classes of digraphs are defined by excluding one-subdivided cliques (which interestingly are the forbidden patterns in monadically dependent classes), and turn out to have a nice combinatorial theory~\cite{kreutzer2012directed}.  
In particular, nowhere crownful classes can alternatively be characterized by directed uniform quasi-wideness, with immediate algorithmic applications. 
Classes with \emph{directed bounded expansion} were studied in~\cite{kreutzer2017structural,KreutzerMMRS19}. 
In particular, these classes can be characterized by directed generalized coloring numbers. 

\section{Concluding remarks and acknowledgments}

The generalized coloring numbers are one of the most influential notions in the theory of bounded expansion and nowhere dense graph classes. 
In this survey I have tried to give an overview over some of the most important results and applications. 
Of course the presented material reflects my own preferences and I hope that I have not upset anyone by missing their results. 

In the text I have included open questions about the generalized coloring numbers in the respective context. 
I conclude with some further explicit questions. Of course, there are many more, especially in the newly developing and very exciting area of graph product theory and in the study of dense but structurally sparse graph classes. 

\begin{enumerate}
    \item Dvo\v{r}\'ak et al.~\cite{DvorakPUY22}: 
    What is the asymptotic behavior of $\col_d$ of touching graphs of unit balls in $\R^d$? It is known to be in $\Oof(k^{d-1})$ and $\Omega(k^{d/2})$. 
    \item Give a direct proof for good quality weak coloring orders from the existence of bounded-size bounded treedepth decompositions, or from sparse neighborhood covers and winning the splitter game. 
    \item Find a good notion of a dense analog of the generalized coloring numbers. 
\end{enumerate}

\bigskip

This work is dedicated to Jarik Ne\v{s}et\v{r}il. Jarik, I want to express my heartfelt thanks for your inspiration and unwavering positivity, which I have had the privilege to experience since the early days of my PhD studies.
Your good mood and enthusiasm have always made it a pleasure to work with you and Patrice.
I hope that we will continue to explore the beauty of combinatorics for many more years to come!

Thank you Patrice for proofreading the manuscript, for many pointers to the literature, many comments on improving constructions, the beautiful figures, and the great time we always have when working together!

\bibliographystyle{abbrv}
\bibliography{ref}

\end{document}